\documentclass[11pt]{article}
\usepackage[margin=1in]{geometry}
\usepackage[utf8]{inputenc}
\usepackage{outlines}
\usepackage{amsmath}
\usepackage{amsthm}
\usepackage{amssymb}
\usepackage{amsfonts}
\usepackage{comment}
\usepackage[table,svgnames]{xcolor}
\usepackage{graphicx}
\usepackage{capt-of}
\usepackage{subcaption}
\usepackage{caption}
\usepackage{mathtools}
 
\usepackage{mathdots}
\usepackage{enumitem}
\usepackage{physics}
\usepackage{slashed}
\usepackage{braket}
\usepackage{nccmath}
\usepackage{float}
\usepackage{dcolumn}
\usepackage{tikz}
\usetikzlibrary{quantikz2}
\usetikzlibrary{chains,shapes}
\usetikzlibrary{calc, calligraphy}
\usetikzlibrary{backgrounds}
\usetikzlibrary{decorations.pathreplacing}

\usepackage{algpseudocode}
\usepackage{algorithm}
\usepackage{bm}
\usepackage{array}
\usepackage{color}
\usepackage{colortbl}
\usepackage{wrapfig}
\usepackage{microtype}
\usepackage{empheq}

\usepackage{hyperref}
\hypersetup{
    colorlinks=true,
    linkcolor=blue,
    filecolor=magenta,      
    urlcolor=cyan,
    pdftitle={Succinct Fermion Data Structures},
    pdfpagemode=FullScreen,
}
\usepackage{tablefootnote}
\usepackage{authblk}

\usepackage{amssymb} 
\usepackage{pifont} 
\definecolor{darkred}{rgb}{0.55, 0.0, 0.0}
\definecolor{darkspringgreen}{rgb}{0.09, 0.45, 0.27}

\newcommand{\bflip}{\mathsf{bit\text{-}flip}}
\newcommand{\srank}{\mathsf{sgn\text{-}rank}}
\newcommand{\bb}{\mathbf{b}}

\newcommand{\strings}[2]{ 
	\genfrac{\{}{\}}{0pt}{}{#1}{#2} 
}

\newcommand{\sstrings}[2]{ 
	\genfrac{[}{]}{0pt}{}{#1}{#2} 
}

\theoremstyle{plain}
\newtheorem{thm}{Theorem}[section]  
\newtheorem{corr}[thm]{Corollary}
\newtheorem{prop}[thm]{Proposition}
\newtheorem{rem}[thm]{Remark}

\newtheorem{defn}[thm]{Definition}

\newtheorem{lem}[thm]{Lemma}

\title{Succinct Fermion Data Structures}

\author[1]{Joseph Carolan\thanks{jcarolan@umd.edu}}
\author[1,2]{Luke Schaeffer\thanks{lschaeffer@uwaterloo.ca}}
\affil[1]{University of Maryland, College Park}
\affil[2]{University of Waterloo}

\begin{document}

\maketitle
 
\begin{abstract}
    Simulating fermionic systems on a quantum computer requires representing fermionic states using qubits. The complexity of many simulation algorithms depends on the complexity of implementing rotations generated by fermionic creation-annihilation operators, and the space depends on the number of qubits used. While standard fermion encodings like Jordan-Wigner are space optimal for arbitrary fermionic systems, physical symmetries like particle conservation can reduce the number of physical configurations, allowing improved space complexity. Such space saving is only feasible if the gate overhead is small, suggesting a (quantum) data structures problem, wherein one would like to minimize space used to represent a fermionic state, while still enabling efficient rotations. 
    
    We define a structure which naturally captures mappings from fermions to systems of qubits. We then instantiate it in two ways, giving rise to two new second-quantized fermion encodings of $F$ fermions in $M$ modes. An information theoretic minimum of $\mathcal{I}:=\lceil\log_2 \binom{M}{F}\rceil$ qubits is required for such systems, a bound we nearly match over the entire parameter regime. \begin{enumerate}[label=(\arabic*)]
        \item Our first construction uses $\mathcal I + o(\mathcal I)$ qubits when $F=o(M)$, and allows rotations generated by creation-annihilation operators in $O(\mathcal I)$ gates and $O(\log M \log \log M)$ depth.
        \item Our second construction uses $\mathcal I + O(1)$ qubits when $F=\Theta(M)$, and allows rotations generated by creation-annihilation operators in $O(\mathcal I^3)$ gates.
    \end{enumerate}
  In relation to comparable prior work, the first represents a polynomial improvement in both space and gate complexity (against Kirby et al. 2022), and the second represents an exponential improvement in gate complexity at the cost of only a constant number of additional qubits (against Harrison et al. or Shee et al. 2022), in the described parameter regimes.
\end{abstract}
\pagebreak
\tableofcontents
\pagebreak

\section{Introduction}

\begin{table}
    \centering
    {
    \rowcolors{3}{black!05}{white}
    \begin{tabular}{|c|c|c|c|}
        \hline
        Encoding & Qubits & Gates \\\hline
        \rowcolor{black!20}
        \multicolumn{3}{|c|}{Regime $F=o(M)$ and $F=\omega(1)$.} \\\hline
        Optimal Degree\cite{KFHM22} & $\Omega(\mathcal I^2 \log^2  M)$ & $\Omega(\mathcal I^2 \log^3 M)$  \\
        Qubit Tapering/Segment\cite{BGMT17, SW18, S19} & $M-o(M)$ & $\Omega(F^2)$  \\
        Permutation Basis\cite{HNAW22} & $\mathcal{I}$ & $\Omega(M^2 2^{\mathcal I})$ \\
        Theorem \ref{thm:fewQubitsLowDepth} & $\mathcal{I} + o(\mathcal I)$ & $O(\mathcal{I})$  \\\hline
        \rowcolor{black!20}
        \multicolumn{3}{|c|}{Regime $F=\Theta(M)$} \\\hline

        Qubit Tapering/Segment\cite{BGMT17, SW18, S19} & $M-O(1)$ & $\Omega(\mathcal I^2)$  \\
        Permutation Basis\cite{HNAW22} & $\mathcal{I}$ & $\Omega(\mathcal I^2 2^{\mathcal I})$ \\
        Theorem \ref{thm:implicitEncoding} & $\mathcal{I} + O(1)$ & $O(\mathcal{I}^3)$  \\\hline
        \rowcolor{black!20}
        \multicolumn{3}{|c|}{General regime} \\\hline
        Theorem \ref{thm:fewQubitsLowDepth} & $\mathcal{I} + O(F)$ & $O(\mathcal{I})$  \\
        Theorem \ref{thm:implicitEncoding} & {$\mathcal{I} + O(\log(M/F))$} & {$O(MF\mathcal I)$}  \\\hline
    \end{tabular}
    }
    \caption{Comparison to prior space efficient second-quantized encodings for $F$ fermion, $M$ mode systems, adapted from \cite{KFHM22}. Gates refers to the complexity of implementing rotations generated by creation-annihilation operators. Here $\mathcal I:=\left\lceil \log \binom{M}{F}\right\rceil$, which satisfies $\mathcal I = \Theta\left(F \log \left(\frac{M}{F}\right)\right)$.}
    \label{tab:fewFermionEncTable}
\end{table}

The simulation of many interacting fermions is a promising application of quantum computers. Such systems quickly become difficult to simulate classically, with accurate simulation being computationally intractable for complex systems \cite{troyer05intractable}. A quantum computer can perform highly accurate simulations with only polynomial resources for local systems, a wide and powerful class beyond what can be simulated classically.

Many quantum algorithms for fermionic simulation require a mapping from fermionic Fock states and creation/annihilation operators to qubit states and operators. Two of the most well known such mappings are the ones due to Jordan-Wigner \cite{JW28} and Bravyi-Kitaev \cite{BK02}. These require $M$ qubits to represent a fermion system with $M$ modes, which is prohibitive in some cases. In particular, for particle preserving systems where the number of fermions $F$ is much smaller than $M$, these encodings have a large amount of redundancy. In this case, the information theoretic lower bound of \begin{align}
    \mathcal{I}:=\left\lceil \log \binom{M}{F} \right\rceil \label{eqn:I-defn}
\end{align} qubits is potentially much smaller than $M$.

Encodings which take advantage of this fact are relevant whenever the number of fermions is small compared to the number of modes, for instance in quantum chemistry simulations where the number of electrons ($F$) is fixed by a physical system yet the orbital basis size ($M$) should be as large as possible to correctly resolve the continuum behaviour. The discretization error in such systems scales like $1/M$ for reasonable bases \cite{SBWRB21}, so large $M$ is required for high accuracy---this is especially important when using basis sets that have not been classically optimized, such as the plane wave basis. This motivates developing simulation algorithms with as mild a dependence on $M$ as possible in both space and gate usage, and a key component of many simulation algorithms is the qubit-to-fermion mapping used. Utilizing as few qubits as possible is important both for near term and fault tolerant devices, due to their limited size and the expense of qubits. However, space savings that incur an exponential gate overhead are not scalable. We therefore seek to minimize space complexity while still allowing relevant operations to be performed efficiently.

\subsection{Prior work}
\label{subsec:prior-work}

There has been a line of work on second-quantized fermion encodings for saving space when the number of fermions is fixed. Bravyi et al.\ give a scheme for utilizing symmetries in particle preserving Hamiltonians to remove degrees of freedom, using parity symmetries to achieve $M-O(1)$ qubits generically and LDPC codes for few fermion systems to achieve $M-O(M/F)$ qubits \cite{BGMT17}, but with gate overhead $O(M^3)$ in the worst case\footnote{Gate complexity here refers to the complexity of performing a single $\exp(-i\theta K)$, for some angle $\theta$ and number-preserving product of $O(1)$ many creation/annihilation operators $K$. For our discussions, we will not consider geometric locality.}. Steudtner and Wehner give a scheme based on segmenting the space of fermions which again achieves $M-O(M/F)$ qubits, though allows more efficient $O(F^2)$ gate overhead \cite{SW18,S19}. The same work also proposed a binary addressing code which shares some similarities with the encodings described here, but did not provide circuits for fermion operations nor bound the resource costs in the general case. The work of Babbush et al.\ \cite{BBS17} describes a way to store and manipulate few fermion systems using a sparse oracle, but do not provide explicit circuits for particle preserving rotations generated by creation-annihilation operators.

The work of Kirby et al.\ \cite{KFHM22} describes a second quantized fermion encoding achieving $O(F^2 \log^4 M)$ space complexity and $O(F^2 \log^5 M)$ gate complexity, subject to a certain conjecture about Hermite interpolation. Additionally, the work of Harrison et al.\ \cite{HNAW22} and Shee et al.\ \cite{STHCG22} give methods for achieving exactly $\mathcal{I}$ qubit usage in second quantization, but at the cost of an exponential $M^{O(F)}$ gate complexity.  

There are also first-quantized representations of fermion systems \cite{SBWRB21, KJLMA08, BBMN19}, where a list of occupied fermion modes is anti-symmetrized. This can save significant space when $F$ is much smaller than $M$. However, such encodings allow for a different class of operations\footnote{In particular, first-quantized operations. See \cite{SBWRB21} for examples of usage of these encodings in quantum simulation.} to be performed efficiently as compared to second quantized encodings, making circuit complexity results within these two paradigms not directly comparable. Other works consider utilizing more than $M$ qubits and restricting to an entangled subspace, allowing certain sets of $k$-local fermion operations to be extremely efficient (see e.g. \cite{VC05, WHT16, DKBC21, BK02}). These encodings are most relevant for geometrically local systems, and in fact increase space beyond the standard Jordan-Wigner encoding---we therefore omit them from comparisons.

\subsection{Contributions}
We describe a quantum data structure which naturally captures the problem of encoding (second quantized) fermion systems in qubits. This structure represents bitstrings of length $M$ and Hamming weight $F$, which naturally correspond to fermionic Fock states. We consider the (quantum circuit) gate/depth complexity of a \textsf{sgn-rank} (``sign rank'') operator, which applies a phase conditioned on how many ones exist in some prefix of the bitstring, as well as a \textsf{bit-flip} operator which flips one of the bits. These operators allow us to straightforwardly express the fermionic operations which are relevant to quantum simulation, as we show in Section \ref{sec:jw-data-structure}.

Within this framework, the Jordan-Wigner encoding can be seen as a data structure which encodes bitstring $b \in \{0,1\}^M$ by the quantum state $\ket{b}$, e.g. the trivial encoding. The sign rank operation is then simply a contiguous string of $Z$ operators on some prefix, and the bit flip operator is a single $X$ operator. Our encodings employ a different representation of bitstrings which is more efficient when $F$ is small, requiring different quantum circuits.

\paragraph{Succinct structure} We first give a second quantized encoding for number preserving fermion systems which is nearly optimal in space usage, yet allows fermionic rotations with linear complexity and low depth. Our encoding uses a sublinear amount of redundancy, i.e. it is within a factor $1+o(1)$ of optimal space usage (this is the meaning of ``succinct'' in the context of data structures) whenever $F=o(M)$. In particular, we prove the following theorem.

\begin{thm}[Informal version of Theorem \ref{thm:fewQubitsLowDepth}]
    A system of $F$ fermions in $M$ modes can be represented by $\mathcal{I}+O(F)$ qubits such that particle preserving fermionic rotations\footnote{By fermionic rotation, we mean the unitary $\exp(-i\theta K)$ for angle $\theta$, where $K$ is a number-preserving product of creation/annihilation operators, e.g. $K=a_j^\dagger a_k$ for some $j, k \in [M]$. We require that $K$ is $O(1)$ local in the sense that it is the product of $O(1)$ creation/annihilation operators (but not geometrically local).} have circuits of complexity $O(\mathcal{I})$ and depth $O(\log M \log \log M)$.
    \label{thm:informalMainEncoding}
\end{thm}

A summary of comparable prior work is shown in Table \ref{tab:fewFermionEncTable}; notably, all previous second-quantized encodings either incur a polynomial overhead in space complexity\footnote{i.e. have asymptotic space usage $\Omega(\mathcal{I}^c)$ for some $c > 1$ when $F \ll M$.}, or the gate complexity of fermionic rotations is exponential in the number of qubits. We build up to this main result by presenting intermediary encodings achieving some, but not all of the aforementioned scalings. These encodings may be of independent interest due to their simplicity.

{\paragraph{Implicit structure} We also give a construction which uses essentially the exact information-theoretic space minimum, except for an $O(1)$ number of ancilla (at constant filling). This construction achieves $\textsf{poly}(M)$ gate complexity. The relevant regime for this encoding is constant filling, e.g. $F=M/3$ or more generally $F=\Theta(M)$. In this regime, the prior space complexity of $\mathcal{I}+O(F)$ may be a significant overhead, whereas when $F=o(M)$ the additive $O(F)$ term is asymptotically negligible. The information theoretic limit in the $F=\Theta(M)$ regime goes like $\mathcal{I} =\Theta(F)$, meaning there can be significant space to save if $F$ is any constant fraction below $M/2$.}

{\begin{thm}[Informal version of Theorem \ref{thm:implicitEncoding}]
    A system of $F$ fermions in $M$ modes can be represented by $\mathcal{I}+O(\log (M/F))$ qubits such that particle preserving fermionic rotations have circuits of complexity $O(MF\mathcal I)$.
    \label{thm:informalImplicitEncoding}
\end{thm}}

{The most comparable prior work is that of Harrison et al. \cite{HNAW22} and Shee et al. \cite{STHCG22}, which both achieve exactly $\mathcal I$ qubit complexity. However, these require an $\Omega\left(\binom{M}{F} \textsf{poly}(M)\right)$ (e.g. exponential in $M$ at constant filling) circuit complexity overhead for performing rotations generated by creation-annihilation operators, in the worst case. Our encoding improves exponentially on this gate complexity, achieving $\textsf{poly}(M)$ in this regime. We do this at the cost of $O(1)$ ancilla qubits.}

\subsection{Technical overview}

We begin in Section \ref{sec:jw-data-structure} by developing a data structure which naturally captures fermion to qubit mappings of number preserving systems. This definition is not fully general, but provides a framework for thinking about such mappings. We then instantiate this data structure in two ways.

\paragraph{First result: a succinct structure.} To build up to the first main result, we begin by developing a fermion data structure in Section \ref{sec:efficientEnc} based on a straightforward idea. Instead of storing the explicit string of ones and zeros, store a sorted list of pointers to the $\leq F$ positions containing a one. This approach builds off ideas presented in prior fermion encodings \cite{BBS17, SW18} as well as work compressing quantum states \cite{buhrman22compress, gidney2022dictionaries}, though we present these ideas in a way which will facilitate understanding our full encoding. Observe that, in contrast to first quantized representations \cite{SBWRB21} (which also store a list of pointers to occupied positions), we will not antisymmetrize the state. This enables us to efficiently implement $\srank$ and $\bflip$ queries (analogous to second-quantized rather than first-quantized operations), which we provide in Section~\ref{subsec:simpleEfficientCircs}. In Section \ref{subsec:varLowDepth} we describe a procedure using buffer registers that reduces the depth in this construction to logarithmic.

From this simple starting point, our next fermion data structure in Section \ref{sec:fewerQubits} reduces the space requirements using combinatorial ideas, without affecting gate complexity. We accomplish this by splitting the most and least significant bits of each register, and using different representations for each. The most significant bits have a large amount of redundancy from being non-increasing, which we avoid by storing them using the stars and bars method. We give circuits for combining information between these two representations to efficiently implement the required operations.

In Section \ref{sec:fullEnc} we combine the two previous ideas with a succinct data structure that enables low-depth access to the most significant bits to show Theorem \ref{thm:informalMainEncoding}.

The key technical idea in this construction is a tree data structure used to store prefix sums related to the most significant bits. This prefix sum tree allows retrieving information about any specific entry in low depth, but requires less total space than a sorted list of numbers. This data structure therefore allows log depth fermionic rotations while being near-optimal in space usage, our first main result.

\paragraph{Second result: an implicit structure.} Using a different approach, in Section \ref{sec:implicit-struct} we describe our second data structure which achieves $\mathcal I +O(1)$ qubits and $\textsf{poly}(M)$ gate complexity when $F=\Theta(M)$. We represent a bit string $b$ (which itself represents a Fock state) by storing $b$'s rank among feasible strings, where we determine rank by the number of bit-strings which are lexicographically before $b$ or have a smaller Hamming weight than $b$. This encoding uses $\mathcal I+O(1)$ many qubits, but it is unclear how to interpret the encoded string to perform efficient operations. This obscurity is, essentially, the reason for the exponential gate overhead in prior works achieving this level of space efficiency.

We show how to perform both bit-flip and sign-rank efficiently on the first (unencoded) position, intuitively because the first and most-significant bit determines whether a given string has rank at most $\binom{M-1}{F-1}$ (first bit a $0$) or rank above $\binom{M-1}{F-1}$ (first bit a $1$) among strings of a given Hamming weight. This means that we can extract the first bit by comparing the encoded label against a fixed value, which can be done efficiently and with a small number of ancilla.

We then transform the ordering to one in which bit-flip and sign-rank can be performed efficiently on the second bit, and so on for each position. This is based on a construction of a certain set of orderings, which we denote $<_j$ for $j \in [0,...,M]$, of the labels. These orderings have the property that, given just the rank of a configuration over order $<_j$, operations on the $j$-th encoded bit can be performed efficiently ($O(F\mathcal I)$ gates) and with few $(O(1))$ ancillas on the label string. Further, we show how to transform a label under the order $<_j$ to a label under either $<_{j+1}$ or $<_{j-1}$ in the same complexity. By cycling over the $M$ possible orderings, we can perform relevant operations with $O(MF\mathcal I)$ gates and $O(1)$ ancillas.

\section{Preliminaries}
\label{sec:prelims}

\subsection{Fermions}

A system of fermions with $M$ modes can be defined by the algebra of creation ($a_i^\dagger$ for $i \in [M]$) and annihilation ($a_i$ for $i \in [M]$) operators---this formalism is referred to as second quantization. These operators are determined by the anti-commutation relations 
\begin{align}
    \{a_i^\dagger, a_j\} &= \delta_{ij}, & \{a_i^\dagger, a_j^\dagger\} &= 0, & \{a_i, a_j\} &= 0.\label{eqn:fermiAnticomm}
\end{align}
Let $\ket{00...0}_{f}$ denote the vacuum state, the unique state which is not annihilated by any of the $a_i^\dagger$. We will primarily work with the Fock states, which are of the form \begin{align}
    \ket{\psi^{\bb}}_{f} :=& (a_1^\dagger)^{b_1} (a_2^\dagger)^{b_2} ... (a_M^\dagger)^{b_M} \ket{00...0}_f & \text{(for all $\textbf{b} \in \{0,1\}^M$)} \nonumber \\
    =& \ket{b_1,b_2, ... ,b_M}_f.
\end{align}
We denote the span of these states $\mathcal{H}^{(fock)}$. Note that the $M$-mode Fock states are in natural correspondence with $M$-bit strings---this observation underlies the famous Jordan Wigner encoding \cite{JW28}, and will similarly underlie our data structures. It will be convenient for us to work in the Majorana basis determined by $\gamma_j$ for $j \in [2m]$. These operators are defined as follows: 
\begin{align}
    \gamma_{2j-1} &= \frac{a_j^\dagger + a_j}{2}, &
    \gamma_{2j} &= \frac{i(a_j^\dagger - a_j)}{2}, & \text{such that: }
    \{\gamma_j, \gamma_k\} = \delta_{jk}.
    \label{eqn:majoranaAnticomm}
\end{align}

The action of Majorana operators on the Fock states can now be written as follows, where $\neg b$ is the negation of bit $b$, \begin{align}
    \gamma_{2j-1} \ket{b_1 ... b_j ... b_M}_f =& \left(\prod_{n=0}^{j-1} (-1)^{b_n}\right)\ket{b_1 ... (\neg b_j) ... b_M}_f, \nonumber\\ \gamma_{2j} \ket{b_1 ... b_j ... b_M}_f =& \, i \cdot \left(\prod_{n=0}^{j} (-1)^{b_n}\right) \ket{b_1 ... (\neg b_j) ... b_M}_f. \label{eqn:majoranaAction}
\end{align}
In particular, these operators flip the occupation of a certain mode and apply a phase depending on the occupation of all preceding modes.

\subsection{Fermion to Qubit Mappings}
\label{subsec:fermToQubit}
A fermion to qubit mapping can be defined by a linear mapping $\mathcal{E}_s$ from Fock states to qubit states, as well as a mapping $\mathcal{E}_o$ from Majorana operators to qubit operators. The qubit states/operators should be isomorphic to Fock states and Majorana operators, i.e. satisfy Equations \ref{eqn:majoranaAnticomm}, \ref{eqn:majoranaAction}. Though some encodings do not fit into this paradigm (discussed in Section \ref{subsec:prior-work}), this definition suffices for the encodings relevant to this paper. In particular, to define a mapping of an $M$ mode system it suffices to construct $2M$ mutually anti-commuting qubit operators which each square to the identity.

We  note that it is not always necessary to encode every possible Fock state, e.g. when simulating a system that does not explore every state. Particularly relevant for us will be particle preserving systems, which live in the subspace of Fock states with exactly $F$ many fermions. 

A well known example of a fermion-to-qubit mapping is the Jordan-Wigner encoding \cite{JW28}, which defines $M$-qubit operators as follows (where $P_i$ denotes a Pauli $P$ on the $i$-th qubit, acting trivially on the rest) 
\begin{align}
    \mathcal{E}_{JW}(\gamma_{2i-1}) &= \left(\prod_{j=1}^{i-1}Z_j \right)\otimes X_i, &  
    \mathcal{E}_{JW}(\gamma_{2i}) &= \left(\prod_{j=1}^{i-1}Z_j\right) \otimes Y_i.
\end{align}
The mapped qubit operators in this encoding act on $M$ qubits, which means $M$ qubits are used to store fermionic states. Additionally, the circuit complexity of the mapped operators is $\Omega(M)$, as some mapped operators act non-trivially on all qubits.

\subsection{Simulating Fermions}
\label{subsec:simulating-fermions}
The main use case for fermion to qubit mappings is in the simulation of interacting fermions. In many such applications the dynamics are governed by a Hamiltonian $H$ which is both $k$-local and term-wise number preserving (i.e. each term commutes with $\sum_i a_i^\dagger a_i$). If we wish to simulate such a system with $M$ modes, the most general Hamiltonian is
\begin{align}
    H &= \sum_{l \leq k; l \text{ even};i_1i_2...i_l \in [M]} h_{i_1i_2...i_l} a_{i_1}^\dagger a_{i_2}...a_{i_{l-1}}^\dagger a_{i_l}+ h.c. \label{eqn:generalKlocalHam}
\end{align} 
where the indices run over the $M$ modes and $h.c.$ denotes the hermitian conjugate of the preceding term. This type of system is ubiquitous in quantum chemistry and physics. To simulate a fermionic system, one requires a mapping from fermion states/operators to qubit states/operators, as discussed in Section~\ref{subsec:fermToQubit}. Almost all second quantized algorithms for approximating evolution $\exp(-iHt)$ involve performing rotations generated by terms in $H$ \cite{C19, CW12, BCC+15, BCK15, LC17, BCS+20, LC19}, so it is better for these rotations to require few gates and act on few qubits. Let $V$ be a product of $k$ Majorana operators, where $k$ is even. Our encodings give efficient circuits for rotations of the form $\exp(-i\theta V)$ ($k$ a multiple of $4$) or $\exp(-\theta V)$ ($k$ not a multiple of $4$)\footnote{Actually, in most places we give circuits for the operator $V$ up to a global phase. In Appendix \ref{app:ptclePresRotations}, we discuss how to generically implement the aforementioned rotations in the same complexity. Further, whether $k$ is a multiple of $4$ dictates whether $\exp(-i\theta V)$ or $\exp(-\theta V)$ is unitary; we implement the unitary one.}. This is sufficient to implement any rotation generated by the product of $O(1)$ creation/annihilation operators and its hermitian conjugate. For an illustrative example, consider a hopping term $a_j^\dagger a_k+h.c.$ for $j\neq k$,
\begin{align*}
    \exp(-i\theta (a_j^\dagger a_k + a_k^\dagger a_j)) =& \exp(-i\theta((\gamma_{2j-1}-i\gamma_{2j})(\gamma_{2k-1}+i\gamma_{2k}) + (\gamma_{2k-1}-i\gamma_{2k})(\gamma_{2j-1}+i\gamma_{2j}))) \\
    =& \exp(2\theta(\gamma_{2j-1}\gamma_{2k} - \gamma_{2j}\gamma_{2k-1})) \\
    =& \exp(2\theta\gamma_{2j-1}\gamma_{2k} )\exp(-2\theta \gamma_{2j}\gamma_{2k-1}),
\end{align*}
from which it is clear that it suffices to implement two Majorana rotations. Note that although each Majorana operator itself need not preserve particle number, it changes the occupation number by at most $k$ (in this case at most two). Further, the whole operator does preserve particle number. Therefore, so long as our encoding has enough ``space'' to fit the intermediate states, we will remain in the proper subspace after implementing the full operator $\exp(-i\theta(a_j^\dagger a_k + a_k^\dagger a_j))$. We refer the reader to Appendix \ref{app:fermionOverflow} for a more general discussion of this procedure.

\section{Fermion Data Structures}
\label{sec:jw-data-structure}

In this section we introduce an alternative perspective of fermion encodings, in a way which makes the connection to data structures explicit. This perspective will facilitate the development of efficient fermion-to-qubit mappings later on.

\subsection{Definitions}
We first define two useful combinatorial operations, $\srank$ and $\bflip$. These are analogous to rank and bit flip queries respectively, which are well-studied in succinct classical data structures for bit-vectors \cite{jacobson88succinct, navarro16compact}\footnote{Note that the differing models of complexity (RAM programs/cell probe versus quantum circuits) prevent most of these classical results from being directly applicable to our problems.}. These operators act on a bit string $\bb=(b_1,...,b_M) \in \{0,1\}^M$ and depend on an index $j \in [M]$, and are combinatorial operations with no direct relation to fermions. However, they both show up naturally when writing down the action of fermionic creation-annihilation operators on Fock states: one can notice the similarities between Definitions \ref{defn:srank}, \ref{defn:bflip} and the action of Majorana operators in Equation \ref{eqn:majoranaAction}.

\begin{defn}
    The operator $\srank$ (``sign rank'') of an index $j \in [M]$ and bit string $\bb \in \{0,1\}^M$ is the operator which returns the parity of the first $j$ bits of $\bb$ (represented as $+1$ for even, $-1$ for odd). In particular, \begin{align*}
        \srank(j, \bb) := \prod_{n=0}^{j} (-1)^{b_n}.
    \end{align*}
    \label{defn:srank}
\end{defn}

\begin{defn}
    The $\bflip$ operator of an index $j \in [M]$ and bit string $\bb \in \{0,1\}^M$ flips the $j$-th bit of $\bb$. In particular, we have the following, 
    where $\neg$ is logical negation, \begin{align*}
        \bflip(j, \bb) :=& (b_1, ... ,(\neg b_j), ... ,b_M).
    \end{align*}
    \label{defn:bflip}
\end{defn}

With these two operations in hand, we can now describe the relevant type of (quantum) data structure for encoding fermionic states. Intuitively, we would like to represent a bit string $\bb \in \{0,1\}^M$ in such a way that both $\srank$ and $\bflip$ queries can be implemented efficiently. Furthermore, in the settings we will care about we will be promised that the Hamming weight of $\bb$ (the number of $1$'s) will never exceed $F+k$ or subceed $F-k$, for some $F < M$ and constant $k=O(1)$ (see Appendix \ref{app:fermionOverflow} for more discussion). We use the notation $\mathcal{H}^{(fock)}_{M, F,k}$ to denote the Fock space of an $M$-mode system of fermions, restricted to states with occupation between $F-k$ and $F+k$: we call this range the capacity. We use the notation $\mathcal{H}^{(qubit)}_n$ to denote the Hilbert space of $n$ qubits.

\begin{defn}
    A \emph{fermion data structure} of size $M$ and capacity $F, k$ is a qubit representation of Fock states with $M$ modes and occupation at most $F+k$ and at least $F-k$, i.e. a linear and invertible mapping $\mathcal{E}_s : \mathcal{H}^{(fock)}_{M, F,k} \rightarrow \mathcal{H}^{(qubit)}_n$ that maps Fock states to computational basis states. We denote $\mathcal{E}_s(\ket{\bb}_f)$ as $\ket{\bb}$. We further require two $n$-qubit quantum circuit families $F_j$ and $R_j$, answering $\bflip$ and $\srank$ queries respectively. Formally, we require \begin{align}
        F_j \ket{\bb} :=& \ket{\bflip(j, \bb)} &\text{(if $\bflip(j, \bb)$ has $F-k \leq \mathsf{HW} \leq F+k$)}, \label{eqn:flip-query-defn} \\
        R_j \ket{\bb} :=& \, \srank(j, \bb)\ket{\bb}. \label{eqn:rank-query-defn}
        \end{align}
    \label{defn:fermion-data-structure}
\end{defn}
The gate/depth complexity of the above structure is the max gate count/depth, respectively, of the $F_j, R_j$ over all $j \in [M]$. The space complexity is $n$---note that ancillas used in $F_j$, $R_j$ are counted in this complexity, as they are circuits acting on exactly $n$ qubits.

We remark that we consider quantum circuit complexity as the relevant cost measure, which prevents us from using many standard data structures results. Classical results often use the cell probe model \cite{gal2007cell-probe}, which is based on a RAM machine that allows random access to data. This model may not capture complexity in a real world quantum computer \cite{jaques2023qramsurveycritique}. This necessitates new data structures and new ideas, which will be the focus of this paper.

\subsection{Properties of fermion data structures}
We first observe a subtlety in Definition \ref{defn:fermion-data-structure}: the operation of $\bflip$ is only required to be ``correct'' so long as the capacity is not exceeded. When one is interested in simulating a $k$-local Hamiltonian on a system with exactly $F$ Fermions, then it often suffices to instantiate a Fermion data structure with capacity $F,k$\footnote{Note that this does incur some space overhead, but it is only an additive $O(1)$ (if $F=\Theta(M)$) or $O(\log M)$ (if $F=O(M / \log M)$)}. In particular, it is straightforward to show that such a data structure can be used to instantiate a Trotterization scheme, or any simulation algorithm which relies on applying $k$-local particle-preserving rotations generated by products of $k$ creation/annihilation operators. Intuitively, so long as we can ``fit'' the intermediate states necessary while performing these rotations, we will be able to perform any arbitrary sequence of rotations.

\begin{rem}
    Let $\ket{\psi}_f$ be a fermionic Fock state of $F$ fermions in $M$ modes, and $C_f$ be a sequence of $k$-local fermionic rotations which preserve particle number, followed by a measurement in the Fock basis. If there is a fermion data structure of size $M$ and capacity $F-k, ..., F+k$, using space $n$ and gate complexity $m$, then there is a qubit state $\ket{\psi}$ on $n$ qubits and a quantum circuit $C$ of size $O(|C_f| \cdot m)$ such that measuring $C \ket{\psi}$ in the computational basis samples from the same distribution as $C_f \ket{\psi}_f$ (up to permuting bit-string labels).
\end{rem}

\begin{proof}
    The initial state $\ket{\psi}$ is simply the encoding $\mathcal{E}_s(\ket{\psi}_f)$. To construct the circuit $C$, we will replace every gate of $C_f$ (which are of the form $g_f=\exp(i\theta V)$ for $k$-local, particle preserving $V$) with the corresponding encoded circuit for performing the rotation, as detailed in Section \ref{subsec:simulating-fermions}. Let us call this encoded circuit $g$. Letting $\ket{\phi}$ be some superposition of Fock states with $F$ fermions, then by Appendix \ref{app:fermionOverflow}, with a data structure $\mathcal{E}_s$ of capacity $F-k,...,F+k$ we have the guarantee that \begin{align*}
        g \mathcal{E}_s(\ket{\phi}) =& \mathcal{E}_s(g_f \ket{\phi}),
    \end{align*}
    as $g_f$ corresponds to a sequence of $\srank$ and $\bflip$ operations satisfying the necessary promises. Noting that $g_f \ket{\phi}$ will also have $F$ fermions by the hypothesis that $g_f$ preserves particle number, we can similarly perform this transformation for every following gate. Finally, by the invertibility of $\mathcal{E}_s$ we will obtain the same final output distribution by measuring in the computational basis, up to decoding the bit-strings (i.e. some permutation on the labels).
\end{proof}

Given the above discussion, the curious reader may note that our definition of fermionic data structures encodes all states of Hamming weight between $F-k$ and$F+k$ for some constant $k$. This is redundant in that $F+k>F$, and also that there are more states having Hamming weight between $F-k$ and $F+k$ than there are of Hamming weight exactly equal to $F+k$. However, as we discuss in Appendix \ref{app:fermionOverflow}, both of these redundancies result in negligible additive space overheads (either $O(1)$ or $O(\log M)$ depending on the regime). The conclusions we draw about space efficiency includes these overheads.

We remark that our encoding can also apply to some post-Trotter methods, e.g. certain linear combination of unitaries algorithms, as we discuss in \ref{app:coherentlyControlled}. We also remark for clarity that the encodings in Sections \ref{sec:efficientEnc}, \ref{sec:fewerQubits}, \ref{sec:fullEnc} will allow encoding every Fock state of occupation at most $F+k$ (though Section \ref{sec:implicit-struct} will only encode occupations $F-k$ through $F+k$).

\section{Sorted List Fermion Data Structure}
\label{sec:efficientEnc}

In this section, we describe a fermion data structure (Definition \ref{defn:fermion-data-structure}) based on storing a sorted list of pointers, rather than a full sparse string. While this encoding has some similarities to first-quantized encodings \cite{SBWRB21} and sparse oracle encodings \cite{BBS17}, we emphasize that we present a generic second-quantized encoding capable of performing any $O(1)$-local particle preserving rotation generated by creation/annihilation operators in the described complexity. Furthermore, this encoding will lay the groundwork for our later, more efficient encodings.

\begin{thm}
    There is a fermion data structure for strings of $M$ qubits having Hamming weight at most $F$ that uses $O(F \log M)$ qubits, with gate complexity $O(F \log M)$.
    \label{thm:mainEncodingThm}
\end{thm}

We will constructively establish this theorem throughout Section \ref{subsec:algDef} and \ref{subsec:simpleEfficientCircs}---in particular Lemma \ref{lem:simpleZZZComplexity} and Lemma \ref{lem:simpleXpComplexity} demonstrate the requisite circuit complexities. An immediate corollary is that a similar statement holds for fermion encodings.
\begin{corr}
    There exists a fermion encoding of an $M$ mode system having at most $F$ fermions which uses $O(F \log M)$ qubits such that any $k$-local, particle preserving rotation can be implemented with $O(F \log M)$ gates.

    \label{corr:fermiEncFromQubits}
\end{corr}

\subsection{Algebraic Definition}
\label{subsec:algDef}
We formally define the state encoding function $\mathcal{E}_s^{(1)}$ and give the algebraic properties it satisfies. Let $e_i \in \{0,1\}^M$ be a binary string that is all $0$'s, except for a single $1$ at position $i$ (using $1$-based indexing). Consider strings of the form \begin{align}
    x &= e_{i_1} \oplus e_{i_2} \oplus ... \oplus e_{i_f} \nonumber\\
    i_1&< i_2 < ... < i_f \nonumber
\end{align}
with $f\leq F$. Define  $\mathcal{E}^{(1)} : \mathcal{D} \rightarrow \{0,1\}^{F \lceil \log (M + 1) \rceil}$, where $\mathcal{D} \subset \{0, 1\}^M$ is the set of strings of Hamming weight at most $F$, as 
\begin{align}
    \mathcal{E}^{(1)}(x) &= \ket{i_1} \ket{i_2} \cdots \ket{i_f} \ket{\infty} \cdots \ket{\infty} \label{eqn:stateEncDefn}
\end{align}
where $\infty$ is a placeholder for ``no fermion'', and is the larger of any comparison with a non-$\infty$. Concretely, we adopt the convention that $\infty$ is represented by a register of all $1$'s. We can interpret the output as a sorted array of $f$ elements, each element (also referred to as register) of size $\lceil \log (M + 1) \rceil$, and padded out to $F$ entries by $\infty$'s. We will now describe the action of sign-rank and bit-flip queries on this list.
\begin{enumerate}[label=(\arabic*)]
    \item A $\srank(j, \bb)$ query should apply a $-1$ phase if there are an odd number of registers having values less than or equal to $j$. Otherwise, it should act like the identity.
    \item A $\bflip(j, \bb)$ should insert a register $\ket{j}$ into the list if it is not present, and otherwise delete $\ket{j}$. This operation should maintain sorted order, and act as described so long as both input and output states have Hamming weight less or equal to $F$.
\end{enumerate}
Observe that both of the aforementioned operations are reversible. This suffices to define a fermion data structure.

\subsection{Efficient Circuits} 
\label{subsec:simpleEfficientCircs}

We will utilize comparison circuits between unsigned integers as a black box, see Appendix~\ref{app:compCircs} for a more detailed discussion. As described in Section \ref{subsec:algDef}, we will assume an upper limit $F$ on the number of pointers we will ever need to store, and therefore only utilize this many registers.
\begin{lem}
    \label{lem:simpleZZZComplexity}
    For the sorted list encoding described in Section \ref{subsec:algDef}, a $\srank(p, \bb)$ query has a circuit of $O(F \log M)$ gates and $O(F \log M)$ ancilla.
\end{lem}
\begin{proof}
    Such a circuit should induce a $-1$ phase for every $1$ present (in the original bit string) before position $p$. To achieve this, for each pointer register $\ket{x}$ we compute $x \leq p$, apply a $Z$ to the outcome bit, then uncompute. This is depicted in Figure \ref{fig:VpSubroutineForZZZ}. Each comparison takes $O(\log M)$ gates to compare $O(\log M)$ size numbers, and there are $O(F)$ comparisons to make; the overall circuit complexity is therefore $O(F \log M)$.
\end{proof}

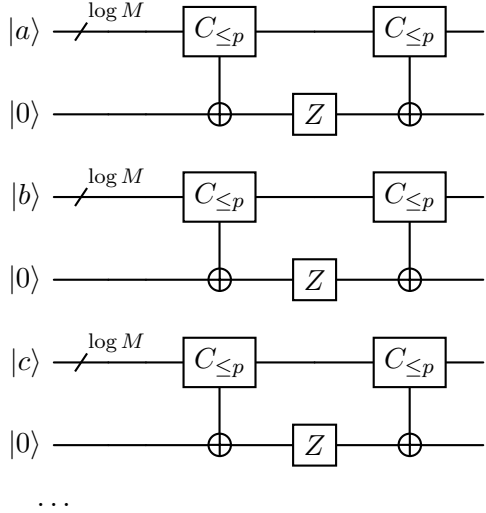
\begin{figure}[H]
    \begin{center}
        \begin{quantikz}[wire types={q,q}]
            \lstick{$\ket{a}$} & \qwbundle{\log M} & & \gate{C_{\leq p}} \wire[d][1]{q} & & \gate{C_{\leq p}} \wire[d][1]{q} & \\
            \lstick{$\ket{0}$} & & &\targ{} & \gate{Z} & \targ{} &  \\
            \lstick{$\ket{b}$} & \qwbundle{\log M} & & \gate{C_{\leq p}} \wire[d][1]{q} & & \gate{C_{\leq p}} \wire[d][1]{q} & \\
            \lstick{$\ket{0}$} & & &\targ{} & \gate{Z} & \targ{} &  \\
            \lstick{$\ket{c}$} & \qwbundle{\log M} & & \gate{C_{\leq p}} \wire[d][1]{q} & & \gate{C_{\leq p}} \wire[d][1]{q} & \\
            \lstick{$\ket{0}$} & & &\targ{} & \gate{Z} & \targ{} &  \\
            \ldots
        \end{quantikz}
    \end{center}
    \caption{Flips the phase of any $\ket{x}$ if $x \leq p$. $\srank$ is achieved by applying such a phase in parallel to all registers.}
    \label{fig:VpSubroutineForZZZ}
\end{figure}

\begin{lem}
    \label{lem:simpleXpComplexity}
     
    For the sorted list encoding described in Section \ref{subsec:algDef}, a $\bflip(p, \bb)$ query has a circuit of $O(F \log M)$ gates and $O(\log M)$ ancilla.
\end{lem}
\begin{proof}
    Such a circuit should add a pointer $p$ to the list if there is none, otherwise it should remove the pointer $p$. At a high level, we will do this in three steps. \begin{enumerate}[label=(\arabic*)]
        \item If $p$ exists, move it to the last register
        \item On the last register exchange $p$ and $\infty$
        \item If the last register is $p$, move it to the sorted position
    \end{enumerate}
    This is depicted as a circuit in Figure \ref{fig:XpFullCircuit}. To accomplish this reversibly, we first implement a reversible ordered-swap $U_p$. This operator swaps two registers if one is $p$ and the other is $>p$, pictured in Figure \ref{fig:UpSubroutineForBubble}. By chaining these together in ascending order, we will move $p$ to the end if it exists. Exchanging classical states can be done by the circuit in Figure \ref{fig:exchangeClassicalsCircuit}, then chaining more $U_p$ together in descending order will move $p$ to sorted order if it was at the end. The full circuit requires $O(F)$ calls to subroutine $U_p$, and each takes $O(\log M)$ (from comparisons); exchanging two classical values on a register has negligible complexity. The overall complexity is therefore $O(F \log M)$.
\end{proof}

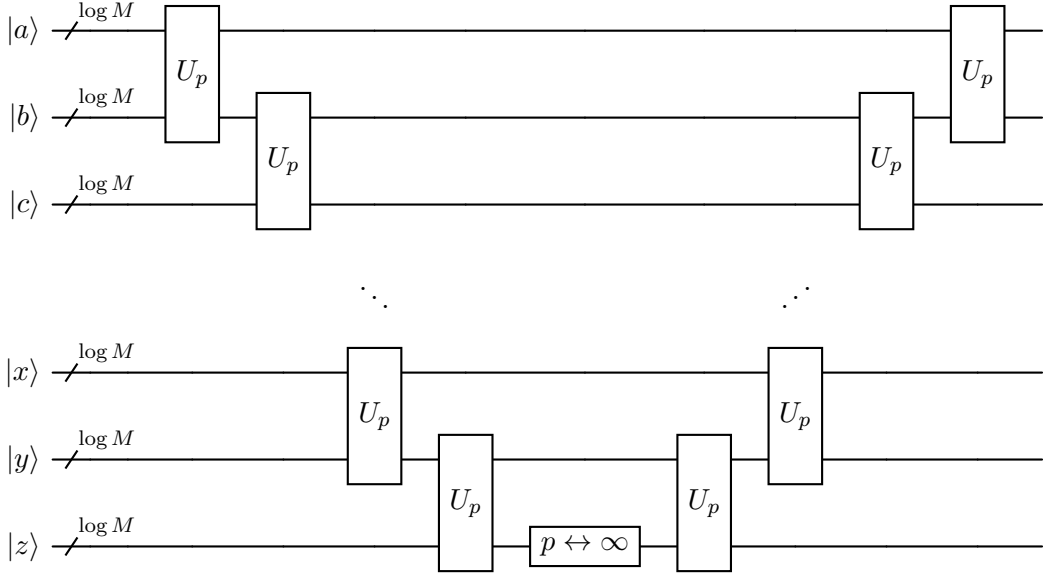
\begin{figure}[H]
    \begin{center}
        \begin{quantikz}[wire types={q,q,q,n,q,q,q}, classical gap=0.7mm]
            \lstick{$\ket{a}$} & \qwbundle{\log M} & & \gate[2]{U_p} & & & & & & & & \gate[2]{U_p} & \\
            \lstick{$\ket{b}$} & \qwbundle{\log M} & & & \gate[2]{U_p}  & & & & & & \gate[2]{U_p} & & \\
            \lstick{$\ket{c}$} & \qwbundle{\log M} & & & & & & & & & & & \\
            & & & & & \ddots & & & & \iddots & & & \\
            \lstick{$\ket{x}$} & \qwbundle{\log M} & & & & \gate[2]{U_p} & & & & \gate[2]{U_p} & & & \\
            \lstick{$\ket{y}$} & \qwbundle{\log M} & & & & & \gate[2]{U_p} & & \gate[2]{U_p} & & & &  \\
            \lstick{$\ket{z}$} & \qwbundle{\log M} & & & & & & \gate{p \leftrightarrow \infty} & & & & & \\
        \end{quantikz}
    \end{center}
    \caption{$X_p$ is achieved by pushing $p$ down with $U_p$ gates and exchanging it out with $\infty$, or exchanging $p$ in and pushing it up with $U_p$ gates. See Appendix \ref{app:circSubroutines} for a definition of circuit subroutines $U_p$ and $p \leftrightarrow \infty$.}
    \label{fig:XpFullCircuit}
\end{figure}

\subsection{Lower Depth}
\label{subsec:varLowDepth}
A notable problem with the previously described encoding is the serial nature of the bit flip operator. In particular, the depth of the circuits as described is $O(F \log M)$, primarily due to sequentially ``bubbling'' a targeted register to the end of the list when implementing a $\bflip(p, \bb)$ operation. In this section, we describe a modification using $O(F \log M)$ ancillas that allows exponentially smaller depth, $O(\log F + \log \log M)$. Further, the buffer register idea developed here will prove useful when lowering depth in our succinct construction.

To achieve $O(\log F + \log \log M)$ depth, we will pad the encoded state with an $\ket{\infty}$ between registers and on either end. We will refer to these as buffer registers, in contrast with logical registers representing pointers. We will also pad the start with a $\ket{-\infty}$ and the end with $\ket{\infty}$ logical registers for convenience. Formally, consider strings of the form \begin{align}
    x &= e_{i^{(1)}} \oplus e_{i^{(2)}} \oplus ... \oplus e_{i^{(f)}} \nonumber\\
    i^{(1)}&< i^{(2)} < ... < i^{(f)} \nonumber
\end{align}
where $f<F$. Define  $\mathcal{E}^{(2)} : \{0, 1\}^M \rightarrow \{0,1\}^{O(F \log M)}$ as
\begin{align}
    \mathcal{E}^{(2)}(x) =& \ket{-\infty} \textcolor{blue}{\ket{\infty}}\ket{i^{(1)}_l}\textcolor{blue}{\ket{\infty}}\ket{i^{(2)}_l}
    \textcolor{blue}{\ket{\infty}}\cdots \textcolor{blue}{\ket{\infty}}\ket{i^{(F)}_l} \textcolor{blue}{\ket{\infty}} \cdots \ket{\infty},
\end{align}
where the blue registers are buffers (this coloring is purely conceptual; the same data is stored in each). This leaves $O(F \log M)$ qubits, as the buffer registers are just a constant factor overhead. The correct action of $\srank$ and $\bflip$ exactly mirrors that of Section \ref{sec:efficientEnc} on the logical registers (excluding the initial padded $-\infty$), except now we must also leave the buffer registers invariant after the computation is finished.

\begin{lem}
    Under the list-with-buffers encoding described in Section \ref{subsec:varLowDepth}, a $\srank(p, \bb)$ query has a circuit of $O(\log \log M)$ depth, $O(F \log M)$ gates, and $O(F \log M)$ ancilla.
    \label{lem:lowDepthZZZ}
\end{lem}
\begin{proof}
    This operation can be done in a similar way to Section \ref{subsec:simpleEfficientCircs}. We again apply a conditional phase to each logical register, as in Figure \ref{fig:VpSubroutineForZZZ} (note we do not do this on any buffer register, nor the padded $\ket{-\infty}$ at the beginning). The gate count $O(F \log M)$ follows from the analysis in Lemma \ref{lem:simpleZZZComplexity}, and one can see that the depth is fully determined by the depth of a comparison. By Lemma \ref{lem:equalityComplexity} and Lemma \ref{lem:comparisonComplexity}, these can be done in $O(\log \log M)$ depth using linearly many ancillae, so the overall depth is $O(\log \log M)$.
\end{proof}

\begin{lem}
    Under the list-with-buffers encoding described in Section \ref{subsec:varLowDepth}, a $\bflip(p, \bb)$ query has a circuit of $O(\log F + \log \log M)$ depth, $O(F \log M)$ gates, and $O(F \log M)$ ancilla.
    \label{lem:lowDepthXp}
\end{lem}

\begin{proof}
    We are promised that whenever an insertion is made, the last element will be an $\infty$. Conversely, whenever a deletion is made, the element added to the end will be an $\infty$. This prevents us from needing to perform a full cycle, instead relying on the buffer registers to provide $\infty$'s wherever needed. At a high level, it suffices to do the following:
    \begin{enumerate}[label=(\arabic*)]
        \item Compute a flag $f_{del}$ depending on whether $p$ appears in the logical registers, and fan it out to all registers (in Figure \ref{fig:lowDepthXp} this fan-out is implicit).
        \item If $p$ is not present ($f_{del}=\top$), insert $p$ at the buffer register between it's immediate predecessor and successor.
        \item If $p$ is present, shuffle registers $>p$ upwards one position, and move $p$ to it's preceding buffer. Otherwise, shuffle all registers $>p$ downwards one position, and move $p$ out of it's preceding buffer.
        \item If $p$ was present ($f_{del}=\bot$), remove $p$ from the buffer between $p$'s predecessor and successor.
        \item Uncompute $f_{del}$.
    \end{enumerate}
    This is depicted in Figure \ref{fig:lowDepthXp}. Note that steps (2) and (4) do not need to explicitly be controlled by $f_{del}$ due to the assumed structure of the data. In step (3) however, the direction in which to shuffle depends on whether $p$ is present or not, so a light cone argument implies that the local operations must depend on whether $p$ is present. It is worth pointing out that, except for the fan-in/fan-out of $f_{del}$, all operations are nearest-register in a one dimensional layout, which could aid implementation on a quantum computer with geometric constraints.

    Using linearly many ancilla registers, a comparison between $n$ bit integers can be done in depth $\log n$ from Lemma \ref{lem:equalityComplexity} and Lemma \ref{lem:comparisonComplexity}. Swaps between registers can be done in $O(1)$ depth, the largest comparison is between $O(\log M)$ bit integers, and a single bit can be fanned out to $F$ positions (one for each register) in depth $O(\log F)$. The total depth is therefore $O(\log F + \log \log M)$ (or $O(\log \log M)$ with unbounded fan-in and fan-out), using $O(F \log M)$ ancillae. Each register of size $\log M$ is acted on by only a constant number of comparison/swaps, so the gate complexity is $O(F \log M)$.
\end{proof}

\begin{figure}[H]
    \centering
    \begin{quantikz}[wire types={q,q,q,q,q,q,q,q,q,q,q,q,n}, classical gap=0.7mm]
        \lstick{$\ket{0}$} & \qw & \qw & \targ{}\slice[style=black]{(1): Compute flag} & \midstick{$\ket{f_{del}}$} & \qw \slice[style=black]{(2): Insert $p$}& \octrl{9} & \qw & \ctrl{9} \slice[style=black]{(3): Shift suffix $>p$}& \qw & \qw & \qw\slice[style=black]{(4): Delete $p$} & \targ{} & \gate{X} & \qw \\
        \lstick{$\ket{-\infty}$} & \qwbundle{\log M} & \qw & \qw & \qw & \gate[3]{E_p} & \qw & \qw & \qw & \qw & \qw & \gate[3]{E_p} & \qw & \qw & \qw \\
        \lstick{$\ket{\infty}$} & \qwbundle{\log M} & \qw & \qw & \qw & \qw & \qw & \gate[2]{S_p} & \qw & \qw & \qw & \qw & \qw & \qw & \qw \\
        \lstick{$\ket{x}$} & \qwbundle{\log M} & \qw & \gate{C_{=p}}\wire[u][3]{q} &\gate[3]{E_p} & \qw & \gate[2]{S_p} & \qw & \gate[2]{S_p} & \qw & \gate[3]{E_p} & \qw & \gate{C_{=p}}\wire[u][3]{q}& \qw & \qw\\
        \lstick{$\ket{\infty}$} & \qwbundle{\log M} & \qw & \qw & \qw & \qw & \qw & \gate[2]{S_p} & \qw & \qw & \qw & \qw & \qw & \qw & \qw \\
        \lstick{$\ket{y}$} & \qwbundle{\log M} & \qw & \gate{C_{=p}}\wire[u][2]{q} & \qw & \gate[3]{E_p} & \gate[2]{S_p} & \qw & \gate[2]{S_p} & \qw & \qw & \gate[3]{E_p} & \gate{C_{=p}}\wire[u][2]{q} & \qw & \qw \\
        \lstick{$\ket{\infty}$} & \qwbundle{\log M} & \qw & \qw & \qw & \qw & \qw & \gate[2]{S_p} & \qw & \qw & \qw & \qw & \qw & \qw & \qw\\
        \lstick{$\ket{z}$} & \qwbundle{\log M} & \qw & \gate{C_{=p}}\wire[u][2]{q} &\gate[3]{E_p} & \qw & \gate[2]{S_p} & \qw & \gate[2]{S_p} & \qw & \gate[3]{E_p} & \qw & \gate{C_{=p}}\wire[u][2]{q}& \qw & \qw \\
        \lstick{$\ket{\infty}$} & \qwbundle{\log M} & \qw & \qw & \qw & \qw & \qw & \gate[2]{S_p} & \qw & \qw & \qw & \qw & \qw & \qw & \qw \\
        \lstick{$\ket{w}$} & \qwbundle{\log M} & \qw & \gate{C_{=p}}\wire[u][2]{q} & \qw &\gate[3]{E_p} & \gate[2]{S_p} & \qw & \gate[2]{S_p} & \qw & \qw & \gate[3]{E_p} & \gate{C_{=p}}\wire[u][2]{q} & \qw & \qw\\
        \lstick{$\ket{\infty}$} & \qwbundle{\log M} & \qw & \qw & \qw & \qw & \qw & \gate[2]{S_p} & \qw & \qw & \qw & \qw & \qw & \qw & \qw \\
        \lstick{\ldots} & & & & & & & & & & & & & &
    \end{quantikz}
    \caption{The full low-depth circuit for an $\bflip(p, \bb)$ operation. Large width operations that target/are controlled by a single qubit can be implemented with additional ancillae and $O(1)$ many fan-in/fan-out gates, or in logarithmic depth by $1$ and $2$ qubit gates. Registers $x,y,z,w$ are the logical registers, with adjacent registers being buffers. Note also the list is padded by a $-\infty$ logical register at the beginning and an $\infty$ at the end, though these are only used for comparisons and so could be removed and replaced by a hard-coded comparison. See Appendix \ref{app:circSubroutines} for circuit subroutines.}
    \label{fig:lowDepthXp}
\end{figure}
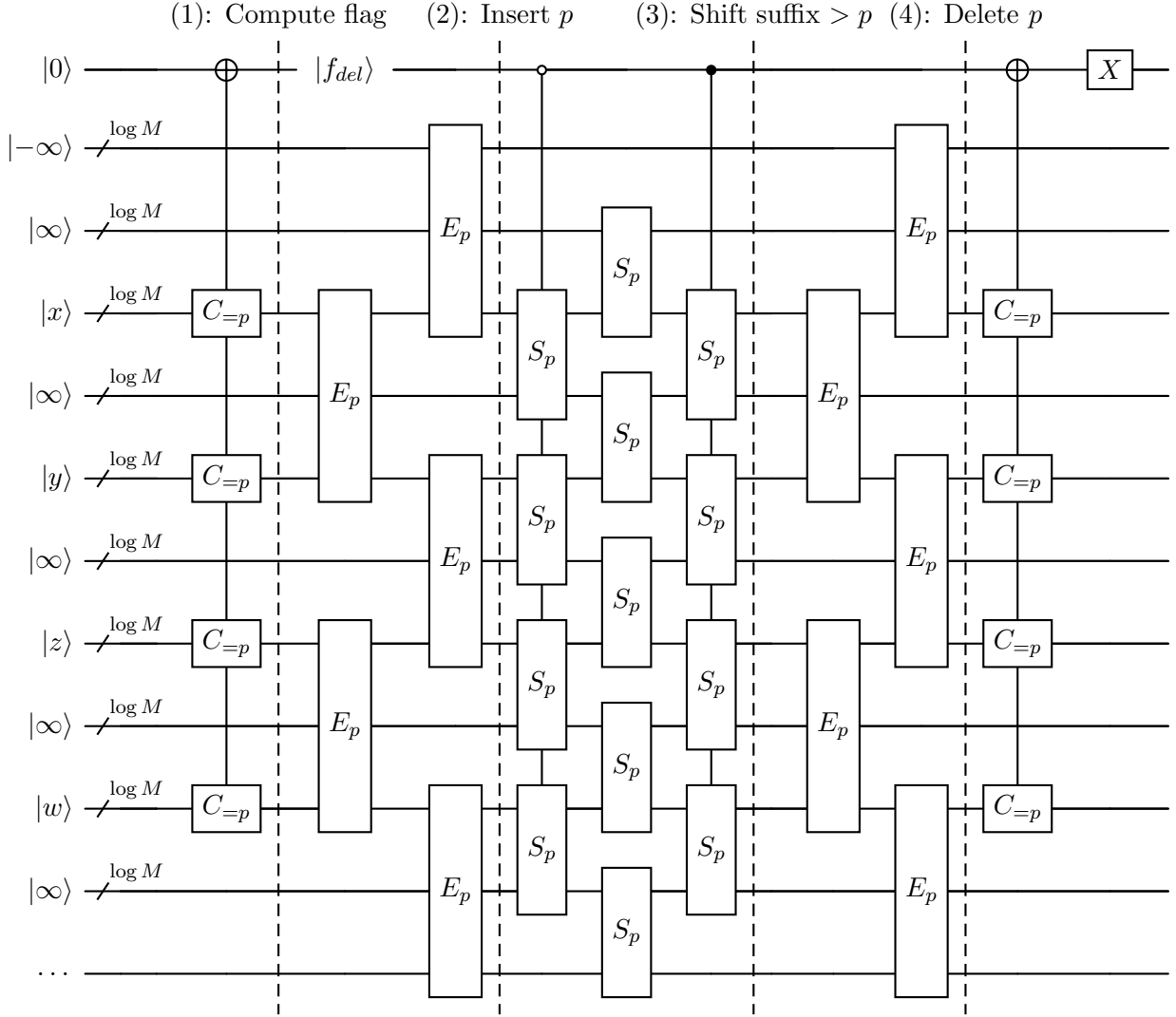

\section{Achieving Succinctness}
\label{sec:fewerQubits}

Recall that the information theoretic limit for encoding $F$ fermions in $M$ modes is $\mathcal{I}:= \lceil \log \binom{M}{F} \rceil$ qubits. From Stirling's approximation we have \begin{align*}
    \mathcal{I} = F \log \frac{M}{F} + O(F),
\end{align*} so the encoding in Section \ref{subsec:algDef} is redundant---for instance when $F=\sqrt{M}$, approximately half of the qubits are unnecessary. Here we present an encoding in which only a sublinear $O(F)$ many qubits are unnecessary (there are known constructions that achieve the \emph{exact} information theoretic minimum $\mathcal{I}$ \cite{HNAW22, STHCG22}, but the gate complexity becomes $M^{O(F)}$). For the example of $F=\sqrt{M}$, this new encoding has only a negligible $o(1)$ fraction of unnecessary qubits, i.e. it is roughly twice as space efficient as in Section \ref{sec:efficientEnc}. These techniques make crucial use of the fact that the list is not antisymmetrized.

\subsection{Succinct encoding}
\label{subsec:simpleSuccinct}

We formally state the main theorem of this section below.

\begin{thm}
    There is a fermion data structure for strings of $M$ qubits having Hamming weight at most $F$ that uses $\mathcal{I}+O(F)$ qubits, with gate complexity $O(F \log M)$.
    \label{thm:fewQubits}
\end{thm}

Such an encoding is constructed in the remainder of this section, with Lemma \ref{lem:fewQubitEffZZZ} and Lemma \ref{lem:fewQubitEffXp} demonstrating that it satisfies the requisite efficiencies. Similar to Corollary \ref{corr:fermiEncFromQubits}, this implies a fermion to qubit mapping with the same complexities. 

\begin{corr}
    There exists a fermion encoding of an $M$ mode system having at most $F$ fermions which uses $\mathcal{I}+O(F)$ qubits such that any $k$-local, particle preserving rotation can be implemented with $O(F \log M)$ gates.
    \label{corr:fermiEncSuccinct}
\end{corr}

To reduce the number of necessary qubits, the key insight is noting that the most significant $G:=\lceil\log F\rceil$ qubits in each register of the encoding in Section \ref{subsec:algDef} are not pulling their weight. They use $F \cdot G = \Omega(F \log F)$ qubits to encode $F$ integers in the range $1$ to $F$, but they are in \emph{sorted order}, reducing the number of possible sequences--- see Figure \ref{fig:list-redundancy} for an illustration. We can instead store the same data (the most significant bits) in $O(F)$ space as a bit string where $2^{G}-1=O(F)$ zeroes delimit $2^G$ bins, and $F$ ones correspond to $F$ fermions. The $i$-th bin is the space from the $i$-th zero (or the start of the bit string) to the $i+1$-st zero (or the end of the bit string), and the number of ones in that range indicate elements with most significant bits $i$. This idea is commonly referred to as the stars and bars method.

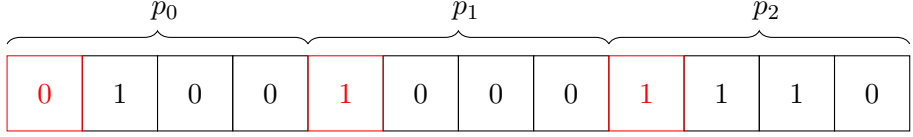
\begin{figure}
    \centering      
        \begin{tikzpicture}[scale=1.0]

            \foreach \i in {3,4,6,7,8,12} {
                \pgfmathtruncatemacro{\bin}{\i-1}
                \draw (\i,0) rectangle ++(1,1) node[midway] {$0$};
            }
            
            \foreach \i in {2,10,11} {
                \pgfmathtruncatemacro{\bin}{\i-1}
                \draw (\i,0) rectangle ++(1,1) node[midway] {$1$};
            }

            \draw [red] (1,0) rectangle ++(1,1) node[midway] {$0$};
            \foreach \i in {5,9} {
                \pgfmathtruncatemacro{\bin}{\i-1}
                \draw [red] (\i,0) rectangle ++(1,1) node[midway] {$1$};
            }

            \foreach \i in {1,2,3} {
                \pgfmathtruncatemacro{\im}{\i-1}
                \draw[fill=none] (4 * \i - 0.9, 1.6)  node {$p_\im$};
            }
            
            \draw [decorate,decoration={brace,amplitude=5pt,raise=10ex}]
  (1,-0.5) -- (5,-0.5) node{};
            \draw [decorate,decoration={brace,amplitude=5pt,raise=10ex}]
  (5,-0.5) -- (9,-0.5) node{};
            \draw [decorate,decoration={brace,amplitude=5pt,raise=10ex}]
  (9,-0.5) -- (13,-0.5) node{};

        \end{tikzpicture}
    \caption{An illustration of the redundancy in a sorted list, for a list of $3$ things out of a universe of $16$ things. The \textcolor{red}{red} bits are non-decreasing, with possibilities: \textcolor{red}{$000,001,011,111$}. Four possibilities means $2$ bits suffice, rather than $3$.}
    \label{fig:list-redundancy}
\end{figure}

The remaining least-significant bits can be stored in the usual way, where the order depends on the value determined by both least significant and most significant parts. So long as this order is maintained, the least-significant bits can be paired with the corresponding most-significant bits by the order in which they appear. This saves $F \lceil\log F\rceil$ qubits and adds $O(F)$ many, so the total qubit usage is \begin{align*}
    \text{Space}(M, F) =& F \lceil \log(M + 1) \rceil - F \lceil \log F \rceil + O(F) \\
    =& F \log \frac{M}{F} + O(F) \\
    =& \mathcal{I}+O(F).
\end{align*}
To define the encoding formally, for strings of the form \begin{align}
    x &= e_{i^{(1)}} \oplus e_{i^{(2)}} \oplus ... \oplus e_{i^{(F)}} \nonumber\\
    i^{(1)}&< i^{(2)} < ... < i^{(f)} \nonumber
\end{align}
for $f<F$, let $i_m$ denote the $G$ most significant bits of $i$, and $i_l$ denote the remaining least significant bits. Let $r_j$ be the number of pointers having most significant bits equal to $j$, i.e. the number of values of $k$ such that $i^{(k)}_m=j$. Now define  $\mathcal{E}^{(3)} : \{0, 1\}^M \rightarrow \{0,1\}^{\mathcal{I}+O(F)}$ as
\begin{align}
    \mathcal{E}^{(3)}(x) =& \ket{i^{(1)}_l} \ket{i^{(2)}_l} \cdots \ket{i^{(f)}_l} \ket{\infty_l} \cdots \ket{\infty_l} || 1^{r_0} 0 1^{r_1} 0 ... 0 1^{r_{2^G-1}}  
\end{align}
Where $||$ is a conceptual divider between information about the most and least significant bits, and $1^n$ denotes $n$ many $1$'s. The bits to the left of $||$ store the least significant bits of each pointer, and are referred to as $l$ or LSBs, and $i$-th register in this section is $l_i$. The bits to the right of $||$ are referred to as $m$ or MSBs and store the most significant bits of each pointer. Define $\mathcal{E}_s^{(3)} : \mathcal{H}_{2^M} \rightarrow \mathcal{H}_{2^{\mathcal{I}+O(F)}}$ as the unique linear extension of $\mathcal{E}^{(3)}$. A simple example of this encoding is shown in Figure \ref{fig:succinctEncodingExample}.

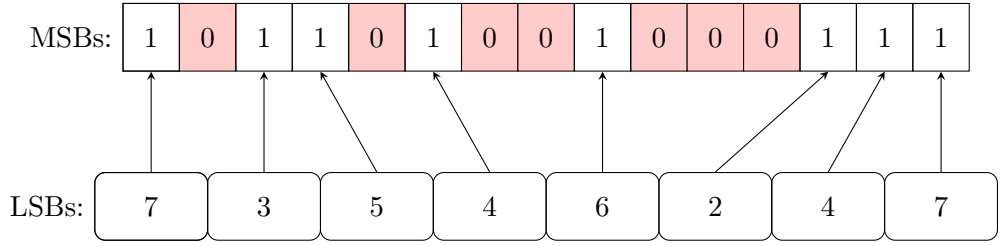
\begin{figure}
    \centering
    \begin{tikzpicture}[
  level 1/.style={sibling distance=10em},
  level 2/.style={sibling distance=5em}, scale=1.5,>=stealth]

    \tikzstyle{bits} = [rectangle, draw, minimum width = 0.75cm, minimum height = 0.9cm]
    \tikzstyle{register} = [rectangle, rounded corners, draw, minimum width = 2.25cm, minimum height = 0.9cm]
    \tikzstyle{registernar} = [rectangle, rounded corners, draw, minimum width = 1.5cm, minimum height = 0.9cm]

    \node[bits, label=left:MSBs:] at (-3.5, 0){ };
    
    \foreach \n in {1,4,6,7,9,10,11}{
        \node[bits, fill=red!20] (bits\n) at (0.5 * \n-3.5, 0){$0$};
    }
    \foreach \n in {0,2,3,5,8,12,13,14}{
        \node[bits] (bits\n) at (0.5 * \n-3.5, 0){$1$};
    }

    \node[registernar, label=left:LSBs:] at (-3.5, -1.5){ };
    \foreach \msbval [count=\n] in {7,3,5,4,6,2,4,7}{
        \node[registernar] (register\n) at (\n-4.5, -1.5){$\msbval$};
    }

    \foreach \bitval [count=\regval] in {0,2,3,5,8,12,13,14}{
        \draw [->] (register\regval.north) -- (bits\bitval.south);
    }
    
    \end{tikzpicture}
    \caption{Illustrative example of the succinct representation, with parameters $F=8$, $M=63$ and fermions at $[7,11,13,20,38,58,60,\infty]$; recall the convention that the largest storable value $(63)$ is identified with $\infty$. Red registers demarcate ``bins'' of different MSB values, out of $8$ possible values. The LSBs then only require $3$ qubits each, so the total storage is $15+24=39$ qubits. Using first quantization/sorted list encodings would use $8\cdot\lceil \log 63\rceil = 48$ qubits, and Jordan-Wigner or similar second quantized encodings would require $63$ qubits.}
    \label{fig:succinctEncodingExample}
\end{figure}

For correctness, it suffices if encoded $\bflip$ and $\srank$ operators satisfy the rules described in Section \ref{subsec:algDef}, on the list implicitly represented by our succinct string. In particular, the encoded $\bflip(p, \bb)$ operator should unitarily insert/delete $p$ from the sorted list defined by the information stored in the MSBs and LSBs, and update this data accordingly. The $\srank(p, \bb)$ operator should apply a minus phase if and only if there are an odd number of elements less or equal to $p$ in the sorted list.

\subsection{Efficient circuits}
\label{subsec:simpleSuccinctCircuits}
We now turn to constructing efficient circuits for the necessary operations on this succinct encoding. We note that retaining our space complexity advantage limits the number of ancillas that can be used to perform our operations to be $O(F)$ (and, naturally, we require these ancillas to be uncomputed by the end of each operation).
\begin{lem}
    Under the succinct list encoding $\mathcal E_s^{(3)}$ described in Section \ref{subsec:simpleSuccinct}, a $\srank(p, \bb)$ query has a circuit of $O(F \log M)$ gates, and $O(F)$ ancilla.
    \label{lem:fewQubitEffZZZ}
\end{lem}
\begin{proof}
    At a high level, this operation can be performed as follows: \begin{enumerate}[label=(\arabic*)]
        \item Scan the most significant bits $m$ to determine a range $[p_{start}, p_{end})$ of least significant bit registers whose most significant bits match $p$
        \item Use this range to implement comparisons $\leq p$ on all registers
        \item Phase the result of this comparison with a $Z$ gate
        \item Uncompute comparisons, then $p_{start}$ and $p_{end}$
    \end{enumerate}
    This is pictured in Figure \ref{fig:FewQubitZZZ}. Scanning the register $m$ takes $O(F)$ iterations costing $O(\log F)$ each. Each comparison afterwards costs $O(\log M/F)$ gates, and there are $O(F)$ many to make, so the total gate complexity is $O(F \log M)$. The only space used beyond comparisons are registers of size $O(\log F)$, satisfying the ancilla requirement.
\end{proof}
\begin{lem}
    Under the succinct list encoding $\mathcal E_s^{(3)}$ described in Section \ref{subsec:simpleSuccinct}, a $\bflip(p, \bb)$ query has a circuit of $O(F \log M)$ gates, and $O(F)$ ancilla.
    \label{lem:fewQubitEffXp}
\end{lem}
\begin{proof}
    Comparisons require information from the most significant bits, which is gathered first. Following this, analogous to the approach in Lemma \ref{lem:simpleXpComplexity} we will bubble the target element to the end if present, exchange with $\infty$, and then bubble back. At a high level: \begin{enumerate}[label=(\arabic*)]
        \item Compute and store a pointer $t$ to the target position in the list of $F$ elements, defined as the first register with value $\geq p$ (based on both it's most and least significant bits).
        \item Compute a flag $f_{del}$, which is on if and only if $p$ is present in the list.
        \item If $f_{del}$, shuffle $l_t$ to the end of $l$.
        \item Exchange $\infty_l \leftrightarrow p_l$ on the last register of $l$.
        \item If not $f_{del}$, shuffle the end of $l$ to $t$.
        \item If $f_{del}$, shuffle the $(t+p)$-th bit of $m$ to the end
        \item If not $f_{del}$, shuffle the last bit of $m$ to the $(t+p)$-th position.
        \item Uncompute $f_{del}$ and $t$.
    \end{enumerate}
    This is pictured in Figure \ref{fig:FewQubitX}. Steps (1) and (2) can be done in $O(F \log M)$ gates by computing temporary $p_{start}, p_{end}$ ranges of indices in $l$ which have the same most significant bits as $p$, implementing comparisons using these, then uncomputing. Steps (3), (4), (6), (7) are simple swap ladders which can be performed in $O(F \log M)$ gates, and Step (5) has negligible complexity. Step (8) has the same cost as steps (1), (2), so the overall gate complexity is $O(F \log M)$. Again, the only additional space stored is $O(\log F)$ size pointers $p_{start}$ and $p_{end}$, satisfying the ancilla requirement.
\end{proof}

\section{Achieving Low Depth and Succinctness}
\label{sec:fullEnc}
Combining the insights of Sections \ref{sec:efficientEnc} and \ref{sec:fewerQubits}, we can achieve the best of all worlds: a fermion data structure using $\mathcal{I}+O(F)$ qubits, with complexity $O(\log M \log \log M)$ depth, and $O(\mathcal{I})$ gates. The result is formally stated below in Theorem \ref{thm:fewQubitsLowDepth}.

\begin{thm}
    There is a fermion data structure for strings of $M$ qubits having Hamming weight at most $F$ that uses $\mathcal{I}+O(F)$ qubits, with gate complexity $O(\mathcal{I})$ and depth $O(\log M \log \log M)$.
    \label{thm:fewQubitsLowDepth}
\end{thm}

Such an encoding is constructed in the remainder of this section, with Lemma \ref{lem:fewQubLowDepEffZp} and Lemma \ref{lem:fewQubLowDepEffXp} demonstrating that it satisfies the requisite efficiencies. This implies a fermion to qubit mapping with the same complexities. 

\begin{corr}
    There exists a fermion encoding of an $M$ mode system having at most $F$ fermions which uses $\mathcal{I}+O(F)$ qubits such that any $k$-local, particle preserving rotation can be implemented with $O(F \log M)$ gates and $O(\log M \log \log M)$ depth.
    \label{corr:fermiEncSuccinctFull}
\end{corr}

We will start from the construction described in Section \ref{subsec:simpleSuccinct}, and build on additional data. This data should have a small (additive $O(F)$) footprint, but allow us to retrieve relevant information about the list in low depth.

The main new construction we will need is a succinct tree structure for storing the most significant bits. We will need to compute start and end intervals of register indices whose most significant bits match some value, for which it will suffice to consider queries for the position of the $t$-th zero in a list of $n$ bits; this is because the aforementioned information suffices to determine which registers match a given value on their most significant bits. This is sometimes referred to as a select query.

To answer such queries, we will keep a tree of sublist sums in such a way that a circuit can ``walk down'' all branches of the tree in parallel to quickly identify the location of a given zero. In particular, define $S_{a,b}$ as the total number of zeroes from position $a$ to $b-1$ in the list of most significant bits. If we store the number of zeroes of the first half of the list, this suffices to quickly determine whether the $t$-th zero is in the left or the right half. We recursively store the same information for each half, down to some fixed cutoff. By walking down this tree of (sub)ranks we can quickly find the position of the $t$-th zero.

\subsection{Algebraic Definition}
\label{subsec:varFewQubLowDepth}
An illustrative example of this definition is shown in Figure \ref{fig:succinctTreeExample}. Let $G := \lceil \log F\rceil$ be a cutoff between most and least significant bits, and $H := 2^G + F-1$ be the size of the most significant bits register, which will be stored using the stars and bars method described in Section \ref{subsec:simpleSuccinct} as well as with a new succinct tree data structure. For strings of the form \begin{align}
    x &= e_{i^{(1)}} \oplus e_{i^{(2)}} \oplus ... \oplus e_{i^{(f)}} \nonumber\\
    i^{(1)}&< i^{(2)} < ... < i^{(f)} \nonumber
\end{align}
with $f<F$, let $i_m$ denote the $G$ most significant bits of $i$, and $i_l$ denote the remaining $l$ bits. Let $r_j$ be the number of pointers with most significant bits matching $j$, i.e. the number of values of $k$ such that $i^{(k)}_m=j$. Let $m=1^{r_0} 0 1^{r_1} 0 ... 0 1^{r_{2^G-1}}$ be the bit string representing the most significant bits. Let $S_{a,b}$ denote the number of $0$'s among $m_a,...,m_{b-1}$, i.e. $S_{a, b}=(b-1)-a-m_a-...-m_{b-1}$. Now define  $\mathcal{E}^{(4)} : \{0, 1\}^M \rightarrow \{0,1\}^{\mathcal{I}+O(F)}$ as
\begin{align}
    \mathcal{E}^{(4)}(x) =& \ket{i^{(1)}_l} \cdots \ket{i^{(f)}_l} \ket{\infty_l} \cdots \ket{\infty_l} \,\, || \,\, m \,\, ||\,\,  S_{0, \frac{H}{2}} \,\,||\,\, S_{0, \frac{H}{4}} ... \,\,||\,\, S_{\frac{H}{2}, \frac{3H}{4}}...  \label{eqn:fewQubitLowDepthStateEncDefn}
\end{align}

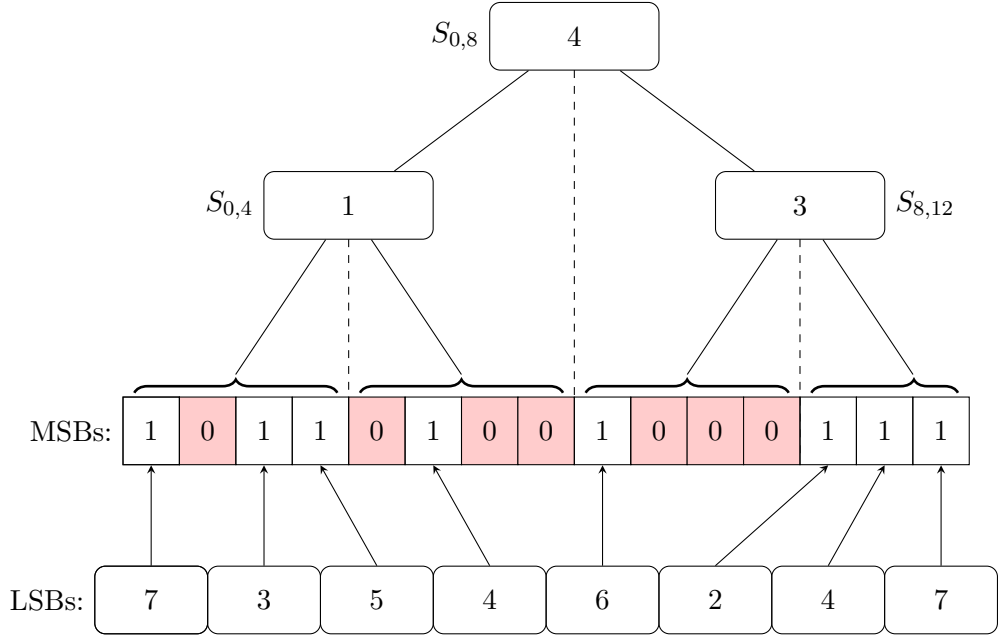
\begin{figure}
    \centering
    
    \begin{tikzpicture}[
  level 1/.style={sibling distance=10.4em},
  level 2/.style={sibling distance=5.2em}, scale=1.5,>=stealth]

    \tikzstyle{bits} = [rectangle, draw, minimum width = 0.75cm, minimum height = 0.9cm]
    \tikzstyle{register} = [rectangle, rounded corners, draw, minimum width = 2.25cm, minimum height = 0.9cm]
    \tikzstyle{registernar} = [rectangle, rounded corners, draw, minimum width = 1.5cm, minimum height = 0.9cm]
     
    \draw (0.25,0) node[register, label=left:$S_{0,8}$] (root) {$4$}
      child { node[register, label=left:$S_{0,4}$] {$1$}
        child { }
        child { }
      }
      child { node[register, label=right:$S_{8,12}$] {$3$}
        child { }
        child { }
      };

    \draw [dashed] (0.25, -0.3) -- (0.25, -3.75);
    \draw [dashed] (-1.75, -1.8) -- (-1.75, -3.75);
    \draw [dashed] (2.25, -1.8) -- (2.25, -3.75);

    \node[bits, label=left:MSBs:] at (-3.5, -3.5){ };
    
    \foreach \n in {1,4,6,7,9,10,11}{
        \node[bits, fill=red!20] (bits\n) at (0.5 * \n-3.5, -3.5){$0$};
    }
    \foreach \n in {0,2,3,5,8,12,13,14}{
        \node[bits] (bits\n) at (0.5 * \n-3.5, -3.5){$1$};
    }

    \foreach \ind in {0,4,8} {
        \draw [decorate,decoration={brace, amplitude=5pt, raise=0.05cm}, line width=1pt] ($(bits\ind.north west)+(3pt,0pt)$) -- ($(bits\the\numexpr\ind+3\relax.north east)-(3pt,0pt)$);
    }
        \draw [decorate,decoration={brace, amplitude=5pt, raise=0.05cm, aspect=0.7}, line width=1pt] ($(bits12.north west)+(3pt,0pt)$) -- ($(bits14.north east)-(3pt,0pt)$);

    \node[registernar, label=left:LSBs:] at (-3.5, -5){ };
    \foreach \msbval [count=\n] in {7,3,5,4,6,2,4,7}{
        \node[registernar] (register\n) at (\n-4.5, -5){$\msbval$};
    }

    \foreach \bitval [count=\regval] in {0,2,3,5,8,12,13,14}{
        \draw [->] (register\regval.north) -- (bits\bitval.south);
    }
    
    \end{tikzpicture}
    \caption{Illustrative example of the succinct tree, with parameters $F=8$ and $M=63$. This stores the pointer list $[7,11,13,20,38,58,60,\infty]$, recalling the convention that the largest storable value ($63$) is identified with $\infty$. Red registers demarcate ``bins'' of different MSB values.}
    \label{fig:succinctTreeExample}
\end{figure}

Where we have a tree with root node $S_{0, H/2}$ (rounding non-integer divisions arbitrarily), and each $S_{a, b}$ either has no children if $b-a=O(1)$ (we sometimes call this the leaf cutoff), or has a left child $S_{a, b+(b-a)/2}$ and right child $S_{b,b+(b-a)/2}$ otherwise. Define $\mathcal{E}_s^{(4)} : \mathcal{H}_{2^M} \rightarrow \mathcal{H}_{2^{\mathcal{I}+O(F)}}$ as the unique linear extension of $\mathcal{E}^{(4)}$. Note that the root of the sublist sum tree stores $O(\log H)$ qubits. To find the total space usage of the tree, we have recurrence \begin{align*}
    T(1) =& O(1) \\
    T(n) =& 2T(n/2) + O(\log n).
\end{align*} 
It follows from the master theorem that the total space usage of the tree is linear in $H$,
\begin{align*}
    T(H) =& O(H).
\end{align*}
Recalling that $H=O(F)$, we therefore only use \begin{align*}
    \underbrace{F \log \frac{M}{F}}_{\text{LSBs}} + \underbrace{O(F)}_{MSBs} = \mathcal{I}+O(F)    
\end{align*}
qubits under such a promise. Correctness is the same condition as described in Section \ref{subsec:algDef} for the implicit list of pointers, though in addition we must now properly maintain the tree data structure. Before constructing efficient circuits, we first explain two key subroutines for making our construction low depth.

\subsection{Parallel index finding}
In order to compare some value $p$ to a pointer which is stored in the succinct representation, we must compare the most significant bits and the least significant bits separately. The comparison against most significant bits is performed by computing a range $[p_{start},...,p_{end})$ of registers whose most significant bits match those of $p$; with this information the least significant bits can be compared as necessary. Here we discuss how to extract $p_{start}$ and $p_{end}$ in low depth.

The subroutine for this employed in Section \ref{subsec:simpleSuccinct}, referred to as $\#_p$, accomplished this by a linear scan. This subroutine counts the number of $0$'s, and finds the position $i$ of the $p_{msb}$-th $0$---this information suffices to compute $p_{start}$ (in particular, $p_{start}=i-p_{msb}+1$). Computing $p_{end}$ can be done similarly. Here we wish to use our sublist sum tree to avoid the high depth this incurs, and along the way will improve the efficiency of this subroutine from $O(F \log F)$ to $O(F)$.

\begin{lem}
    Given a list $m_1,...,m_H$ of bits as well as a sublist sum tree $(S_{0, H/2} (S_{0,H/4} ...) (S_{H/2,3H/4} ...))$ where $S_{a,b}$ counts the number of $0$'s among $m_a,...,m_{b-1}$, there is a reversible circuit which computes the position $i_t$ of the $t$-th $0$ in the list of bits using $O(H)$ gates, $O(H)$ ancilla, and $O(\log H \log \log H)$ depth.
    \label{lem:lowDepthCounting}
\end{lem}

\begin{proof}
    Let $t$ denote the target index. Define base-case parameters offset $o_{in}=t$ and flag $f_{in}=\top$, where the flag $f_{in}$ is true if the target is in the current sublist (i.e. at the root it is always somewhere in the full sublist, so it begins true) and offset $o_{in}$ is the value to be searched for if it is in the sublist. Starting at the root of the succinct tree and then recursing on children, perform: \begin{enumerate}[label=(\arabic*)]
        \item Let $S_{a,b}$ denote the current subinterval sum.
        \item If $f_{in}=\top$, determine whether the $o_{in}$-th $0$ is within this list lies to the left or right of index $d$ using $S_{0, H/2}$. \begin{enumerate}[label=(\alph*)]
            \item If $o_{in}$ is greater or equal to $S_{a,b}$, propagate $o_l=0$ (value does not matter) and $f_l=\bot$ to the left subarray, and $o_r=o_{in}-S_{a,b}$, $f_r=\top$ to the right subarray,
            \item If $o_{in}$ is less than $S_{a,b}$, propagate $o_l=o_{in}$ and $f_l=\top$ to the left subarray, and $o_r=0$, $f_r=\bot$ to the right subarray,
            \item If $f_{in}=\bot$, propagate $o_l=0$, $f_l = \bot$ to the left subarray, and $o_r=0$, $f_r=\bot$ to the right subarray.
        \end{enumerate}
        \item Recurse on children, or terminate at a leaf.
    \end{enumerate}
    An illustrative example is depicted in Figure \ref{fig:treeWalkdownCountExample}, and the key circuit subroutine is depicted in Figure \ref{fig:succinctTreeCountWalkdown}. Recursively call this procedure on children leaves until the base case of a single leaf, at which point we perform a sequential scan of the subarray conditioned on the incoming flag (i.e. the standard $\#_p$ operation). This scan produces a constant size number $i$ for each leaf. With this information, now walk back up the tree performing at each layer: \begin{enumerate}[label=(\arabic*)]
        \item Take positions $i_l$ and $i_r$ from left and right children, swap both into target register (at least one is promised to be all $0$),
        \item If the position was found in the right subtree, add the corresponding offset,
        \item Uncompute the children's flags,
        \item Recurse on parent (after both children are done).
    \end{enumerate}
    The circuit for this procedure is depicted in Figure \ref{fig:succinctTreeCountWalkup}. We recursively call this procedure until reaching the root, at which point all ancillas have been uncomputed and we are left with a single value $i$. This value corresponds to the position of the $p_{msb}$-th $1$ in the list. Further, each step uses at most a linear number of ancillas in the number of qubits being touched, and each qubit is touched $O(1)$ times. It follows that the whole procedure uses $O(H)$ ancillas. The depth analysis follows from the fact that each layer consists of $O(1)$ many comparisons on registers of size at most $O(\log H)$. Each comparison can be done with linearly many ancilla in $O(\log \log H)$ depth by Lemma \ref{lem:comparisonComplexity} and there are $O( \log H)$ layers, so the overall depth complexity of this step is $O(\log H \log\log H)$.
\end{proof}

\begin{figure}
    \centering
    \begin{tikzpicture}[
  level 1/.style={sibling distance=10.4em},
  level 2/.style={sibling distance=5.2em}, scale=1.5,>=stealth]

    \tikzstyle{bits} = [rectangle, draw, minimum width = 0.75cm, minimum height = 0.9cm]
    \tikzstyle{register} = [rectangle, rounded corners, draw, minimum width = 2.25cm, minimum height = 0.9cm]
    \tikzstyle{registernar} = [rectangle, rounded corners, draw, minimum width = 1.5cm, minimum height = 0.9cm]

    \draw (0.25,0) node[register, label=left:$S_{0,8}$, label=above:\textcolor{blue}{$3, \top$}] (root) {$4$}
      child { node[register, label=left:$S_{0,4}$] (left) {$1$}
        child { node[] (leftleft) {} }
        child { node[] (leftright) {} }
      }
      child { node[register, label=right:$S_{8,12}$] (right) {$3$}
        child { node[] (rightleft) {} }
        child { node[] (rightright) {} }
      };

    \draw [->, blue] (root) -- (left) node[midway, above left] {$3, \top$};
    \draw [->, blue] (left) -- (leftright) node[midway, above right] {$2, \top$};
    \draw [->, blue] (-1.75,-3) -- (-0.75,-3) node[right, right] {$i_0=2$};
    \draw [decorate,decoration={brace},blue] (-0.75,-3.175) -- (-0.25,-3.175);

    \draw [->, red] (root) -- (right) node[midway, above right] {$0, \bot$};
    \draw [->, red] (right) -- (rightright) node[midway, above right] {$0, \bot$};
    \draw [->, red] (right) -- (rightleft) node[midway, above left] {$0, \bot$};
    \draw [->, red] (left) -- (leftleft) node[midway, above left] {$0, \bot$};

    \node[bits, label=left:MSBs:] at (-3.5, -3.5){ };
    
    \foreach \n in {1,4,6,7,9,10,11}{
        \node[bits, fill=red!20] (bits\n) at (0.5 * \n-3.5, -3.5){$0$};
    }
    \foreach \n in {0,2,3,5,8,12,13,14}{
        \node[bits] (bits\n) at (0.5 * \n-3.5, -3.5){$1$};
    }

    \draw [dashed] (0.25, -2.75) -- (0.25, -4.25);
    \draw [dashed] (-1.75, -2.75) -- (-1.75, -4.25);
    \draw [dashed] (2.25, -2.75) -- (2.25, -4.25);

    \end{tikzpicture}
    \caption{Illustrative example of the comparisons needed to find the third zero in the succinct tree, with parameters $F=8$ and $M=63$. This is the same tree as in Figure \ref{fig:succinctTreeExample}, but with the LSBs excluded for clarity. As the position is walked back up the tree, requisite offsets are added and the ancillas are uncomputed.}
    \label{fig:treeWalkdownCountExample}
\end{figure}
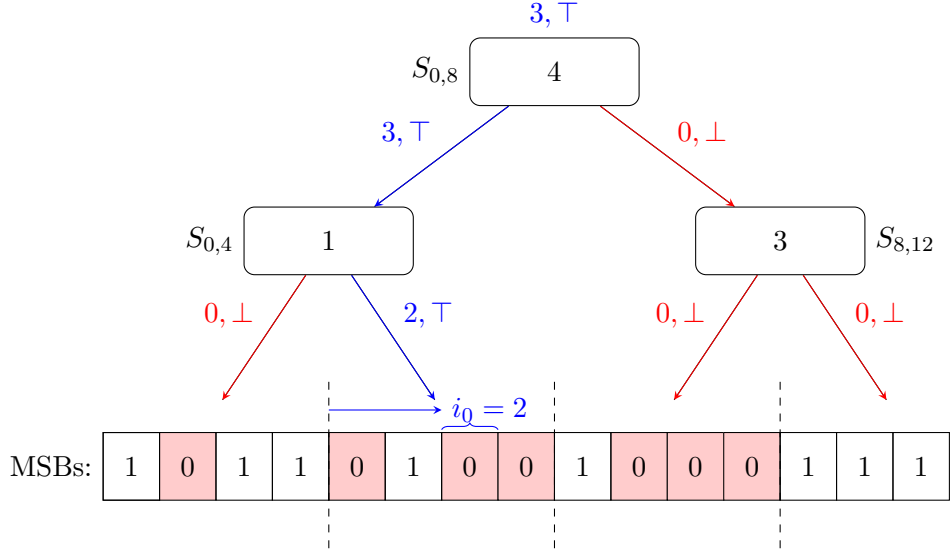

\subsection{Parallel list rotation}
Another key component of fermionic operations in our encoding is unitary insertion/deletion, which must preserve order. The core subroutine here is a list suffix rotation, shifting every element in some suffix down by one, and moving the last register to a target position. This operation either makes space for an element to be inserted, or the inverse circuit fills in the space of an element which was deleted. 

A useful subroutine will be cycling a list downwards and inserting an element, i.e. a low-depth implementation of the $L$ subroutine depicted in Figure \ref{fig:MSBSwapladder}---the unitary insertion can be done with the inverse circuit. This will allow permuting the list of least significant bits, as well as the stars and bars array representing the most significant bits. To accomplish this, assume that each register holds a flag for whether it is included in the suffix to be cycled (referred to as $f_{\geq p}$). Note that this subroutine takes as a promise that the last element is an $\infty$ (i.e. all $1$'s), which holds in our encoding by the assumption that $F$ fermions is not exceeded. We can permute such a list generally using the following lemma.

\begin{lem}
    Given a list of elements $x_1,...,x_q$ each of $n$ qubits, each with a flag $f_{\geq p}$ which is $1$ for only a suffix of the registers, and promised that $x_q = \infty$, for any value $p$ of $n$ bits there is a reversible circuit with $O(q \log n)$ gates, depth $O(\log n)$, and $O(q)$ ancillas that inserts $p$ in between elements marked $\geq p$ and the remaining, as well as removing the last entry (promised to be $\infty$).
    \label{lem:lowDepthRotationList}
\end{lem}

\begin{proof}
    The approach mirrors that taken in Lemma \ref{lem:lowDepthXp}, but instead of a full buffer register we will use a single qubit initialized to $1$. We perform the low-depth cycle circuit qubit by qubit; observe that by the premise all of the relevant comparison bits are already available, and so can be used in place of explicitly computing a comparison. Also note that the circuit described in Lemma \ref{lem:lowDepthXp} decides whether the target $p$ is present or not; in our circuit we assume it is not, i.e. all equality comparisons to $p$ are false. In this way, we obtain constant depth per round but require $O(\log n)$ rounds to swap all of the qubits, giving total depth $O(\log n)$. We use $O(q)$ buffer registers as ancillas, and $O(q)$ gates per round of swaps to give $O(q \log n)$ gate complexity.
\end{proof}

Note the difference in the above subroutine from that in Lemma \ref{lem:lowDepthXp}; this subroutine assumes the element is not present. When it is later called to permute a list, we will externally compute which direction to shift the list. Following this, we will either run the subroutine or it's inverse conditioned on the result. With this subroutine in hand, we can now describe the procedure for updating the succinct tree of sublist sums.

\begin{lem}
    Given a list $m_1,...,m_H$ of bits as well as a sublist sum tree $(S_{0, H/2} (S_{0,H/4} ...) (S_{H/2,3H/4} ...))$ extending to a constant cut off $S_{a, a+O(1)}$ where $S_{a,b}$ counts the number of $0$'s among $m_a,...,m_{b-1}$, there is a reversible circuit which cycles the last $t$ bits downwards (i.e. $m_t \rightarrow m_{t+1}$, ..., $m_H \rightarrow m_t$) using $O(H)$ gates, $O(H)$ ancilla, and $O(\log H \log \log H)$ depth.
    \label{lem:lowDepthRotationTree}
\end{lem}

\begin{proof}
    Let $t$ denote the target index. Define base-case parameters offset $o_{in}=t$ and flag $f_{in}=\top$. The flag $f_{in}$ will represent whether the index is beyond the start of the list to shuffle, and the offset will represent how far within a target list the index to shuffle is (or $\infty$ if it is beyond the end). Starting at the root $S_{0, H/2}$ of the succinct tree and walking down to children recursively, perform: \begin{enumerate}[label=(\arabic*)]
        \item Let $S_{a,b}$ denote the range under consideration.
        \item Determine how to update the sublist sum $S_{a,b}$.\begin{enumerate}[label=(\alph*)]
            \item If $f_{in}=\bot$, increment $S_{a,b}$ if $m_{a-1}$ is a $0$ (entering $0$), and decrement $S_{a,b}$ if $m_{b-2}$ is a $0$ (exiting $0$), then short circuit the remaining cases.
            \item If $o_{in}$ is greater or equal to $b-a$, do not update.
            \item If $o_{in}$ is less than $b-a$, decrement $S_{a,b}$ if $m_{b-2}$ is a $0$ (exiting $0$).
        \end{enumerate}
        \item Determine where the offset lies relative to the interval $[a, b)$, and propagate a signal to the corresponding side. \begin{enumerate}[label=(\alph*)]
            \item If $f_{in}=\bot$, propagate $o_l=\text{garbage}$, $f_l = \bot$ to the left subarray, and $o_r=\text{garbage}$, $f_r=\bot$ to the right subarray, and short circuit the remaining cases.
            \item If $o_{in}$ is greater or equal to $b-a$, propagate $o_l=\infty$ and $f_l=\top$ to the left subarray, and $o_r=o_{in}-(b-a)$ (or $\infty$ if $o_{in}=\infty$), $f_r=\top$ to the right subarray.
            \item If $o_{in}$ is less than $b-a$, propagate $o_l=o_{in}$ and $f_l=\top$ to the left subarray, and $o_r=\text{garbage}$ (i.e. value does not matter), $f_r=\bot$ to the right subarray.
        \end{enumerate}
        \item Recurse on children, or terminate at a leaf.
    \end{enumerate}
    The key circuit subroutine for this procedure is depicted in Figure \ref{fig:succinctTreeShiftWalkdown}. Recursively call this procedure on children leaves until the base case of a single leaf, at which point each bit $m_i$ of a given leaf can be marked with a flag corresponding to whether they must participate in the rotation. Following this, we utilize the circuit subroutine described in Lemma \ref{lem:lowDepthRotationList} to rotate the list of $m_i$'s which are marked in low depth---by the promise, the inserted value will always be a $1$. An illustrative example of this entire process is depicted in Figure \ref{fig:treeWalkdownShiftExample}.
    
    To uncompute, walk back up the tree and uncompute all the ancilla flags. We recursively call this procedure until reaching the root, at which point all ancillas have been uncomputed---note that we can perform this without modifying any of the updated $S_{a,b}$ because these sums were not used in computing the ancillas (see Figure \ref{fig:succinctTreeShiftWalkdown} for details).
    
    Further, each step uses at most a linear number of ancillas in the number of qubits being touched, and each qubit is touched $O(1)$ times. It follows that the whole procedure uses $O(H)$ ancillas. The depth analysis follows from the fact that each layer consists of $O(1)$ many comparisons on registers of size at most $O(\log H)$. Each comparison can be done with linearly many ancilla in $O(\log \log H)$ depth by Lemma \ref{lem:comparisonComplexity} and there are $O( \log H)$ layers, so the overall depth complexity of this step is $O(\log H \log\log H)$.
\end{proof}

\begin{figure}
    \centering
    \begin{tikzpicture}[
  level 1/.style={sibling distance=10.4em},
  level 2/.style={sibling distance=5.2em}, scale=1.5]

    \tikzstyle{bits} = [rectangle, draw, minimum width = 0.75cm, minimum height = 0.9cm]
    \tikzstyle{register} = [rectangle, rounded corners, draw, minimum width = 2.25cm, minimum height = 0.9cm]
    \tikzstyle{registernar} = [rectangle, rounded corners, draw, minimum width = 1.5cm, minimum height = 0.9cm]

    \draw (0.25,0) node[register, label=left:$S_{0,8}$, label=above:\textcolor{blue}{$10, \top$}] (root) {$4$}
      child { node[register, label=left:$S_{0,4}$] (left) {$1$}
        child { node[] (leftleft) {} }
        child { node[] (leftright) {} }
      }
      child { node[register, label=right:$S_{8,12}$] (right) {$3$}
        child { node[] (rightleft) {} }
        child { node[] (rightright) {} }
      };

    \draw [->, blue] (root) -- (right) node[midway, above right] {$2, \top$};
    \draw [->, blue] (right) -- (rightleft) node[midway, above left] {$2, \top$};
    \draw [decorate,decoration={brace},blue] (1.25,-3.175) -- (3.75,-3.175) node[midway, above] {flag};

    \draw [->, red] (root) -- (left) node[midway, above left] {$\infty, \top$};
    \draw [->, red] (left) -- (leftleft) node[midway, above left] {$\infty, \top$};
    \draw [->, red] (left) -- (leftright) node[midway, above right] {$\infty, \top$};
    \draw [->, red] (right) -- (rightright) node[midway, above right] {$\text{garb}, \bot$};

    \node[bits, label=left:MSBs:] at (-3.5, -3.5){ };
    
    \foreach \n in {1,4,6,7,9,10,11}{
        \node[bits, fill=red!20] (bits\n) at (0.5 * \n-3.5, -3.5){$0$};
    }
    \foreach \n in {0,2,3,5,8,12,13,14}{
        \node[bits] (bits\n) at (0.5 * \n-3.5, -3.5){$1$};
    }

    \draw [dashed] (0.25, -2.75) -- (0.25, -4.25);
    \draw [dashed] (-1.75, -2.75) -- (-1.75, -4.25);
    \draw [dashed] (2.25, -2.75) -- (2.25, -4.25);

    \draw [->] (bits10.south) to [out=-70, in=-110] (bits11.south);
    \draw [->] (bits11.south) to [out=-70, in=-110] (bits12.south);
    \draw [->] (bits12.south) to [out=-70, in=-110] (bits13.south);
    \draw [->] (bits13.south) to [out=-70, in=-110] (bits14.south);
    \draw [->] (bits14.south) to [out=-110, in=-70] (bits10.south) node[below,xshift=1.25cm,yshift=-1cm] {$L$ subroutine};

    \draw [->, red] (bits11.north) -- (right.south) node[midway, left] {$-1$};
    
    \end{tikzpicture}
    \caption{Illustrative example of the comparisons needed to shift everything at or beyond the tenth bit entry ($0$-based indexing) in the succinct tree, with parameters $F=8$ and $M=63$. The data structure is the same as in Figure \ref{fig:succinctTreeExample}, but with LSBs excluded for clarity. As the offset is walked back up the tree, ancillas are uncomputed. Note the difference in the meaning of the flag between that in Figure \ref{fig:treeWalkdownCountExample}; here it represents whether the target position lies to the right of the array or within it. This is because elements to the right of target position must be shifted whereas those to the left do not, so the two cases must be handled separately.}
    \label{fig:treeWalkdownShiftExample}
\end{figure}
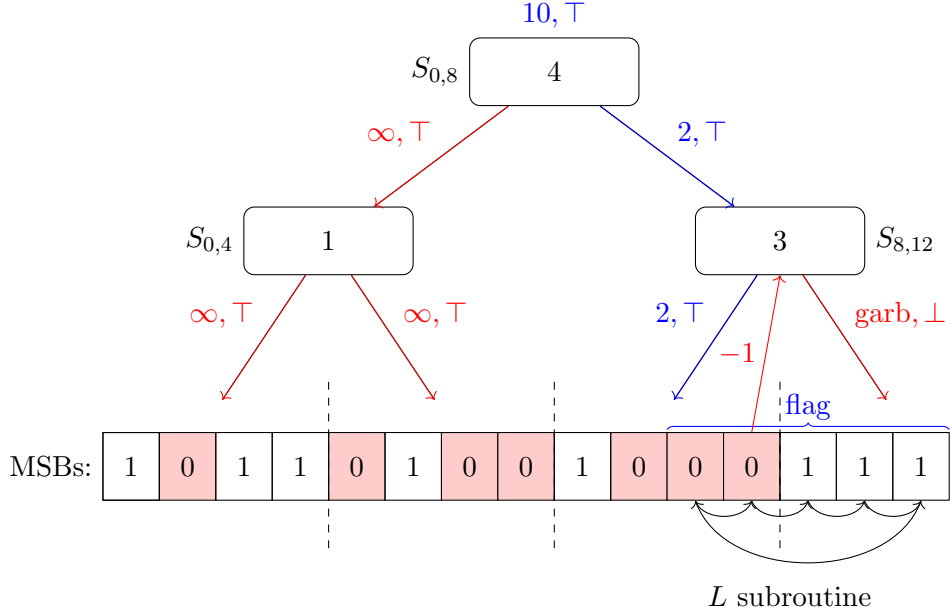

\subsection{Circuit Algorithms}
We now turn to implementations of the operations needed to define a fermion data structure. The implementation will broadly be similar to that in Section \ref{sec:fewerQubits}, but we will here describe how to parallelize the high-depth components.  

\begin{lem}
    Under the succinct list+rank-tree encoding $\mathcal E_s^{(4)}$ described in Section \ref{subsec:varFewQubLowDepth}, a $\srank(p, \bb)$ query has a circuit of $O(\mathcal I)$ gates, depth $O(\log M \log \log M)$, and $O(F)$ ancilla.
    \label{lem:fewQubLowDepEffZp}
\end{lem}

\begin{proof}
    This algorithm is similar to the one described in Lemma \ref{lem:fewQubitEffZZZ} and pictured in Figure \ref{fig:FewQubitZZZ}. In particular, the only high depth components of that algorithm are the $\#_p$ algorithm for implementing select queries to the most significant bits, as well as the sequential comparisons. We know from Lemma \ref{lem:lowDepthCounting} that the $\#_p$ subroutine can be done in depth $O(\log F \log \log F)$, with gate complexity $O(F)$ and using at most $O(F)$ ancilla.

    The only remaining high depth part of the algorithm are the sequential comparisons, which can be avoided once the start and end registers are computed by fanning out these values to make $O(F / \log M)$ copies. This uses $O(F)$ ancillas, and can be done in depth $O(\log F)$. Then one can perform $O(F / \log M)$ comparisons in parallel using depth $O(\log \log M)$ each by Lemma \ref{lem:comparisonComplexity} with the approach described in Figure \ref{fig:LSBComparisons}. Each register is touched once, so the gate complexity is $O(F \log M/F)$. After $O(\log M)$ rounds this is completed, giving a total depth of $O(\log M \log \log M)$, and never using more than $O(F)$ ancillas.

    Performing these two steps sequentially, we obtain full depth $O(\log M \log \log M + \log F \log \log F) = O(\log M \log \log M)$. This circuit utilizes $O(F)$ ancilla total, and $O(F\log M/F)+O(F)=O(\mathcal{I})$ many one and two qubit gates.
\end{proof}

\begin{lem}
    Under the succinct list+rank-tree encoding $\mathcal E_s^{(4)}$ described in Section \ref{subsec:varFewQubLowDepth}, a $\bflip(p, \bb)$ query has a circuit of $O(\mathcal I)$ gates, depth $O(\log M \log \log M)$, and $O(F)$ ancilla.
    \label{lem:fewQubLowDepEffXp}
\end{lem}

\begin{proof}

    This algorithm is similar to the one described in Lemma \ref{lem:fewQubitEffZZZ} and pictured in Figure \ref{fig:FewQubitX}. In particular, the high depth components of this algorithm are the $\#_p$ algorithm for finding indices in the succinct representation of the most significant bits, the sequential comparisons, and the array shifting subroutines. We know from Lemma \ref{lem:lowDepthCounting} that the $\#_p$ subroutine can be done in depth $O(\log F \log \log F)$, with gate complexity $O(F)$ and using at most $O(F)$ ancilla, which we use to compute a $p_{start}$ and $p_{end}$.

    Using the computed $p_{start}$ and $p_{end}$, the sequential comparisons can be avoided by fanning out these values to make $O(F / \log M)$ copies, which uses $O(F)$ ancillas. Then one can perform $O(F / \log M)$ comparisons in parallel using depth $O(\log \log M)$ each by Lemma \ref{lem:comparisonComplexity}, using the approach in Figure \ref{fig:LSBComparisons}. After $O(\log M)$ rounds this is completed, giving a total depth of $O(\log M \log \log M)$. Each register is acted on once, so the gate complexity is $O(F \log M/F)$. Using this, we can compute a flag $f_{ins}$ and $f_{\geq p}$, placing a copy at each register of least significant bits $l_1,...,l_F$. Flag $f_{ins}$ is true if and only if $p$ is present in the encoded list, and $f_{\geq p}$ is true at register $r$ if and only if the value encoded by $r$ and it's corresponding MSBs is $\geq p$.

    With these flags placed, we uncompute $p_{start}$ and $p_{end}$. Following this, we permute the registers of the most significant bits using the approach described in Lemma \ref{lem:lowDepthRotationList}, conditioned on $f_{ins}=\top$. If flag $f_{ins}=\bot$, then conditioned on this we perform the inverse circuit to remove the target element. Finally, using $p_{start}$ we cycle the list of most significant bits and update the tree using the approach of Lemma \ref{lem:lowDepthRotationTree}, again conditioning on $f_{ins}$: if $f_{ins}=\top$ then we rotate downwards using the described circuit, if $f_{ins}=\bot$ then we rotate upwards using the inverse circuit. Finally, we compute new $p_{start}', p_{end}'$, using these to uncompute every $f_{ins}$ and each registers $f_{\geq p}$. Note that $f_{\geq p}$ is preserved by the cycling subroutine, making this value easy to uncompute. Further, $f_{ins}$ will be $\top$ if and only if, after the insertion/deletion subroutine, $p$ is present in the list (i.e. the reverse condition of before the subroutine). We finish by uncomputing these $p_{start}'$ and $p_{end}'$.

    Note that every described operation on the least significant bits uses $O(F)$ gates, depth at most $O(\log F \log \log F)$, and $O(F)$ ancillas. Every described operation on the most significant bits uses $O(F \log M/F)$ gates, and depth $O(\log M \log \log M)$, and $O(F)$ ancillas. It follows from $F \leq M$ that the entire operation requires depth $O(\log M \log \log M)$, gate complexity $O(F \log M/F)+O(F)=O(\mathcal{I})$, and uses $O(F)$ ancillas.
\end{proof}

\section{Implicit Fermion Data Structure}
\label{sec:implicit-struct}

In this section, we abandon our previous framework of a sorted list of pointers and aim for an even more concise data structure. Inspired by the techniques of Harrison et al.\ \cite{HNAW22}, we will instead (implicitly) associate all bit strings of length $M$, Hamming weight $f$ with consecutive integers $0, \ldots, \binom{M}{f}-1$. For a system with $F$ fermions, we build a capacity $k$ fermion data structure, which stores bit strings of weight $f \in [F-k, ..., F+k]$. For $f = F+k$, the number of bit strings is bounded by  
\[
\binom{M}{F+k} = \binom{M}{F} \cdot \frac{F! (M-F)!}{(F+k)! (M-F-k)!} \leq \binom{M}{F} \left( \frac{M}{F} \right)^{k}.
\]
The smaller values of $f$ have negligible contribution due to the exponential growth in $k$. In terms of space usage, this translates to $\mathcal{I} + O(k \cdot \log \tfrac{M}{F})$ additional qubits. When $F = \Theta(M)$ and $k = 1$, the overhead is $O(1)$ bits, meaning this is an \emph{implicit} fermionic data structure.

However, fermionic operations on this new data structure use $\mathsf{poly}(M)$ gates; tolerable for our target regime of $F = \Theta(M)$, but potentially exponentially worst than our previous encodings if, e.g., $F = \Theta(\log M)$. Nonetheless, $\mathsf{poly}(M)$ gate complexity compares favorably with prior work using $\mathcal I + O(1)$ qubits, where the same operations cost $O(M^{F+2})$ \cite{HNAW22} or $M^{O(F)}$ \cite{STHCG22} gates. Furthermore, when $F=\Theta(M)$ this gate overhead is comparable to the overheads in prior works that do not achieve near-optimal space efficiency.

In short, we prove the following result.
\begin{thm}
	\label{thm:implicitEncoding}
	There exists a fermion data structure using $\mathcal{I} + O(k\log \tfrac{M}{F})$ qubits where $\bflip$ and $\srank$ operations are possible with $O(MF\mathcal{I})$ gates.
\end{thm}
As previously discussed, the most interesting regime is $k = O(1)$, $F = \Theta(M)$ since we have $\mathcal{I} + O(1)$ space usage. 

The high-level construction associates configurations/bit strings directly to the numbers. First, we imagine the configurations are mapped to \emph{consecutive} numbers $0, 1, \ldots$, where the configurations are ordered first by weight and breaking ties with lexicographic order. The operations $\bflip(1, \bb)$ and $\srank(1, \bb)$ are relatively easy in this configuration, since the configurations are nearly sorted by $b_1$ (after Hamming weight, of which we allow only $2k+1$ values). 

In fact, for each bit $j$ we define an ordering of the configurations/bit strings where $\bflip(j, \bb)$ and $\srank(j, \bb)$ are comparatively easy. However, the proof of Theorem~\ref{thm:implicitEncoding} depends on using extra qubits to align blocks of configurations with powers of two. Section~\ref{sec:padding} spells out the details: how many ancilla qubits, how much padding per block, and where it exists in the array when we do not need it.

\subsection{Notation}
\label{subsec:notation}

\subsubsection{Bit strings and orderings}

Let \( \strings{m}{f} \subseteq \{ 0, 1 \}^{*} \) represent the set of bit strings of length $m$ and Hamming weight $f$.\footnote{Not to be confused with the Stirling numbers, which use the same notation, but will not appear in this paper.} Clearly $\strings{m}{f}$ has $\binom{m}{f} = \frac{m!}{(m-f)!f!}$ elements. Hence, $\strings{m}{f}$ is non-empty if and only if $0 \leq f \leq m$, and there is a nice recursive decomposition, 
\begin{equation}
\label{eqn:pascal}
\strings{m}{f} = 0 \strings{m-1}{f} \cup 1 \strings{m-1}{f-1},
\end{equation}
analogous to Pascal's rule for binomial coefficients. We use the notation $\sstrings{m}{f}$ to represent the same set of strings, but sorted in lexicographic order. Note that the above decomposition works analogously for the sorted list, replacing the union with concatenation. We use the notation $\sstrings{m}{f}^R$ to denote the same set of strings but sorted by \textit{reversed} lexicographic order, i.e. the most significant bit is the right-most (the comparison is $x^R \prec y^R$ where $x^R$ is string reversal and $\prec$ is ordinary lexicographic order). Note that this is not the inverse of lexicographic order.

\subsubsection{Subroutines}
\label{subsec:array-subroutines}

In this section, we introduce important subroutines for permuting a range of integers, which we view as manipulating blocks of a hypothetical array $\mathcal A$ --- rotating, interleaving, de-interleaving, etc. Outside the bounds of the array, we require the permutation to be the identity. 

\begin{prop}
	\label{prop:direct_sum}
	Suppose $U$ is a unitary of the form $U_{<b} \oplus U_{\geq b}$, meaning that it acts on basis states $\ket{0}, \ldots, \ket{b-1}$ and $\ket{b}, \ket{b+1}, \ldots$ separately. Then we can implement $U' = U_{<b} \oplus I$ using a controlled-$U$ gate, $O(1)$ ancillas, and $O(n)$ gates, where $n$ is the number of qubits.
\end{prop}
\begin{proof}
	Test if the register is $<b$ and put the result into an ancilla. Apply $U$ controlled on the ancilla. Test if the register is $<b$ again to uncompute the ancilla. Note that uncomputing depends on the direct sum structure of $U$. 
\end{proof}

\noindent
We now introduce the first operation, rotation. 
\begin{defn}
	For $i \geq 0$ and $0 \leq k < n$, let $\Call{rotate}{i, k, n}$ be the permutation 
	\[
	j \to \begin{cases}
		j & \text{if $j < i$ or $j \geq i+n$,} \\
		j+k & \text{if $i \leq j < i+n-k$,} \\
		j+k-n & \text{if $i+n-k \leq i+n$.}
	\end{cases}
	\]
	In other words, rotate the subarray $\mathcal{A}[i..i+n-1]$ by $k$. 
\end{defn}

\begin{prop}
	\label{prop:prefix_rotation}
	For $0 \leq k < n$, there is a circuit with $O(1)$ ancillas implementing $\Call{Rotate}{0, k, n}$.
\end{prop}
\begin{proof}
    For indices less than $n$, the map is addition (modulo $n$) by $k$, i.e., $j \mapsto (j+k) \bmod n$. When $n$ is a power of $2$, a simple ripple adder can add $k$ and forget the final carry to achieve this with $1$ ancilla qubit. (See \cite{CDKM04}; our task is actually even simpler since one summand is fixed)
	
	Otherwise, let us pad the array to a power of two length, and then double it. Since the length is a power of two, we can rotate the whole list by an arbitrary amount. By construction, each half is a power of two, so by controlling the rotation on the most significant bit, we can arbitrarily rotate either half.
	
	We claim that the following sequence of such operations will exchange $A$ and $B$, as desired. We will also refer to this operation as a ``\textsf{SWAP}'' between $A$ and $B$.
	
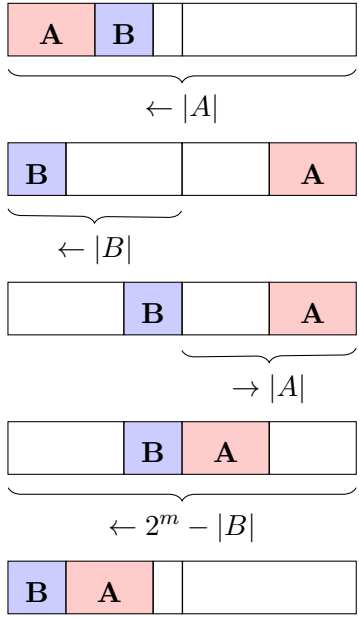
	\begin{center}
		\begin{tikzpicture}[start chain=1 going right, node distance=-0.1mm,
			start chain=2 going right,
			start chain=3 going right,
			start chain=4 going right,
			start chain=5 going right]
			\edef\sizetape{0.7cm}
			\edef\vertspace{1.2cm}
			\tikzstyle{tape}=[draw,minimum size=\sizetape]
			
			\node [on chain=1,draw=none,minimum width=1em,text height=1em] (tape1) {};
			\node [on chain=1,tape,minimum width=3em,fill=red!20,text height=1em] {\textbf{A}};
			\node [on chain=1,tape,minimum width=2em,fill=blue!20,text height=1em] {\textbf{B}};
			\node [on chain=1,tape,minimum width=1em,text height=1em] {};
			\node [on chain=1,tape,minimum width=6em,text height=1em] (mid1) {};
			\node [on chain=1,draw=none,minimum width=1em,text height=1em] (end1) {};
			
			\node [on chain=2 going right,below=\vertspace of tape1,draw=none,minimum width=1em,text height=1em] (tape2) {};
			\node [on chain=2,tape,minimum width=2em,fill=blue!20,text height=1em] {\textbf{B}};
			\node [on chain=2,tape,minimum width=4em,text height=1em] {};
			\node [on chain=2,tape,minimum width=3em,text height=1em] (mid2) {};
			\node [on chain=2,tape,minimum width=3em,fill=red!20,text height=1em] {\textbf{A}};
			\node [on chain=2,draw=none,minimum width=1em,text height=1em] (end2) {};
			
			\node [on chain=3 going right,below=\vertspace of tape2,draw=none,minimum width=1em,text height=1em] (tape3) {};
			\node [on chain=3,tape,minimum width=4em,text height=1em] {};
			\node [on chain=3,tape,minimum width=2em,fill=blue!20,text height=1em] {\textbf{B}};
			\node [on chain=3,tape,minimum width=3em,text height=1em] (mid3) {};
			\node [on chain=3,tape,minimum width=3em,fill=red!20,text height=1em] {\textbf{A}};
			\node [on chain=3,draw=none,minimum width=1em,text height=1em] (end3) {};
			
			\node [on chain=4 going right,below=\vertspace of tape3,draw=none,minimum width=1em,text height=1em] (tape4) {};
			\node [on chain=4,tape,minimum width=4em,text height=1em] {};
			\node [on chain=4,tape,minimum width=2em,fill=blue!20,text height=1em] {\textbf{B}};
			\node [on chain=4,tape,minimum width=3em,fill=red!20,text height=1em] (mid4) {\textbf{A}};
			\node [on chain=4,tape,minimum width=3em,text height=1em] {};
			\node [on chain=4,draw=none,minimum width=1em,text height=1em] (end4) {};
			
			\node [on chain=5 going right,below=\vertspace of tape4,draw=none,minimum width=1em,text height=1em] (tape5) {};
			\node [on chain=5,tape,minimum width=2em,fill=blue!20,text height=1em] {\textbf{B}};
			\node [on chain=5,tape,minimum width=3em,fill=red!20,text height=1em] {\textbf{A}};
			\node [on chain=5,tape,minimum width=1em,text height=1em] {};
			\node [on chain=5,tape,minimum width=6em,text height=1em] {};
			\node [on chain=5,draw=none,minimum width=1em,text height=1em] (end5) {};
			
			\draw (tape1.south east) edge[decorate,decoration={brace,amplitude=5pt,mirror,raise=2mm}] node[below=4mm] {$\gets |A|$} (end1.south west);
			
			\draw (tape2.south east) edge[decorate,decoration={brace,amplitude=5pt,mirror,raise=2mm}] node[below=4mm] {$\gets |B|$} (mid2.south west);
			
			\draw (mid3.south west) edge[decorate,decoration={brace,amplitude=5pt,mirror,raise=2mm}] node[below=4mm] {$\to |A|$} (end3.south west);
			
			\draw (tape4.south east) edge[decorate,decoration={brace,amplitude=5pt,mirror,raise=2mm}] node[below=4mm] {$\gets 2^{m} - |B|$} (end4.south west);
		\end{tikzpicture}
	\end{center}
	
	\noindent
	Here $m$ is chosen so that $2^m \geq |A| + |B|$. 
\end{proof}

\begin{prop}
    \label{prop:shift}
    If we can apply an operation $U$ to the beginning of the array ($\ket{0}, \ldots, \ket{b-1}$) then we can apply it to an arbitrary subarray (i.e., $\ket{i}, \ldots, \ket{i+b-1}$). 
\end{prop}
\begin{proof}
    Prefix rotation is enough to move any subarray to the beginning. Conjugate $U$ by that rotation to achieve the desired operation.
\end{proof}

\begin{corr}
    \label{cor:rotation}
    For any $i \geq 0$, $0 \leq k < n$, there is a circuit implementing $\Call{Rotate}{i, k, n}$ with $O(1)$ ancillas. 
\end{corr}
\begin{proof}
    Clearly $\Call{Rotate}{i, k, n}$ is a shifted instance of $\Call{Rotate}{0, k, n}$, and Proposition~\ref{prop:shift} finishes the proof.
\end{proof}

The other chief operation is interleaving two lists of blocks. 
\begin{defn}
	  Let $\Call{interleave}{}$ be the permutation which maps blocks $A_1, \ldots, A_n$ and $B_1, \ldots, B_n$ as follows 
	\[
	A_1 A_2 \cdots A_n B_1 B_2 \cdots B_{n} \mapsto A_1 B_1 A_2 B_2 \cdots A_{n} B_{n}
	\]
	where $|A_1|= \cdots = |A_n|$ and $|B_1| = \cdots = |B_n|$.
\end{defn}
We will assume we have such an algorithm for general $n$, $|A_i|$, $|B_i|$, although space/efficiency constraints will eventually limit us to Theorem~\ref{thm:power_interleave_full}. For now, we proceed with these operations.

\subsection{Algorithm invariants}
\label{subsec:algorithm-orderings}
For a fermion data structure of capacity $k$, the set of configurations is $\mathcal L = \strings{M}{F-k} \cup ... \cup \strings{M}{F+k}$. The initial encoding of a string $x \in \mathcal L$ is given by the order of $x$ under the ordering by Hamming weight with ties broken by lexicographic order. That is, the initial ordering is $\mathcal L_0 = \sstrings{M}{F-k} + ... + \sstrings{M}{F+k}$. We define \begin{align*}
    l_0(x) &= i \text{ s.t. } \mathcal L_0[i] = x,
\end{align*}
and we call $l_0(x)$ the ``label'' of $x$ in the list $\mathcal L_0$. In different stages of the algorithm the strings will be ordered differently: we will use subscripts to denote this. The subscripts will take values from $0$ up to $M$, and will usually be denote with the letter $j$. At a given stage $j$, it will also be useful to define the concept of a ``book'' and a ``chapter'' of configurations.

\begin{defn}
    A \emph{book} of configurations at level $j$ is a maximal subset of configurations $S \subset \sstrings{M}{F}$, such that every $x \in S$ has a size-$(M-j)$ suffix of the same hamming weight. In particular, $|x[j+1,...,M]|_{HW}$ is a constant for all $x \in S$.
\end{defn}

\begin{defn}
    A \emph{chapter} of configurations at level $j$ is a maximal subset $S$ of a level $j$ book configurations, such that every $x \in S$ has the same size-$(j)$ prefix. In particular, $x[1,...,j]$ is the same string for all $x \in S$.
\end{defn}

We can analogously define a book or chapter of labels as the maximal subset of labels of a book or chapter of configurations. The high level idea of our algorithm is that it will be easy to manipulate books and chapters of labels using the subroutines discussed in Section \ref{subsec:array-subroutines}. At the $j$-th level, it will be helpful to visualize storing the implicit list of $j$ level books (and the chapters within said books) in a certain convenient ordering. This ordering will allow us to perform operations on the $j$-th bit of the configuration efficiently, and will enable us to transition efficiently to the neighboring $j-1$-th and $j+1$-th levels efficiently as well. 

As a warm-up, we will consider flipping the first bit of the configuration, given just the label string. Our algorithm for performing the $\bflip$ or $\srank$ operators will largely be similar. Clearly a $\bflip(j, \bb)$ is a permutation on the label set, and $\srank(j, \bb)$ applies a conditional phase to some subset of the labels. Using the encoding above, we can efficiently flip the first bit of $\bb$, or apply a phase dependent on this first bit, acting only on the label $l_0(\bb)$. Consider first the action of flipping the first bit, i.e. $\bflip(1, \bb)$. For the case of $k=1$, $F=2$, and $M=4$, the permutation is depicted in Figure \ref{fig:firstBitSwap}. In particular, note that it is simply the product of two different \textsf{SWAP} subroutines: between the chapter of prefix $0$ in the $\sstrings{4}{1}$ book and the chapter of prefix $1$ in the $\sstrings{4}{2}$ book, and between the chapter of prefix $0$ n the $\sstrings{4}{2}$ book and the chapter of prefix $1$ in the $\sstrings{4}{3}$ book.

\begin{figure}
    \begin{center}
   	\begin{tikzpicture}[start chain=1 going below, start chain=2 going below, node distance=-0.1mm]
		\node [on chain=1] (L0001) {$0001$};
		\node [on chain=1] (L0010) {$0010$};
		\node [on chain=1] (L0100) {$0100$};
		\node [on chain=1] (L1000) {$1000$};
		\node [on chain=1] (L0011) {$0011$};
		\node [on chain=1] (L0101) {$0101$};
		\node [on chain=1] (L0110) {$0110$};
		\node [on chain=1] (L1001) {$1001$};
		\node [on chain=1] (L1010) {$1010$};
		\node [on chain=1] (L1100) {$1100$};
		\node [on chain=1] (L0111) {$0111$};
		\node [on chain=1] (L1011) {$1011$};
		\node [on chain=1] (L1101) {$1101$};
		\node [on chain=1] (L1110) {$1110$};
		
		\node [on chain=2, right=3cm of L0001] (R0001) {$0001$};
		\node [on chain=2] (R0010) {$0010$};
		\node [on chain=2] (R0100) {$0100$};
		\node [on chain=2] (R1000) {$1000$};
		\node [on chain=2] (R0011) {$0011$};
		\node [on chain=2] (R0101) {$0101$};
		\node [on chain=2] (R0110) {$0110$};
		\node [on chain=2] (R1001) {$1001$};
		\node [on chain=2] (R1010) {$1010$};
		\node [on chain=2] (R1100) {$1100$};
		\node [on chain=2] (R0111) {$0111$};
		\node [on chain=2] (R1011) {$1011$};
		\node [on chain=2] (R1101) {$1101$};
		\node [on chain=2] (R1110) {$1110$};
		
		\draw (L1000.east) edge[->,red] (R1000.west);
		
		\draw (L0100.east) edge[->] (R1100.west);
		\draw (L0010.east) edge[->] (R1010.west);
		\draw (L0001.east) edge[->] (R1001.west);
		\draw (L1100.east) edge[->] (R0100.west);
		\draw (L1010.east) edge[->] (R0010.west);
		\draw (L1001.east) edge[->] (R0001.west);
		
		\draw (L0110.east) edge[->] (R1110.west);
		\draw (L0101.east) edge[->] (R1101.west);
		\draw (L0011.east) edge[->] (R1011.west);
		\draw (L1110.east) edge[->] (R0110.west);
		\draw (L1101.east) edge[->] (R0101.west);
		\draw (L1011.east) edge[->] (R0011.west);
		
		\draw (L0111.east) edge[->,red] (R0111.west);
		
		\draw (L1000.west) edge[decorate,decoration={brace,amplitude=5pt,raise=0mm}] node[left=3mm] {$\displaystyle \sstrings{4}{1}$} (L0001.west);
		
		\draw (L1100.west) edge[decorate,decoration={brace,amplitude=5pt,raise=0mm}] node[left=3mm] {$\displaystyle \sstrings{4}{2}$} (L0011.west);
		
		\draw (L1110.west) edge[decorate,decoration={brace,amplitude=5pt,raise=0mm}] node[left=3mm] {$\displaystyle \sstrings{4}{3}$} (L0111.west);
		
		\node[right=3mm of R1000] {$1 \sstrings{3}{0}$};
		
		\draw (R0100.east) edge[decorate,decoration={brace,amplitude=5pt,mirror,raise=0mm}] node[right=3mm] {$0\sstrings{3}{1}$} (R0001.east);
		\draw (R1100.east) edge[decorate,decoration={brace,amplitude=5pt,mirror,raise=0mm}] node[right=3mm] {$1\sstrings{3}{1}$} (R1001.east);
		\draw (R0110.east) edge[decorate,decoration={brace,amplitude=5pt,mirror,raise=0mm}] node[right=3mm] {$0\sstrings{3}{2}$} (R0011.east);
		\draw (R1110.east) edge[decorate,decoration={brace,amplitude=5pt,mirror,raise=0mm}] node[right=3mm] {$1\sstrings{3}{2}$} (R1011.east);
		
		\node[right=3mm of R0111] {$0 \sstrings{3}{3}$};
	\end{tikzpicture}
    \end{center}
    \caption{An example of $\bflip$ on the first bit of $\strings{4}{1} \cup \strings{4}{2} \cup \strings{4}{3}$, pictured in ordering $\mathcal L_0$. Due to the lexicographic ordering, $\sstrings{4}{2}$ is $0\sstrings{3}{2}$ followed by $1\sstrings{3}{1}$. The $\bflip$ exchanges these with (respectively) the $0 \sstrings{3}{1}$ tail of $\sstrings{4}{1}$, and the $1 \sstrings{3}{2}$ head of $\sstrings{4}{3}$. When $\bflip$ is not possible, e.g., on $1000$ or $0111$, our convention is to do nothing. The simple structure of this operation enables it to be done efficiently.}
    \label{fig:firstBitSwap}
\end{figure}

In fact, the above is a special case of a more general fact that certain operations are easy under certain orderings. In particular, we will define $M$ many orderings $<_j$ for $s \in [M]$, such that there is an efficient circuit for flipping the $j$-th bit or applying a phase conditioned on it when the labels are ordered by the $<_j$. We define this ordering as follows.

\begin{defn}
    Let $x, y \in \mathcal L$ and $s \in [M]$. Let $x=x_1 \Vert x_2, y = y_1 \Vert y_2$ where $x_1$ and $y_1$ are prefixes of size $j$. Letting $\prec$ denote lexicographic order and $a^R$ denote the reversal of string $a$, we define the order $<_j$ as 
    \begin{align}
        x <_j y :=& \begin{cases}
            \top & \text{ if } |x|_{HW} < |y|_{HW} \\
            \top & \text{ if } |x|_{HW} = |y|_{HW} \text{ and } |x_2|_{HW} < |y_2|_{HW} \\
            \top & \text{ if } |x|_{HW} = |y|_{HW} \text{ and } |x_2|_{HW} = |y_2|_{HW} \text{ and } x_1^R \Vert x_2 \prec y_1^R \Vert y_2 \\
            \bot & \text{ otherwise}
        \end{cases}
    \end{align}
\end{defn}

In words, sort (a) first by over all hamming weight, breaking ties by (b) hamming weight on the suffix, breaking further ties by (c) lexicographic order on the reversal of the prefix, and breaking ties in all three prior orders by (d) lexicographic order on the suffix. We will refer to these cases by (a-d) following this. Observe that $<_0$ is precisely the ordering of the initial list $\mathcal L_0$, as the prefix is of size zero. Each ordering above naturally defines a set of labelling functions for the configurations.

\begin{defn}
    Define $\mathcal L_j$ as a list containing all elements of $\mathcal L$ such that for any $x, y \in \mathcal L$ we have that $x$ appears before $y$ if $x <_j y$. Furthermore, we define a ``block'' as a contiguous list of $x \in \mathcal L$ that share the same hamming weight and length $j$ prefix, followed by vacuous labels padding to the nearest power of two. We require that $\mathcal L_j$ is formed by the concatenation of blocks in the order given by $<_j$.
\end{defn}

Note that the labelling $l_j$ keeps all level $j$ books of labels contiguous and distinct. Within each book, $l_j$ keeps all level $j$ chapters contiguous and distinct, and sorted by reverse lexicographic order of the size-$(j)$ prefix defining said chapters. Further, chapters within a book appear at constant stride.

Now, for any $j \in [1,...,M]$ (note $j\neq 0$) let us define the action of a ``logical'' bit-flip $\tilde X_j$ as the permutation induced by $\bflip(j, \bb)$ on the label space determined by $l_j$, and a ``logical'' phase-flip $\tilde Z_j$ similarly. In particular, we define them as follows.

\begin{defn}
    Let $x \in \mathcal L$ be a bit string with label $l_j(x)$ given by labelling function $l_j$. We define $\tilde X_j$ as the quantum circuit acting on labels such that \begin{align*}
        \tilde X_j \ket{l_j(x)} =& \ket{l_j(\bflip(j, x))},
    \end{align*}
    and the operator $\tilde Z_j$ as a quantum circuit acting on labels such that \begin{align*}
        \tilde Z_j \ket{l_j(x)} =& \begin{cases}
            \ket{l_j(x)} & \text{ if } x[j] = 0 \\
            -\ket{l_j(x)} & \text{ if } x[j] = 1
        \end{cases}
    \end{align*}
\end{defn}

Note that the action of $\tilde X_j$ and $\tilde Z_j$ depend only on $x$, but as explicit circuits will depend on the labelling function $l_j$: the details of $l_j$ will determine the complexity of $\tilde X_j$ and/or $\tilde Z_j$. We define these only for label function $l_j$, as we aim to construct $l_j$ such that the corresponding $\tilde X_j$ and $\tilde Z_j$ can be performed efficiently. Clearly, the ability to perform $\tilde X_j$ and $\tilde Z_j$ operators suffice to generate the operators necessary for a fermion data structure.

For the remainder of this section, we will ignore the padding to powers of two in order to simplify explanations: it's only purpose is to simplify subroutines. We discuss this aspect in more detail in Section \ref{sec:padding}.

\begin{lem}
    Let $\mathcal L_j$ be the list as above, which defines a label $l_j(x)$ for any $x \in \mathcal L$. Then there are quantum circuits $A_j, B_j$ using $O(F \mathcal I)$ gates and $O(1)$ ancillas such that \begin{align*}
        \mathsf{X}_j \ket{l_j(x)} =& \tilde X_j \ket{l_j(x)} \\
        \mathsf{Z}_j \ket{l_j(x)} =& \tilde Z_j \ket{l_j(x)}.
    \end{align*}
    \label{lem:efficientSingleQubsOrder}
\end{lem}

\begin{proof}
    Let us begin by constructing $\mathsf{X}_j$. Consider parameters $n \in [F-k, ..., F+k-1]$ and $f \in [0, ..., n]$. Let $\mathcal X$ be a subset of the book of configurations of hamming weight $n+1$ whose suffix of length $M-j$ has hamming weight $f$; specifically $\mathcal X$ are those chapters whose $j$-th bit is a $0$. By the fact that books are distinct and contiguous, and the fact that chapters are sorted by reverse lexicographic order on size-$(j)$ prefix (hence the $j$-th bit is most significant), $\mathcal X$ is a contiguous set of chapters at constant stride. We will use $[x, ..., x+d]$ to denote the range of labels in $\mathcal X$.
    
    Let $\mathcal Y$ be a subset of the book of configurations of hamming weight $n+1$ whose suffix of length $M-j$ has hamming weight $f$; specifically $\mathcal Y$ is those chapters whose $j$-th bit is a $1$. Note that, by the same argument above, $\mathcal Y$ is a contiguous set of chapters. Further, both $\mathcal Y$ and $\mathcal X$ are sets of the same size, i.e. $\mathcal Y$ is a range $[y, ..., y+d]$. This is because a suffix $s_2$ of size $(M-j)$ and hamming weight $f$, and a prefix $s_1$ of size $(j-1)$ and hamming weight $n-f$, determines a string $s_1 \Vert 0 \Vert s_2$ in $\mathcal X$ and $s_1 \Vert 1 \Vert s_2$ in $\mathcal Y$. Further, all such strings are of this form.
    
    It also follows that sets $\mathcal X$ and $\mathcal Y$ will be swapped by $\tilde X_j$, with their internal orders preserved. This is because, using the form described previously, the elements of $\mathcal Y$ and $\mathcal X$ are sorted by lexicographic rank of $s_1^R \Vert s_2$, as the $j$-th bit is and relevant hamming weights are fixed.
    
    To implement this operation, we simply perform a \textsf{SWAP} subroutine between $[x, ..., x+d]$ and $[y, ..., y+d]$, using the procedure described in Proposition \ref{prop:prefix_rotation}.

    The values $x, y, d$ can be efficiently calculated classically offline, as they will only be used as fixed values in the circuit (again, we ignore padding here, but the calculation can straightforwardly be adapted to include it). Therefore, said \textsf{SWAP} implements $\tilde X_j$ on the sets $\mathcal X$ and $\mathcal Y$ using $O(\mathcal I)$ gates. To perform the full $\tilde X_j$, we simply repeat this operation for every value of $f$ and $n$, of which there are $O(F) \cdot O(k)$ possibilities, giving full complexity $O(F \mathcal I)$ for constant $k$. Each iteration takes $O(1)$ ancilla which can be uncomputed between rounds, so overall $O(1)$ ancilla are required.

    To perform a $\tilde Z_j$ we can similarly perform $O(F)$ repetitions for all values of $f, n$. For a given pair $(f, n)$ defining sets $\mathcal X$ and $\mathcal Y$ as before, we will apply a phase to all labels of bitstrings in $\mathcal Y$ (but not $\mathcal X$). Recall that $\mathcal Y$ is precisely the labels of strings having a $1$ at position $j$ (and hamming weight $n$ with suffix hamming weight $f$). We can compute a flag for whether a given label resides in $\mathcal Y$, i.e. whether it is within $[y, ..., y+d]$, apply a Pauli $Z$ to this flag, and then uncompute. This again takes $O(\mathcal I)$ gates to compute the flag, and $O(F)$ iterations to attain complexity $O(F \mathcal I)$. Each iteration needs only $O(1)$ ancillas which are uncomputed at the end, so the whole routine takes $O(1)$ ancilla.
\end{proof}

\subsection{Traversing the layers}

Given Lemma \ref{lem:efficientSingleQubsOrder}, it is clear that we would like to be able to permute the labels in such a way that they correspond to the rank under any order $<_j$, for differing $j$'s. We will further maintain the invariant that the labels appear in ``blocks'' sharing the same $j$-size prefix, and padded out to the nearest multiple of two. Let us define permutations $\pi_{0\mapsto 1}$, $\pi_{1\mapsto 2}, ..., \pi_{M-1 \mapsto M}$ that act on the labels. Writing the action implied on the full list of labels, we require that these permutations satisfy \begin{align*}
    \pi_{j \mapsto j+1} (\mathcal L_j) :=& \mathcal L_{j+1}.
\end{align*}
Suppose for now that we can implement these permutations $\pi_{j \mapsto j+1}$ in complexity $O(F \mathcal I)$ (as circuits on the label strings with $O(1)$ ancilla): this is proven in Theorem \ref{thm:efficient-cycle-perm}. Using this fact, we could then implement an arbitrary prefix of $\tilde Z_j$ operators or instance of $\tilde X_j$ with $O(MF\mathcal I)$ gates and $O(1)$ ancilla. Noting that $\srank$ and $\bflip$ are precisely these two operations, this suffices to prove our main theorem.

\begin{proof}[Proof of Theorem \ref{thm:implicitEncoding}]
    First consider applying a $\tilde Z_j$ on all $j < p$, i.e. the $\srank(j, \bb)$ operator, on label string $l_0(x)$. We first apply $\pi_{0 \mapsto 1}$ to obtain $l_1(x)$. We then apply $\tilde Z_1$, which can be done in $O(F \mathcal I)$ gates by Lemma \ref{lem:efficientSingleQubsOrder}. We follow with permutation $\pi_{1 \mapsto 2}$, transforming $l_1(x) \mapsto l_2(x)$. Again calling Lemma \ref{lem:efficientSingleQubsOrder}, we can now perform $\tilde Z_2$. We continue with $\pi_{2 \mapsto 3}$ to attain $l_3(x)$, then $\tilde Z_3$, and etc. continuing up to $\tilde Z_p$. After this sequence, we have $l_p(\srank(p, x))$. To return to the original encoding we simply undo the re-ordering permutations, applying $\pi_{p-1 \mapsto p}^\dagger$ down to $\pi_{1 \mapsto 2}^\dagger$, finally obtaining $l_1(\srank(p, \tilde Z_p))$.

    To apply a  $\bflip(p, x)$ to $l_1(x)$ for some $x \in \mathcal{L}$, we can again simply apply the sequence of permutations $\pi_{1 \mapsto 2}$ up to $\pi_{p-1 \mapsto p}$ to obtain $l_p(x)$. We then apply $\tilde X_p$ by the procedure in Lemma \ref{lem:efficientSingleQubsOrder} to obtain $l_p(\bflip(p, x))$, and then undo the re-ordering permutations, applying $\pi_{p-1 \mapsto p}^\dagger$ down to $\pi_{1 \mapsto 2}^\dagger$ to obtain $l_1(\bflip(p, x))$.

    In either case, there are $O(M)$ many alternations of re-ordering permutation and $\tilde X_j$/$\tilde Z_j$, each of which requiring $O(F \mathcal I)$ gates. Hence, the overall complexity of the operation would be $O(M F \mathcal I)$.
\end{proof}

It remains to show that we can implement the permutations in the desired efficiency, which is the main technical component of this section.

\begin{thm}
    Let $\pi_{j \mapsto j+1}$ be the permutation on bitstrings $l_j(x)$ for $x \in \mathcal L$, as defined above. There is a quantum circuit $C$ that satisfies \begin{align}
        C \ket{l_j(x)} =& \ket{l_{j+1}(x)},
    \end{align}
    meaning $C$ implements the permutation $\pi_{j \mapsto j+1}$ on label strings. Further, $C$ uses $O(1)$ ancillas and requires $O(F \mathcal I)$ quantum gates.
    \label{thm:efficient-cycle-perm}
\end{thm}

\begin{proof}
    Note first that the permutation $\pi_{j \mapsto j+1}$ does not re-order strings of a different hamming weight: under all orderings $<_j$, a string of lesser hamming weight is always lesser. Hence, we can consider applying the re-ordering $\pi_{j \mapsto j+1}$ for only a subset of strings of a fixed hamming weight, say $F$. We can then repeat this procedure $O(k)$ times, with cycles as necessary, to perform the procedure on labels of strings of any hamming weight in $F-k,...,F+k$. We construct a circuit that takes the label under $<_j$ of $x \in \strings{M}{F}$, meaning the number strings $y \in \strings{M}{F}$ that satisfy $y<_j x$. The circuit produces the label under $<_{j+1}$ of $x$, meaning the number strings $y \in \strings{M}{F}$ that satisfy $y<_{j+1} x$. The action of this circuit is depicted in Figure \ref{fig:listITerationsExample3}.

    To ease notation, let us define the list product of lists-of-strings $[x_1,...,x_m]$ and $[y_1,...,y_n]$ as $$[x_1, ..., x_m] \times [y_1, ..., y_n] = [x_1\Vert y_1, ..., x_1 \Vert y_n, x_2 \Vert y_1, ..., x_m\Vert y_n],$$
    and the list sum, $+$, of two lists as the concatenation of the two lists. These operations should be thought of as notational short-hand, not algebraic operations: note for instance that the list product is not necessarily distributive over list addition. In general, the implicit array of labelled values is of the following form. \begin{align}
        \mathcal Q :=& \sstrings{j}{j}^R \times \sstrings{M-j}{F-j} + \sstrings{j}{j-1}^R \times \sstrings{M-j}{F-j+1}  + \cdots + \sstrings{j}{0}^R \times \sstrings{M-j}{F}. \nonumber
    \end{align}
    We can re-write $\mathcal Q$ as 
     \begin{align}
        \mathcal Q =& \sum_{m=0}^{F} \sstrings{j}{j-m}^R \times \sstrings{M-j}{F-j+m} \label{eqn:arrayDecompDirect}\\
        =& \sum_{m=0}^{F} \underbrace{\sstrings{j}{j-m}^R \times \underbrace{\left(0\sstrings{M-j-1}{F-j+m} + 1\sstrings{M-j-1}{F-j+m-1}\right)}_{\text{Chapter}}}_{\text{Book}}, \label{eqn:arrayDecompFactored}
    \end{align}
    where we have denoted the $j$-th level book by the product term, and the chapters being sub-arrays of said books with a fixed prefix. We can similarly define the $j+1$-th level books and chapters, but observe that they are not yet sorted. However, for a given level $j+1$ book $B$, we have that it's chapters $C_1,...,C_h$ for some $h$ are already sorted. Note that the $j+1$-th level book $B$ with size-$(M-j-1)$ suffix hamming weight $m$ is formed from two parts: one part from the $j$-th level book $B_l$ with size-$(M-j)$ suffix hamming weight $m$ and a $0$ at position $j+1$, and one part from the $j$-th level book $B_r$ with size-$(M-j)$ suffix hamming weight $m+1$ and a $1$ at position $j+1$. The book $B_l$ contributes all strings which have a $0$ at position $j+1$, and the book $B_r$ contributes all strings which have a $1$ at position $j+1$. Furthermore, all strings from book $B_r$ appear after any string in book $B_l$. Finally, from the inductive hypothesis the chapters of book $B_l$ and $B_r$ are each correctly sorted by reverse-lexicographic on prefix and lexicographic on suffix order (ignoring the $j+1$-th bit, which is fixed), we have that the chapters of $B$ are properly sorted. By the same reasoning, each chapter is contiguous.

    Observe however that the books are out of order, as their chapters are jumbled together. In particular, note that a given level $j-1$ book $B$, as labelled in Equation \ref{eqn:arrayDecompFactored} and indexed by $m$, contains first strings with size-$(M-j-1)$ suffix having hamming weight $F-j+m+1$, followed by those with hamming weight $F-j+m$. These are out of order, and precisely $\binom{j-1}{j-m-1}$ copies per term are out of order pairs after taking the list product.

    Hence the only way the desired order is broken is that books are not properly sorted, while for a given book it's chapters and their contents appear in order. We can fix this by a simple \textsf{SWAP} then \textsf{DEINTERLEAVE} procedure, pictured in Figure \ref{fig:listITerationsExample3}. First, we swap the order of the two sum terms in the final product in Equation \ref{eqn:arrayDecompDirect}, which can be done with $O(F)$ many parallel \textsf{SWAP} operators (one case for each value of the hamming weight of size-$(j)$ prefix). This gives list \begin{align}
        \textsf{SWAP}^{\otimes F}\mathcal Q &= \sum_{m=0}^{F} \sstrings{j}{j-m}^R \times \left(1\sstrings{M-j-1}{F-j+m-1} + 0\sstrings{M-j-1}{F-j+m}\right),\label{eqn:arrayDecompSwapped}
    \end{align}
    Note that this operation does not change the order of the chapters of any given book, just re-ordering the chapters between books.
    
    Following this, we can \textsf{DEINTERLEAVE} over each half of the chapters of the $j$-level books, with have determined by the $j+1$-th bit. We do this using the procedure outlined in Lemma \ref{lem:power_interleave}. This preserves order for a given $j+1$-level book, as its chapters are always on the same side of a \textsf{DEINTERLEAVE} operation. However, this orders the $j+1$-level books to be sorted, as each chapter is of the $j$-th book is sorted by hamming weight of size-$(M-j-1)$ prefix. We are then left with the list \begin{align}
        \mathcal Q_F &= \textsf{DEINTERLEAVE}^{\otimes F}\textsf{SWAP}^{\otimes F}\mathcal Q \nonumber \\
        =& \sum_{m=0}^{F} \sstrings{j+1}{j+1-m}^R \times \left(\sstrings{M-j-1}{F-j+m-1} \right),\label{eqn:arrayDecompFinal}
    \end{align}
    which is ordered by $<_{j+1}$. Each \textsf{SWAP/DEINTERLEAVE} requires $O(\mathcal I)$ gates, and there are $O(F)$ of each, so the complexity is $O(F \mathcal I)$. Similarly, each requires $O(1)$ ancilla which can be uncomputed, so the whole operation requires $O(1)$ ancilla.

\end{proof}

\subsection{Handling Padding}
\label{sec:padding}

In this section, we finally discuss \emph{padding} --- in-bounds indices of $\mathcal{A}$ which are not associated with a fermion configuration. The purpose of padding is to align configurations to powers of $2$, which makes it trivial to apply operations in parallel, see below. 
\begin{prop}
	Let $\mathcal{A}$ be an array. Given an $\ell$-bit operation $U$ (on an array of length $2^{\ell}$) and a length $k$, there is trivial circuit which applies $U$ to $k$ consecutive subarrays $\mathcal{A}[i \cdot 2^{\ell}..(i+1) \cdot 2^{\ell}-1]$ in parallel, where $0 \leq i < k$.
\end{prop}
\begin{proof}
    Compare all but the $\ell$ least significant bits of the index with $k$, and store the result in an ancilla. Apply $U$ to the $\ell$ least significant bits when the rest are $<k$. Repeat the test to uncompute and clear the ancilla. 
\end{proof}
We also note that we can shift the parallel operation with Proposition~\ref{prop:shift}, so we have considerable flexibility where we apply it. 

\noindent
In other words, finding the quotient and remainder is \emph{much} easier in base $2$ when the divisor is a power of $2$. For the same reasons, power of $2$ block lengths dramatically simplify $\Call{interleave}{}$.

We first discuss our strategy for padding, the amount of extra space required ($O(1)$), and how the padding moves through the computation. After that, we look at the specific padding requirements to do $\Call{interleave}{}$ efficiently. 

\subsubsection{High-Level Strategy}

The number of bit strings in a typical block, $\binom{m}{f}$, is not a power of $2$. Rounding up to the next power of $2$ can nearly double its size. On top of this, the $\Call{interleave}{}$ construction in the next subsection depends on a further factor of two increase. In the worst case, we quadruple the length of any particular block. Since we operate on \emph{disjoint} blocks at any one time, it is sufficient to quadruple the array length with padding. This translates to only two ancilla qubits, not counting ancillas used for adders and controlled gates.

Initially, the padding exists in a \emph{pool} at the end of the array. Rotation (Corollary~\ref{cor:rotation}) makes it easy to insert padding from the pool into any position in the array, or conversely, return padding from the array back to the pool. Since there is overhead to adding or removing padding, it is convenient to leave each block padded to a power of $2$. More precisely, as we traverse the layers, we keep ``chapters'' in the current layer padded to a power of $2$, since this is necessary to use Theorem~\ref{thm:power_interleave_full} for the de-interleaving step. 

\subsubsection{Interleaving with Padding}

To realize the $\Call{interleave}{}$ subroutine, we quickly find that we need to divide by $|A_i|$ or $|B_i|$ or $|A_i| + |B_i|$. The division circuit is considerably more complicated when these divisors are not powers of two. In this section we show how to do $\Call{interleave}{}$ with the use of padding to align the blocks to power-of-two lengths.  

To begin, consider the case of identical, power-of-two block lengths: $|A_i| = |B_i| = 2^k$.

\begin{lem}
	\label{lem:power_interleave}
	There is a subroutine $\Call{simple-interleave}{}$ which permutes 
	\[
	A_1 A_2 \cdots A_n B_1 B_2 \cdots B_{n} \mapsto A_1 B_1 A_2 B_2 \cdots A_{n} B_{n}
	\]
	where $A_1, \cdots, A_{n}, B_1, \cdots B_{n}$ are blocks of length $2^k$.
\end{lem}
\begin{proof}
	Let $i$ be an arbitrary index. Since the blocks are shuffled and all the same length, $i \bmod 2^k$ does not change and does not affect the high order bits. Let us make the simplifying assumption that $k = 0$. 
	
	When $n$ is a power of $2$, there is a particularly nice bitwise expression for the permutation: rotate right. I.e., shift the bits right and let the most significant wrap around to the least significant. 
	
	Otherwise, use rotations (Corollary~\ref{cor:rotation}) to extend the $A_i$ and $B_i$ lists to have a power of two blocks, where the extra blocks are entirely padding drawn from the pool. We can then apply a power of two

\end{proof}

\newcommand{\blank}{\textvisiblespace}

\begin{figure}
	\begin{tikzpicture}[start chain=1 going right, node distance=-0.1mm,
		start chain=2 going right,
		start chain=3 going right,
		start chain=4 going right,
		start chain=5 going right,
		start chain=6 going right,
		every on chain/.append style={text height=1em}]
		\edef\sizetape{0.7cm}
		\edef\vertspace{1.2cm}
		\tikzstyle{xtape}=[draw,minimum size=\sizetape,inner sep=0pt]  
		\tikzstyle{atape}=[draw,minimum size=\sizetape,fill=red!20]  
		\tikzstyle{btape}=[draw,minimum size=\sizetape,fill=blue!20]  
		\tikzstyle{dtape}=[draw=none,minimum size=\sizetape]  
		\tikzstyle{ntape}=[draw=none,minimum size=\sizetape]  
		
		\node [on chain=1,ntape,minimum width=1.0em] (tape1) {};
		\node [on chain=1,atape,minimum width=2.0em] (A11) {$\mathbf{A}_1$};
		\node [on chain=1,xtape,minimum width=1.0em] {};
		\node [on chain=1,atape,minimum width=2.0em] (A21) {$\mathbf{A}_2$};
		\node [on chain=1,xtape,minimum width=1.0em] {};
		\node [on chain=1,dtape,minimum width=2.0em] {$\ldots$};
		\node [on chain=1,atape,minimum width=2.0em] (An1) {$\mathbf{A}_n$};
		\node [on chain=1,xtape,minimum width=1.0em] {};
		\node [on chain=1,btape,minimum width=3.5em] (B11) {$\mathbf{B}_1$};
		\node [on chain=1,xtape,minimum width=2.5em] {};
		\node [on chain=1,btape,minimum width=3.5em] (B21) {$\mathbf{B}_2$};
		\node [on chain=1,xtape,minimum width=2.5em] {};
		\node [on chain=1,dtape,minimum width=2.0em] {$\ldots$};
		\node [on chain=1,btape,minimum width=3.5em] (Bn1) {$\mathbf{B}_n$};
		\node [on chain=1,xtape,minimum width=2.5em] {};
		\node [on chain=1,xtape,minimum width=9.0em] (pool1) {\textcolor{gray}{\textit{pool}}};
		\node [on chain=1,ntape,minimum width=1.0em] (end1) {};
		
		\node [on chain=2,ntape,minimum width=1.0em,below=\vertspace of tape1] (tape2) {};
		\node [on chain=2,atape,minimum width=2.0em] (A12) {$\mathbf{A}_1$};
		\node [on chain=2,xtape,minimum width=1.0em] {};
		\node [on chain=2,atape,minimum width=2.0em] (A22) {$\mathbf{A}_2$};
		\node [on chain=2,xtape,minimum width=1.0em] {};
		\node [on chain=2,dtape,minimum width=2.0em] {$\ldots$};
		\node [on chain=2,atape,minimum width=2.0em] (An2) {$\mathbf{A}_n$};
		\node [on chain=2,xtape,minimum width=1.0em] {};
		\node [on chain=2,xtape,minimum width=9.0em] (pool2) {\textcolor{gray}{\textit{pool}}};
		\node [on chain=2,btape,minimum width=3.5em] (B12) {$\mathbf{B}_1$};
		\node [on chain=2,xtape,minimum width=2.5em] {};
		\node [on chain=2,btape,minimum width=3.5em] (B22) {$\mathbf{B}_2$};
		\node [on chain=2,xtape,minimum width=2.5em] {};
		\node [on chain=2,dtape,minimum width=2.0em] {$\ldots$};
		\node [on chain=2,btape,minimum width=3.5em] (Bn2) {$\mathbf{B}_n$};
		\node [on chain=2,xtape,minimum width=2.5em] {};
		\node [on chain=2,ntape,minimum width=1.0em] (end2) {};
		
		\node [on chain=3,ntape,minimum width=1.0em,below=\vertspace of tape2] (tape3) {};
		\node [on chain=3,atape,minimum width=2.0em] (A13) {$\mathbf{A}_1$};
		\node [on chain=3,xtape,minimum width=4.0em] {};
		\node [on chain=3,atape,minimum width=2.0em] (A23) {$\mathbf{A}_2$};
		\node [on chain=3,xtape,minimum width=4.0em] {};
		\node [on chain=3,dtape,minimum width=2.0em] {$\ldots$};
		\node [on chain=3,atape,minimum width=2.0em] (An3) {$\mathbf{A}_n$};
		\node [on chain=3,xtape,minimum width=4.0em] {};
		\node [on chain=3,btape,minimum width=3.5em] (B13) {$\mathbf{B}_1$};
		\node [on chain=3,xtape,minimum width=2.5em] {};
		\node [on chain=3,btape,minimum width=3.5em] (B23) {$\mathbf{B}_2$};
		\node [on chain=3,xtape,minimum width=2.5em] {};
		\node [on chain=3,dtape,minimum width=2.0em] {$\ldots$};
		\node [on chain=3,btape,minimum width=3.5em] (Bn3) {$\mathbf{B}_n$};
		\node [on chain=3,xtape,minimum width=2.5em] {};
		\node [on chain=3,ntape,minimum width=1.0em] (end3) {};
		
		\node [on chain=4,ntape,minimum width=1.0em,below=\vertspace of tape3] (tape4) {};
		\node [on chain=4,atape,minimum width=2.0em] (A14) {$\mathbf{A}_1$};
		\node [on chain=4,xtape,minimum width=4.0em] {};
		\node [on chain=4,btape,minimum width=3.5em] (B14) {$\mathbf{B}_1$};
		\node [on chain=4,xtape,minimum width=2.5em] {};
		\node [on chain=4,atape,minimum width=2.0em] (A24) {$\mathbf{A}_2$};
		\node [on chain=4,xtape,minimum width=4.0em] {};
		\node [on chain=4,btape,minimum width=3.5em] (B24) {$\mathbf{B}_2$};
		\node [on chain=4,xtape,minimum width=2.5em] {};
		\node [on chain=4,dtape,minimum width=4.0em] (dots4) {$\ldots$};
		\node [on chain=4,atape,minimum width=2.0em] (An4) {$\mathbf{A}_n$};
		\node [on chain=4,xtape,minimum width=4.0em] {};
		\node [on chain=4,btape,minimum width=3.5em] (Bn4) {$\mathbf{B}_n$};
		\node [on chain=4,xtape,minimum width=2.5em] {};
		\node [on chain=4,ntape,minimum width=1.0em] (end4) {};
		
		\node [on chain=5,ntape,minimum width=1.0em,below=\vertspace of tape4] (tape5) {};
		\node [on chain=5,atape,minimum width=2.0em] (A15) {$\mathbf{A}_1$};
		\node [on chain=5,btape,minimum width=3.5em] (B15) {$\mathbf{B}_1$};
		\node [on chain=5,xtape,minimum width=0.5em] {};
		\node [on chain=5,xtape,minimum width=6.0em] {};
		\node [on chain=5,atape,minimum width=2.0em] (A25) {$\mathbf{A}_2$};
		\node [on chain=5,btape,minimum width=3.5em] (B25) {$\mathbf{B}_2$};
		\node [on chain=5,xtape,minimum width=0.5em] {};
		\node [on chain=5,xtape,minimum width=6.0em] {};
		\node [on chain=5,dtape,minimum width=4.0em] {$\ldots$};
		\node [on chain=5,atape,minimum width=2.0em] (An5) {$\mathbf{A}_n$};
		\node [on chain=5,btape,minimum width=3.5em] (Bn5) {$\mathbf{B}_n$};
		\node [on chain=5,xtape,minimum width=0.5em] {};
		\node [on chain=5,xtape,minimum width=6.0em] {};
		\node [on chain=5,ntape,minimum width=1.0em] (end5) {};
		
		\node [on chain=6,ntape,minimum width=1.0em,below=\vertspace of tape5] (tape6) {};
		\node [on chain=6,atape,minimum width=2.0em] (A16) {$\mathbf{A}_1$};
		\node [on chain=6,btape,minimum width=3.5em] (B16) {$\mathbf{B}_1$};
		\node [on chain=6,xtape,minimum width=0.5em] {};
		\node [on chain=6,atape,minimum width=2.0em] (A26) {$\mathbf{A}_2$};
		\node [on chain=6,btape,minimum width=3.5em] (B26) {$\mathbf{B}_2$};
		\node [on chain=6,xtape,minimum width=0.5em] {};
		\node [on chain=6,dtape,minimum width=4.0em] {$\ldots$};
		\node [on chain=6,atape,minimum width=2.0em] (An6) {$\mathbf{A}_n$};
		\node [on chain=6,btape,minimum width=3.5em] (Bn6) {$\mathbf{B}_n$};
		\node [on chain=6,xtape,minimum width=0.5em] {};
		\node [on chain=6,xtape,minimum width=18.0em] {\textcolor{gray}{\textit{pool}}};
		\node [on chain=6,ntape,minimum width=1.0em] (end6) {};
		
		\draw (B11.south west) edge[decorate,decoration={brace,amplitude=5pt,mirror,raise=1mm}] node[below=3mm] {\Call{rotate}{} pool in} (end1.south west);
		\draw (A12.south west) edge[decorate,decoration={brace,amplitude=5pt,mirror,raise=1mm}] node[below=3mm] {\Call{simple-interleave}{}} (pool2.south east);
		\draw (A13.south west) edge[decorate,decoration={brace,amplitude=5pt,mirror,raise=1mm}] node[below=3mm] {\Call{simple-interleave}{}} (end3.south west);
		
		\draw (A14.south east) edge[decorate,decoration={brace,amplitude=5pt,mirror,raise=1mm}] node[below=3mm] {parallel \Call{rotate}{}} (A24.south west);
		\draw (A24.south east) edge[decorate,decoration={brace,amplitude=5pt,mirror,raise=1mm}] node[below=3mm] {''} (dots4.south west);
		\draw (An4.south east) edge[decorate,decoration={brace,amplitude=5pt,mirror,raise=1mm}] node[below=3mm] {''} (end4.south west);
		
		\draw (A15.south west) edge[decorate,decoration={brace,amplitude=5pt,mirror,raise=1mm}] node[below=3mm] {$\Call{simple-interleave}{}^{-1}$} (end5.south west);
	\end{tikzpicture}
	\caption{The mechanics of $\textsc{interleave}$ in terms of easier subroutines}
    \label{fig:interleave}
\end{figure}
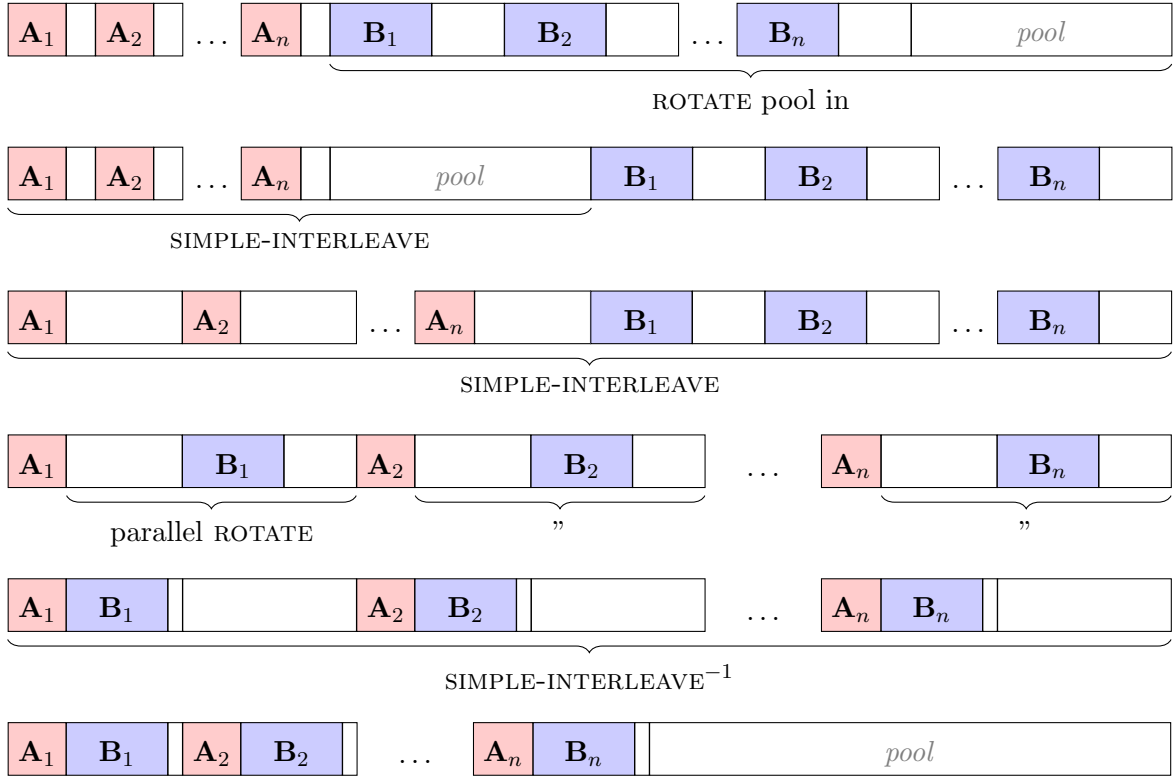

\begin{thm}
    \label{thm:power_interleave_full}
	There exists a circuit which interleaves blocks $A_1, \ldots, A_n$ of length $2^{k_A}$ with blocks $B_1, \ldots, B_n$ of length $2^{k_B}$, producing blocks $A_1 B_1, \ldots, A_n B_n$ of length $2^{k_{AB}}$. See Figure~\ref{fig:blockInterleaving} for the intended result. 
\end{thm}
\begin{proof}
    See Figure~\ref{fig:interleave} for a worked example of this procedure. 
    
    To start, if the input blocks are not the same length, i.e., $2^{k_A} \neq 2^{k_B}$, then add padding to the smaller of the two from the pool. Doubling the block length can be achieved with \Call{simple-interleave}{} above, and multiple applications naturally compose since \Call{simple-interleave}{} is achieved by a bit permutation. 
	
	Once the $A$ and $B$ blocks are the same (power of $2$) length, apply \Call{simple-interleave}{} to interleave them. The result is that $A_i$ and $B_i$ are separated by $A_i$'s padding, so the next step is to rotate $B_i$ into position, in parallel over all $i$. Parallelism is possible here because the stride between rotations is a power of $2$. 
	
	Finally, the padding of $A$ combined with the padding of $B$ may be more than we need. An inverse \Call{simple-interleave}{} step can remove that padding and return it to a pool at the end. 
\end{proof}

\begin{figure}
	\centering

	\begin{tikzpicture}[start chain=1 going right, node distance=-0.1mm,
		start chain=2 going right,
		every on chain/.append style={text height=1em}]
		\edef\sizetape{0.7cm}
		\edef\vertspace{2.5cm}
		\tikzstyle{tape}=[draw,minimum size=\sizetape]
		
		\node [on chain=1,draw=none,minimum width=1em] (tape1) {};
		\node [on chain=1,tape,minimum width=2em,fill=red!20] (A1) {$\mathbf{A}_1$};
		\node [on chain=1,tape,minimum width=1em] {};
		\node [on chain=1,tape,minimum width=2em,fill=red!20] (A2) {$\mathbf{A}_2$};
		\node [on chain=1,tape,minimum width=1em] {};
		\node [on chain=1,tape,minimum width=2em,dashed] {$\ldots$};
		\node [on chain=1,tape,minimum width=2em,fill=red!20] (An) {$\mathbf{A}_n$};
		\node [on chain=1,tape,minimum width=1em] {};
		\node [on chain=1,tape,minimum width=4.5em,fill=blue!20] (B1) {$\mathbf{B}_1$};
		\node [on chain=1,tape,minimum width=1.5em] {};
		\node [on chain=1,tape,minimum width=4.5em,fill=blue!20] (B2) {$\mathbf{B}_2$};
		\node [on chain=1,tape,minimum width=1.5em] {};
		\node [on chain=1,tape,minimum width=2em,dashed] {$\ldots$};
		\node [on chain=1,tape,minimum width=4.5em,fill=blue!20] (Bn) {$\mathbf{B}_n$};
		\node [on chain=1,tape,minimum width=1.5em] {};
		\node [on chain=1,tape,minimum width=7em] {\textcolor{gray}{\textit{pool}}};
		\node [on chain=1,draw=none,minimum width=1em] (end1) {};
		
		\node [on chain=2,draw=none,minimum width=1em,below=\vertspace of tape1] (tape2) {};
		\node [on chain=2,tape,minimum width=2em,fill=red!20] (A1prime) {$\mathbf{A}_1$};
		\node [on chain=2,tape,minimum width=4.5em,fill=blue!20] (B1prime) {$\mathbf{B}_1$};
		\node [on chain=2,tape,minimum width=5.5em] {};
		\node [on chain=2,tape,minimum width=2em,fill=red!20] (A2prime) {$\mathbf{A}_2$};
		\node [on chain=2,tape,minimum width=4.5em,fill=blue!20] (B2prime) {$\mathbf{B}_2$};
		\node [on chain=2,tape,minimum width=5.5em] {};
		\node [on chain=2,tape,minimum width=2em,dashed] {$\ldots$};
		\node [on chain=2,tape,minimum width=2em,fill=red!20] (Anprime) {$\mathbf{A}_n$};
		\node [on chain=2,tape,minimum width=4.5em,fill=blue!20] (Bnprime) {$\mathbf{B}_n$};
		\node [on chain=2,tape,minimum width=5.5em] {};
		\node [on chain=2,draw=none,minimum width=1em] (end2) {};
		
		\draw (A1.south) edge[dotted,->,>=stealth,in=90,out=270] (A1prime.north);
		\draw (A2.south) edge[dotted,->,>=stealth,in=90,out=270] (A2prime.north);
		\draw (An.south) edge[dotted,->,>=stealth,in=90,out=270] (Anprime.north);
		\draw (B1.south) edge[dotted,->,>=stealth,in=90,out=270] (B1prime.north);
		\draw (B2.south) edge[dotted,->,>=stealth,in=90,out=270] (B2prime.north);
		\draw (Bn.south) edge[dotted,->,>=stealth,in=90,out=270] (Bnprime.north);
		
		\draw (A1.north west) edge[decorate,decoration={brace,amplitude=5pt,raise=1mm}] node[above=2mm] {$2^{k_A}$} (A2.north west);
		\draw (B1.north west) edge[decorate,decoration={brace,amplitude=5pt,raise=1mm}] node[above=2mm] {$2^{k_B}$} (B2.north west);
		\draw (A1prime.south west) edge[decorate,decoration={brace,amplitude=5pt,mirror,raise=1mm}] node[below=2mm] {$2^{k_{AB}}$} (A2prime.south west);
		
	\end{tikzpicture}
	
	\caption{Schematic of the interleaving operation.}
	\label{fig:blockInterleaving}
\end{figure}
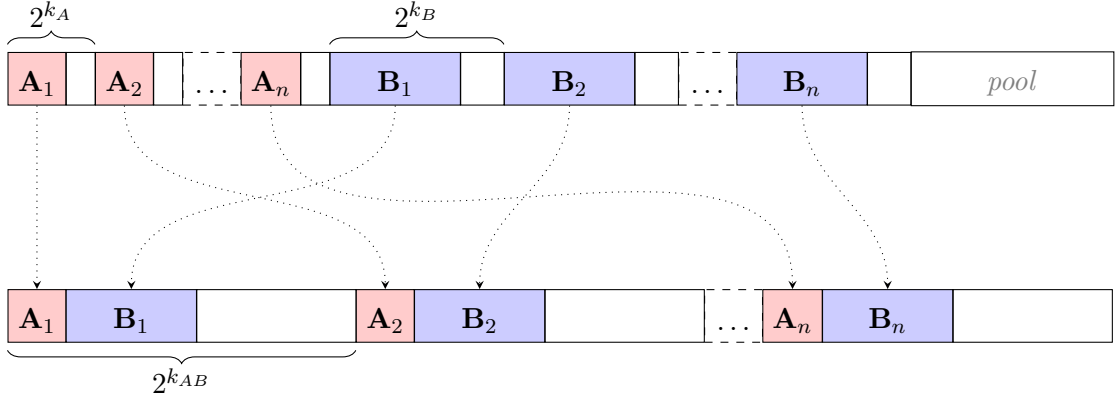

\begin{figure}
    \centering
    \begin{tikzpicture}[
    level 1/.style={sibling distance=2.75cm, level distance=2cm},
    every node/.style={circle,solid, draw=black,thin, minimum size = 0.5cm},
    emph/.style={edge from parent/.style={dashed,red,very thick,draw}},
    norm/.style={edge from parent/.style={solid,black,thin,draw}}]

    \tikzstyle{bits} = [rectangle, draw, minimum width = 0.75cm, minimum height = 0.9cm]
    \tikzstyle{register1} = [rectangle, draw, minimum width = 3.5cm, minimum height = 0.9cm]
    \tikzstyle{register2} = [rectangle, draw, minimum width = 2.75cm, minimum height = 0.9cm]
    \tikzstyle{register3} = [rectangle, draw, minimum width = 2cm, minimum height = 0.9cm]
    \tikzstyle{register4} = [rectangle, draw, minimum width = 1.25cm, minimum height = 0.9cm]

    \draw (0.25,0) node[register1] (root) {$\sstrings{M}{F}$}
      child[dotted] { node[register2] (l) {$0\sstrings{M-1}{F}$}
      }
      child[dotted] { node[register2] (r) {$1\sstrings{M-1}{F-1}$}
      };
      \draw [decoration={calligraphic brace, amplitude=5pt}, decorate, line width=1.25pt]
  ($(r.south east)-(0,0.1)$) -- ($(l.south west)-(0,0.1)$) node[below=0.2cm,midway,rectangle,draw=none] {\textsf{DEINTERLEAVE}} ;

      \path [<->] ($(r.north)+(0,0.1)$) edge [bend right=20]  node[above,midway,rectangle,draw=none] {\textsf{SWAP}} ($(l.north)+(0,0.1)$);

  \begin{scope}[on background layer]
      \draw[draw=red!20, fill=red!20] ($(l.north west)$)  rectangle ($(l.south east)$);
      \draw[draw=blue!20, fill=blue!20] ($(r.north west)$)  rectangle ($(r.south east)$);
    \end{scope}

    \draw (7.75,0) node[register1, draw=none] (root2) {}
      child[dotted] { node[register2] (l2) {$1\sstrings{M-1}{F-1}$} edge from parent[draw=none]
      }
      child[dotted] { node[register2] (r2) {$0\sstrings{M-1}{F}$} edge from parent[draw=none]
      };

      \draw [|->] ($(r.east) + (0.2, 0)$) -- ($(l2.west) + (-0.2, 0)$) node[above=0.1cm,midway,rectangle,draw=none] {$\pi_{0\mapsto 1}$};
      
  \begin{scope}[on background layer]
      \draw[draw=red!20, fill=red!20] ($(r2.north west)$)  rectangle ($(r2.south east)$);
      \draw[draw=blue!20, fill=blue!20] ($(l2.north west)$)  rectangle ($(l2.south east)$);
    \end{scope}
    \end{tikzpicture}
    
    \hrulefill\par\textcolor{white}{.}\par
    \begin{tikzpicture}[
    level 1/.style={sibling distance=2cm, level distance=2cm},
    every node/.style={circle,solid, draw=black,thin, minimum size = 0.5cm},
    emph/.style={edge from parent/.style={dashed,red,very thick,draw}},
    norm/.style={edge from parent/.style={solid,black,thin,draw}}]

    \tikzstyle{bits} = [rectangle, draw, minimum width = 0.75cm, minimum height = 0.9cm]
    \tikzstyle{register1} = [rectangle, draw, minimum width = 2.75cm, minimum height = 0.9cm]
    \tikzstyle{register2} = [rectangle, draw, minimum width = 2cm, minimum height = 0.9cm]
    \tikzstyle{register3} = [rectangle, draw, minimum width = 2cm, minimum height = 0.9cm]
    \tikzstyle{register4} = [rectangle, draw, minimum width = 1.25cm, minimum height = 0.9cm]

    \draw (0.25,6) node[register1] (root) {$11\sstrings{M-2}{F-2}$}
      child[dotted] { node[register2] (l) {$110\sstrings{M-3}{F-2}$}
      }
      child[dotted] { node[register2] (r) {$111\sstrings{M-3}{F-3}$}
      };
      \draw [decoration={calligraphic brace, amplitude=5pt}, decorate, line width=1.25pt]
  ($(r.south east)-(0,0.1)$) -- ($(l.south west)-(0,0.1)$) node[below=0.2cm,midway,rectangle,draw=none] {\textsf{DEINTERLEAVE}} ;

      \path [<->] ($(r.north)+(0,0.1)$) edge [bend right=20]  node[above,midway,rectangle,draw=none] {\textsf{SWAP}} ($(l.north)+(0,0.1)$);

  \begin{scope}[on background layer]
      \draw[draw=blue!20, fill=blue!20] ($(l.north west)$)  rectangle ($(l.south east)$);
      \draw[draw=ForestGreen!20, fill=ForestGreen!20] ($(r.north west)$)  rectangle ($(r.south east)$);
    \end{scope}
    
    \draw (4.25,6) node[register1] (root2) {$10\sstrings{M-2}{F-1}$}
      child[dotted] { node[register2] (l2) {$100\sstrings{M-3}{F-1}$}
      }
      child[dotted] { node[register2] (r2) {$101\sstrings{M-3}{F-2}$}
      };

      \path [<->] ($(r2.north)+(0,0.1)$) edge [bend right=20]  node[above,midway,rectangle,draw=none] {\textsf{SWAP}} ($(l2.north)+(0,0.1)$);

      \begin{scope}[on background layer]
          \draw[draw=red!20, fill=red!20] ($(l2.north west)$)  rectangle ($(l2.south east)$);
          \draw[draw=blue!20, fill=blue!20] ($(r2.north west)$)  rectangle ($(r2.south east)$);
      \end{scope}

      \draw (8.25,6) node[register1] (root3) {$01\sstrings{M-2}{F-1}$}
      child[dotted] { node[register2] (l3) {$010\sstrings{M-3}{F-1}$}
      }
      child[dotted] { node[register2] (r3) {$011\sstrings{M-3}{F-2}$}
      };
      \draw [decoration={calligraphic brace, amplitude=5pt}, decorate, line width=1.25pt]
  ($(r3.south east)-(0,0.1)$) -- ($(l2.south west)-(0,0.1)$) node[below=0.2cm,midway,rectangle,draw=none] {\textsf{DEINTERLEAVE}} ;

      \path [<->] ($(r3.north)+(0,0.1)$) edge [bend right=20]  node[above,midway,rectangle,draw=none] {\textsf{SWAP}} ($(l3.north)+(0,0.1)$);

  \begin{scope}[on background layer]
      \draw[draw=red!20, fill=red!20] ($(l3.north west)$)  rectangle ($(l3.south east)$);
      \draw[draw=blue!20, fill=blue!20] ($(r3.north west)$)  rectangle ($(r3.south east)$);
    \end{scope}
    
      \draw (12.25,6) node[register1] (root4) {$00\sstrings{M-2}{F}$}
      child[dotted] { node[register2] (l4) {$000\sstrings{M-3}{F}$}
      }
      child[dotted] { node[register2] (r4) {$001\sstrings{M-3}{F-1}$}
      };
      \draw [decoration={calligraphic brace, amplitude=5pt}, decorate, line width=1.25pt]
  ($(r4.south east)-(0,0.1)$) -- ($(l4.south west)-(0,0.1)$) node[below=0.2cm,midway,rectangle,draw=none] {\textsf{DEINTERLEAVE}} ;

      \path [<->] ($(r4.north)+(0,0.1)$) edge [bend right=20]  node[above,midway,rectangle,draw=none] {\textsf{SWAP}} ($(l4.north)+(0,0.1)$);
    
  \begin{scope}[on background layer]
      \draw[draw=black!20, fill=black!20] ($(l4.north west)$)  rectangle ($(l4.south east)$);
      \draw[draw=red!20, fill=red!20] ($(r4.north west)$)  rectangle ($(r4.south east)$);
    \end{scope}

      \draw [|->] ($(6.25, 2.5)$) -- ($(6.25, 1.7)$) node[right=0.1cm,midway,rectangle,draw=none] {$\pi_{2 \mapsto 3}$};

    \draw (0.25, 3) node[register1, draw=none] (broot) {}
      child[dotted] { node[register2] (bl) {$111\sstrings{M-3}{F-3}$} edge from parent[draw=none]
      }
      child[dotted] { node[register2] (br) {$110\sstrings{M-3}{F-2}$} edge from parent[draw=none]
      };

    \draw (4.25, 3) node[register1, draw=none] (broot2) {}
      child[dotted] { node[register2] (bl2) {$101\sstrings{M-3}{F-2}$} edge from parent[draw=none]
      }
      child[dotted] { node[register2] (br2) {$011\sstrings{M-3}{F-2}$} edge from parent[draw=none]
      };
      
    \draw (8.25, 3) node[register1, draw=none] (broot3) {}
      child[dotted] { node[register2] (bl3) {$100\sstrings{M-3}{F-1}$} edge from parent[draw=none]
      }
      child[dotted] { node[register2] (br3) {$010\sstrings{M-3}{F-1}$} edge from parent[draw=none]
      };
      
    \draw (12.25, 3) node[register1, draw=none] (broot4) {}
      child[dotted] { node[register2] (bl4) {$001\sstrings{M-3}{F-1}$} edge from parent[draw=none]
      }
      child[dotted] { node[register2] (br4) {$000\sstrings{M-3}{F}$} edge from parent[draw=none]
      };

  \begin{scope}[on background layer]
      \draw[draw=ForestGreen!20, fill=ForestGreen!20] ($(bl.north west)$)  rectangle ($(bl.south east)$);
      \draw[draw=blue!20, fill=blue!20] ($(br.north west)$)  rectangle ($(br2.south east)$);
      \draw[draw=red!20, fill=red!20] ($(bl3.north west)$)  rectangle ($(bl4.south east)$);
      \draw[draw=black!20, fill=black!20] ($(br4.north west)$)  rectangle ($(br4.south east)$);
    \end{scope}
    \end{tikzpicture}
    \caption{Depiction of the permutations $\pi_{0\mapsto 1}$ and $\pi_{2\mapsto 3}$ (padding not depicted). Rectangles denote chapters, and books are color-coded. The initial list is illustrated in two layers, with the first matching the decomposition in Equation \ref{eqn:arrayDecompDirect} and the second matching the decomposition in Equation \ref{eqn:arrayDecompFactored} (though both represent the pre-permutation ordering). The arrows indicate where the \textsf{SWAP} operations occur, and the brackets indicate the \textsf{DEINTERLEAVE} operations---alternating shades denote the blocks to be deinterleaved. The final array on the opposite side of the $\mapsto$ arrow indicates the array after the permutation, as in Equation \ref{eqn:arrayDecompFinal}.}
    \label{fig:listITerationsExample3}
\end{figure}
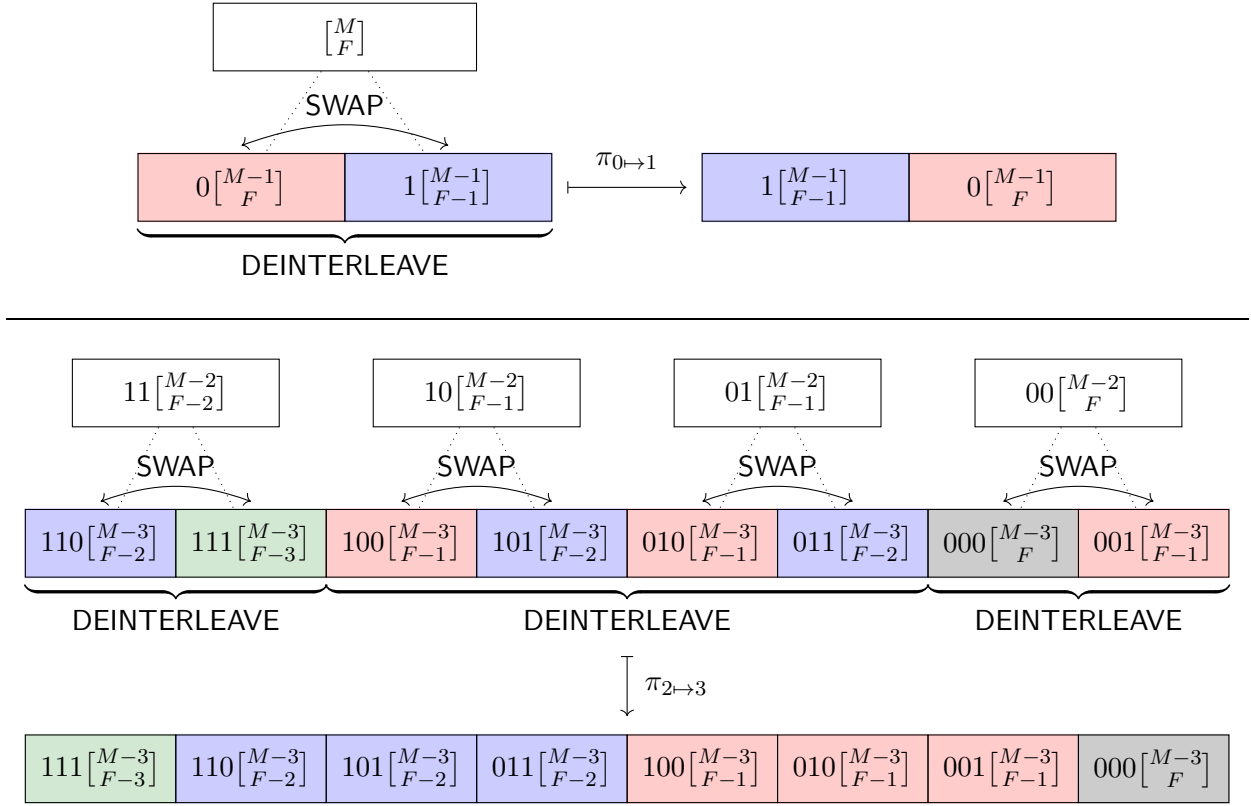

\section{Applications}
\label{sec:applications}

In this section we give a more detailed comparison of our encoding to prior encodings, and outline scenarios where ours outperforms prior work. We will focus on the succinct construction, e.g. in Section \ref{sec:fullEnc}, as it is likely the more practical of the two constructions.

For concreteness we consider the parameter regime where the number of fermions is polynomially smaller than the number of modes, i.e. $F=M^c$ for some constant $0<c<1$, though we note that our encoding is efficient outside this regime as well. This regime is well motivated for applications such as quantum chemistry, where the number of modes is preferred to be as large as possible to accurately model the continuum, but algorithm's polynomial scaling with $M$ prevents $M$ from being exponentially larger than $F$.

The primary application we find is in randomized simulation algorithms such as qDRIFT \cite{C19}, where our encoding gives a polynomial space savings with only logarithmic depth overhead. We also show that our encoding allows for an optimal implementation of a second quantized \textsc{SELECT} subroutine, used in linear combination of unitary methods \cite{CW12}, no matter how $M$ and $F$ relate.

\subsection{Space Efficient Randomized Simulation}
Randomized product formulas like the qDRIFT algorithm \cite{C19} are strong candidates where our encoding will be useful, as the complexity of such algorithms depend heavily on the fermion encoding used. The qDRIFT algorithm requires $O((\lambda t)^2 / \epsilon)$ fermionic rotation operators to evolve for time $t$ within error $\epsilon$, where $\lambda$ is the sum of the magnitude of Hamiltonian coefficients when written in a second quantized basis. The scaling with $\lambda$ is preferable to scaling with $M$ when the magnitude of coefficients varies drastically, as is often the case for quantum chemistry Hamiltonians. 
\begin{thm}
    \label{thm:qdrift-algorithm}
    For a particle preserving local fermionic Hamiltonian $H$ on an $M$ mode system, there is a second quantized algorithm for evolving an $F$ fermion state to time $t$ using space $\mathcal{I}+O(F)$, with gate complexity $O\left(\epsilon^{-1} (\lambda t)^2 \mathcal{I}\right)$,
    and circuit depth $O\left(\epsilon^{-1} (\lambda t)^2 \log M \log \log M \right)$, where $\lambda$ denotes the sum of magnitudes of coefficients in $H$ and $\epsilon$ denotes the target diamond-norm error.
\end{thm}
This theorem follows straightforwardly from using the qDRIFT algorithm \cite{C19} with our encoding in Section \ref{sec:fullEnc}. In the regime discussed above, the space complexity of this algorithm is  \begin{align*}
    \mathcal{I}+O(F) = O(M^c \log M^{1-c}) = O(M^{c} \log M),
\end{align*}
with gate complexity \begin{align*}
    O\left(\epsilon^{-1} (\lambda t)^2 \mathcal{I}\right) = O(\epsilon^{-1} (\lambda t)^2 M^{c} \log M),
\end{align*}
and circuit depth \begin{align*}
    O\left(\epsilon^{-1} (\lambda t)^2 \log M \log \log M \right).
\end{align*}

We now compare these complexities against those achieved by other second-quantized encodings in our parameter regime.

\paragraph{Space Complexity.} The most space-efficient prior second quantized encodings in this regime that are amenable to randomized simulation methods are the qubit efficient \cite{STHCG22} and permutation basis \cite{HNAW22}, but these encodings would incur an exponential ($M^{F}$) gate and depth overhead. Barring these encodings, the next best are the segment code \cite{SW18, S19}, and the optimal degree code \cite{KFHM22}. We use polynomially less space than the former ($O(M^{c} \log M)$ vs.\ $\Omega(M)$) and quadratically less than the latter ($O(F \log M)$ vs.\ $\Omega(F^2 \log^4 M)$).

\paragraph{Circuit Complexity} When $c < 1/2$, the prior best second quantized encoding in terms of space usage (barring codes with exponential gate overhead) is the optimal degree code \cite{KFHM22}. Compared to this encoding, along with better space efficiencies we also improve on the $O(F^2 \log^5 M)$ gate complexity of fermion operations. Our scaling of $O(\mathcal{I})$ (with $\mathcal{I} < F \log M$) is always quadratically better in dependence on $F$, and at least four powers better in dependence on $\log M$. Further, while the depth is not analyzed in the optimal degree encoding a simple counting argument implies that it is at least $\Omega((\lambda t)^2 \log M / \epsilon)$, which is at most a negligible $O(\log \log M)$ better than our encoding. Compared to other space efficient second-quantized encodings, we give at least a quadratic improvement in circuit complexity while also removing many log factors.

Allowing space inefficient encodings, the most gate efficient prior encoding in our regime is the Bravyi-Kitaev encoding \cite{BK02}, which uses $O(\log M)$ gates to implement a fermionic rotation. This is the encoding used in the original resource estimates for the qDRIFT algorithm, making it a natural competitor. When the order of fermionic rotations is selected randomly, the Bravyi-Kitaev encoding of these operators cannot be significantly parallelized. The depth for full simulation will generically be $\Omega((\lambda t)^2 / \epsilon)$ (and potentially larger if fermion operations take super-constant depth to implement), which is only a logarithmic factor lower than in our encoding. Furthermore, the Bravyi-Kitaev encoding requires polynomially more qubits ($O(M)$ versus $O(M^c \log M)$) than our encoding in this regime---though fewer gates ($O(\log M)$ versus $O(M^c \log M)$). This tradeoff is therefore more beneficial the smaller $c$ is.

\subsection{Optimal SELECT Subroutine}
In the context of post-Trotter methods, our encoding method can be relevant not only for saving space but also for saving gates. We illustrate this by giving a provably optimal implementation of the \textsc{SELECT} subroutine, a key component in the linear combination of unitaries algorithm for Hamiltonian simulation \cite{CW12}. This algorithm involves expressing the Hamiltonian as a linear combination of unitaries \begin{align*}
    H=\alpha_1 U_1 + ... + \alpha_q U_q,
\end{align*}
then block encoding the Hamiltonian using a \text{PREPARE} oracle that creates a superposition over every number in $j \in[q]$, weighted by amplitude $\alpha_j$, in conjunction with as a so-called \textsc{SELECT} oracle. The \textsc{SELECT} oracle has a control wire indicating a number $j \in [q]$, and conditioned on $j$ performs the unitary $U_j$ on the physical state---this oracle is the only part whose implementation is dependent on the fermion encoding used. One way to split a $k$-local second quantized fermionic Hamiltonian into a linear combination of unitaries is to split every $a_j$ into Majoranas $a_j=\gamma_{2j-1} + i \gamma_{2j}$. The $\gamma_j$ operators are unitary, meaning the full Hamiltonian is now a linear combination of unitaries. Using our encoding, we can perform a \textsc{SELECT} oracle in the optimal complexity. Specifically, we show how to perform an $k$-local product of Majorana operators, given $k$ index wires denoting which Majorana's appear in the product. This translates to a sequence of $O(k)$ many $\srank$ and $\bflip$ operations, so in particular it suffices to do a single $\srank(j, \bb)$ or $\bflip(j, \bb)$ operator controlled on an index wire $\ket{j}$.

\begin{thm}
    \label{thm:select}
    Using the encoding described in Section \ref{sec:fullEnc}, a coherently controlled $\srank(\ket{j}, \bb)$ or $\bflip(\ket{j}, \bb)$ can be performed with $O(\mathcal{I})$ gates.
\end{thm}

\begin{proof}
    Consider the $M$ fermion Hamiltonian \begin{align*}
        H^{(m)} =& \alpha_1 U_1^{(m)} + ... + \alpha_q U_q^{(m)}
    \end{align*}
    where every $U_i^{(m)}$ is a product of Majorana operators. When the fermionic Hamiltonian $H$ is $k$-local, then every $U_j^{(m)}$ appearing in $H^{(m)}$ has $k$-many Majorana operators. It follows that we can implement the action of $U_j^{(m)}$ within the encoding described in Section \ref{sec:fullEnc} with complexity $O(\mathcal{I})$. From the discussion in Section \ref{app:coherentlyControlled}, we can further coherently condition which $U_j^{(m)}$ is performed using a control register, introducing at most a constant factor overhead. For this \textsc{SELECT} oracle we therefore obtain circuits of size $O(\mathcal{I})$. Note that in general there are many ways to split a Hamiltonian into a linear combination of unitaries, and we leave to future work whether \textsc{SELECT} has an efficient implementation under alternative choices.
\end{proof}

As argued above, this is sufficient to implement a \textsf{SELECT} oracle in the same complexity. We now turn to lower bounds.

\begin{thm}
    \label{thm:select-lower-bound}
    Using any fermion to qubit mapping, a coherently controlled, second-quantized, $k$-local particle preserving fermion operation (a \textsc{SELECT} oracle) requires $\Omega(\mathcal{I})$ gates.
\end{thm}
\begin{proof}
    The lower bound for representing $F$ fermions in $M$ modes is $\mathcal{I}$ qubits, due to a straightforward counting argument. Suppose there were a generic implementation of the second quantized \textsc{SELECT} oracle of the aforementioned kind which used $\lfloor \frac{1}{2}(\mathcal{I}-1)\rfloor$ one- and two-qubit gates, and implemented any $k$-local fermion term for a $k > 1$. It follows that such an oracle would act on at most $\mathcal{I}-1$ qubits, so consider keeping only the $\mathcal{I}-1$ qubits it acts on. We are then left with a fermion encoding on $\mathcal{I}-1$ qubits, as by assumption the \textsc{SELECT} oracle is capable of performing arbitrary pairs of Majorana operators; a contradiction. Hence, any generic implementation of the \textsc{SELECT} oracle requires $\Omega(\mathcal{I})$ gates, matching our upper bound up to constant factors.
\end{proof}

To compare this result against existing second quantized encodings, note that in all prior encodings (summarized in Table \ref{tab:fewFermionEncTable}) the complexity of a coherently controlled fermion operator is at least linear in the number of qubits, or exponential in the case of prior succinct encodings. For every gate-efficient encoding, the number of qubits is polynomially larger than $\mathcal{I}$ in our regime $F=M^c$, from which it follows that the gate complexity is at least polynomially larger than $\mathcal{I}$. Concretely, when $c<1/2$ the best prior implementation of the select oracle is the optimal degree encoding, with gate complexity $\Omega(F^2 \log^5 M)$. Our implementation of complexity $O(\mathcal{I})$ where $\mathcal{I} < F \log M$ is quadratically better in $F$ dependence, and at least four powers better in $\log M$ dependence. We note also that in the parameter regime $c>1/2$ our encoding maintains at least a factor $\tilde O(M^{1-c})$ advantage in gate complexity over prior work, giving a polynomial advantage over the whole regime.

\subsection{Comparison to First Quantized Encodings}
\label{subsec:compareFirstQuant}

First quantized algorithms maintain a list of pointers representing fermion positions similar to the encoding described in Section \ref{sec:efficientEnc}, but antisymmetrize the state over all possible permutations. The fermion phases then arise from the exchange statistics of the registers, allowing for a different class of efficient operations. Comparisons against first quantization are subtle, as algorithms within these two paradigms can have drastically different complexities for simulating similar systems. In general one would expect first quantized algorithms to outperform second quantized algorithms when $F$ is extremely small compared to $M$, but there is evidence that second quantized algorithms can remain advantageous when $F$ is polynomially smaller than $M$ \cite{Su2021nearlytight}. This motivates the parameter regime we consider here.
\paragraph{Space Complexity.} Using Stirling's approximation we can rewrite \begin{align}
    \mathcal{I} =& F \log \frac{M}{F} + O(F).
\end{align}
In our parameter regime of interest, we therefore have \begin{align}
    \mathcal{I} =& (1-c) \cdot F \log M + O(F) \label{eqn:succinctBenchmark}
\end{align}
Note that the $O(F)$ term is asymptotically less than the $(1-c) \cdot F \log M$ term, and as a corollary the asymptotic space usage of our encoding ($\mathcal{I}+O(F)$) is a factor $1+o(1)$ from optimal, i.e. it is succinct. Standard first-quantized encodings \cite{SBWRB21} use $F \lceil\log M\rceil$, which is asymptotically a factor $\frac{1}{1-c}$ above $\mathcal{I}$, which is not succinct. Our encodings saves more space the closer $c$ is to $1$.

\paragraph{Gate Complexity.} Both our representation and first-quantized representations support efficient (i.e. at most linear in the number of qubits) fermion operations of the first-quantized or second-quantized variety, respectively. Differences in gate complexity, therefore, stem primarily from differences in algorithms built in a first-quantized versus second-quantized picture. Many of the algorithms in these different paradigms are designed for different contexts, making a fair comparison difficult.

\section{Conclusion and Open Problems}

We presented two second-quantized fermion encodings which are almost exactly optimal in terms of space usage, yet allows gate efficient and low depth second-quantized fermion rotations. In the context of randomized simulation algorithms, the succinct encoding allows for a polynomial space savings with only a logarithmic depth overhead. Along the way, we presented a more simple succinct encoding which did not achieve low depth, as well as a low depth encoding that did not achieve succinctness, though both of these encodings may be of independent interest due to their simplicity. This work leaves open the following questions.

\begin{enumerate}[label=(\arabic*)]
    \item In first quantization, a list of positions is stored in an antisymmetrized wavefunction with $F \lceil \log M \rceil$ qubits in standard encodings, allowing for efficient first quantized fermion operations. Can such a representation be made succinct while still antisymmetrizing?
    \item Is there an efficient fermionic Fourier transform circuit for our encodings?
    \item Can the qubit count be improved beyond $\mathcal{I}+O(F)$ without exceeding gate complexity $O(\mathcal{I})$? We conjecture a lower bound on the qubits required to maintain this gate complexity is $\mathcal{I}+\Omega(F)$.

    \item When compiled to a universal gate set like $\mathrm{CNOT}, \mathrm{H}, T$, the complexities of each gate is the same as the overall complexity in our encodings. In many error correcting codes certain gates are harder than others (for instance $T$ gates are more expensive in a surface code), motivating analysis of each gate individually. Can the number of $T$ gates be reduced in this construction?
    \item Can the degree of the $\textsf{poly}(M)$ circuit complexity scaling in our implicit encoding be improved?
\end{enumerate}

\section*{Acknowledgements}
The authors thank James Watson for helpful conversation about fermion encodings, and Andrew Childs for providing feedback on an early draft of this manuscript. JC is supported by the US Department of Energy grant no.\ DESC0020264. LS thanks the Joint Center for Quantum Information and Computer Science (QuICS) where a large portion of this research occurred. 

\pagebreak

\bibliography{ref}
\pagebreak
\appendix
\section{Comparison Circuits}
The two comparison operations we will need are equality and order. 
\begin{lem}
    Equality between two $n$-bit registers $A, B$ can be reversibly computed with $O(n)$ gates, with either $O(1)$ ancillae and depth $O(n)$, or $O(n)$ ancillae and depth $O(\log n)$. \label{lem:equalityComplexity}
\end{lem}
\begin{proof}
    Equality between two such registers $A, B$ can be done in place by a bitwise conditional xor from $A$ to $B$, fan-in gate, then uncomputing, as depicted in Figure \ref{fig:simpleEqCirc}. A width $n$ Toffoli can be implemented in depth $O(n)$ using one ancilla \cite{G15}, or alternatively in $O(\log n)$ depth with $O(n)$ ancilla using a simple recursive construction.
\end{proof}
As a corollary, equality with a fixed value can be done with the same depth. To see this, consider in Figure \ref{fig:simpleEqCirc} if $B$ were instead a fixed value. We could then eliminate register $B$, replacing all conditional nots with explicit nots where necessary.

\begin{lem}
    Order comparison $A<B$ between two registers of size $n$ can be reversibly computed with $O(n)$ gates, with either $O(1)$ ancillae and depth $O(n)$, or $O(n)$ ancillae and depth $O(\log n)$.\label{lem:comparisonComplexity}
\end{lem}
\begin{proof}
    A comparison $A<B$ can be done by subtracting $A-B$ via twos-complement, reporting the sign of the result (and considering $0$ the same as positive sign), then uncomputing the sum, as depicted in Figure \ref{fig:simpleLessThanCirc}. The addition can be computed with $O(1)$ ancillae and $O(n)$ depth \cite{CDKM04, WLLQW16} or with $O(n)$ ancillae and $O(\log n)$ depth \cite{DKRS06}, proving the claim.
\end{proof}
Again a corollary of this is that the same complexity holds for comparing one register to a fixed number, where the second register in the addition circuit is replaced by some fixed value.

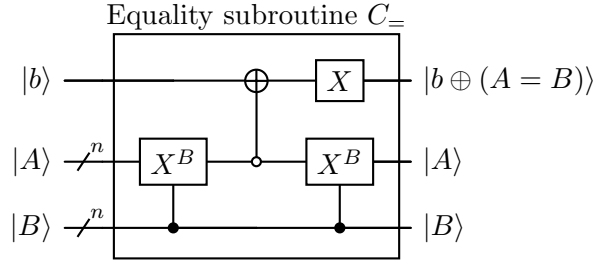
\begin{figure}[H]
    \centering
    \begin{quantikz}
        \lstick{$\ket{b}$} & \qw & \qw \gategroup[3,steps=3,style={inner sep=6pt}]{Equality subroutine $C_=$} &  \targ{} & \gate{X} & \qw \rstick{$\ket{b \oplus (A = B)}$} \\
        \lstick{$\ket{A}$} & \qwbundle[]{n} & \gate{X^B} & \octrl{-1} & \gate{X^B} & \qw \rstick{$\ket{A}$} \\
        \lstick{$\ket{B}$} & \qwbundle[]{n} & \ctrl{-1} & \qw & \ctrl{-1} & \qw \rstick{$\ket{B}$}
    \end{quantikz}
    \caption{Simple equality circuit, where $X^s$ denotes an $X$ on any bit where the corresponding bit in $s$ is $1$; these are transversely controlled in the figure. An open control node on a register indicates a Toffoli with an open control on all bits.}
    \label{fig:simpleEqCirc}
\end{figure}
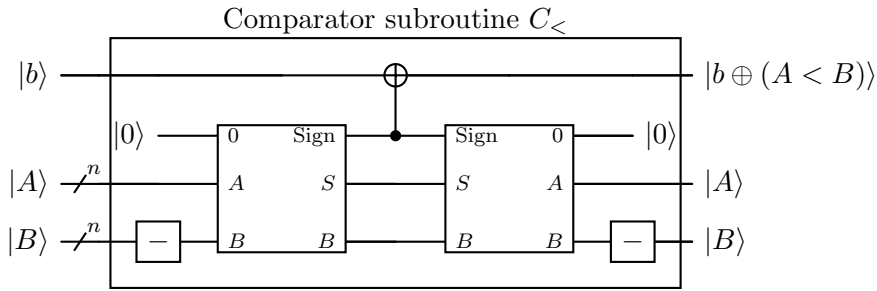
\begin{figure}[H]
    \centering
    \begin{quantikz}
        \lstick{$\ket{b}$} &\qw & \qw \gategroup[4,steps=5,style={inner
sep=6pt}]{Comparator subroutine $C_{<}$}& \qw &  \targ{} & \qw & \qw & \qw \rstick{$\ket{b \oplus (A < B)}$} \\
        \setwiretype{n} & & \lstick{$\ket{0}$} & \gate[3][1.7cm]{}\gateinput{$0$}\gateoutput{Sign}\setwiretype{q} & \ctrl{-1} & \gate[3][1.7cm]{}\gateinput{Sign}\gateoutput{$0$} & \qw \rstick{$\ket{0}$} \\
        \lstick{$\ket{A}$} & \qwbundle[]{n} & \qw &\gateinput{$A$}\gateoutput{$S$} & \qw & \gateinput{$S$}\gateoutput{$A$} & \qw & \qw \rstick{$\ket{A}$} \\
        \lstick{$\ket{B}$} & \qwbundle[]{n} & \gate{-} &\gateinput{$B$}\gateoutput{$B$} & \qw & \gateinput{$B$}\gateoutput{$B$} & \gate{-} & \qw \rstick{$\ket{B}$}
    \end{quantikz}
    \caption{Simple comparison circuit, where ``$-$'' denotes twos-complement negation and the Sign output of the adder (unlabelled box) is $1$ if the result is $<0$, in twos-complement. Wires beginning and ending inside the box are ancillae.}
    \label{fig:simpleLessThanCirc}
\end{figure}
\label{app:compCircs}

\section{Technicalities}

In this appendix, we explain various technicalities related to our encoding. In Appendix \ref{app:ptclePresRotations} we describe how to transform a circuit for applying a Pauli operator to applying the rotation generated by said Pauli operator. In Appendix \ref{app:coherentlyControlled} we outline why our circuit algorithms support coherently controlled operations. Finally, in Appendix \ref{app:fermionOverflow} we describe the behaviour of our encoding when $F$ fermions is exceeded, and how to avoid this possibility.

\subsection{Fermion rotation circuits}
\label{app:ptclePresRotations}
Consider the encoding of a Majorana product $V=\gamma_{i_1}...\gamma_{i_k}$ where $k$ is even and the $i_j$ are distinct. Note that $V$ is hermitian with $V^2=I$ if $k$ is twice an even integer ($k=4,8,etc.$) and otherwise $iV$ is hermitian with $(iV)^2=I$. We give circuits for $V$, which is equivalent to $iV$ up to a global phase; however this phase must be tracked to implement a controlled $iV$, e.g. by performing an $S$ gate on the control wire. We will want to implement either the rotation $\exp(-i\theta V)$ or $\exp(\theta V)$, depending on which is unitary.

Suppose we have a circuit $U$ which satisfies $U^2=I$, which maps to a Majorana product in the way discussed above. We can implement the rotation $\exp(iU\theta)$ using the circuit depicted in Figure~\ref{fig:involRotationCirc}, adapted from Kirby et al.~\cite{KFHM22}. This introduces a constant factor overhead as we need to control the encoded circuits, but we note that in our construction logical operations (comparisons, equality, addition, etc.) do not need to be controlled, as they are always uncomputed; only swaps and register exchanges.
\begin{figure}[H]
    \centering
    \begin{quantikz}
        \setwiretype{n} & \gategroup[2,steps=11,style={inner sep=6pt}]{Subroutine $\exp(iUt)$} & \lstick{$\ket{0}$} & \gate{\mathrm{H}}\setwiretype{q} & \ctrl{1} & \gate{\mathrm{H}} & \gate{R_z(2\theta)} & \gate{\mathrm{H}} & \ctrl{1} & \gate{\mathrm{H}} & \rstick{$\ket{0}$}&  \setwiretype{n} \\
        \lstick{$\ket{\psi}$} & & & \qw & \gate{U} & \qw & \qw & \qw & \gate{U} & \qw & \qw & & 
    \end{quantikz}
    \caption{A circuit for rotations generated by involutory operator $U$, adapted from \cite{KFHM22}. Wires which begin and end inside the box are ancillae.}
    \label{fig:involRotationCirc}
\end{figure}
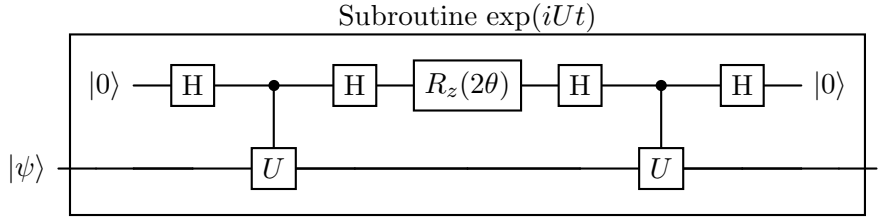
Further, note that this incurs some additional control wires on the ancilla qubit. Note that this does not affect our depth requirement, as we can fan out $O(F)$ copies of the flag such that whenever parallel controlled operations are required there is an available copy. Every algorithm in Section \ref{sec:fewerQubits} and Section \ref{sec:fullEnc} requires at most $O(F)$ many parallel operations to be controlled (note that uncomputed operations, such as comparisons, generically do not need to be controlled), so the number of ancillas remains $O(F)$.

\subsection{Coherently controlled rotations}
\label{app:coherentlyControlled}
In all of the encodings described in this paper, we assumed that the fermion operator is known classically and design a fixed circuit for that fermion operator. However, in Section \ref{subsec:varFewQubLowDepth} and prior, this dependence only enters in terms of comparisons against the binary number representing either (a) the position $p$ of an $\bflip(p, \bb)$ operation or (b) the end-position $p$ of a $\srank(p, \bb)$. It follows that these comparisons could be done against a quantum register using the approach outlined in Lemma \ref{lem:comparisonComplexity} and Lemma \ref{lem:equalityComplexity}. Further, this will not affect the depth or gate complexity of these circuits, and using the fan-out trick described in Lemma \ref{lem:fewQubLowDepEffZp} can be done with $O(F)$ ancilla qubits.

\subsection{Fermion overflows}
\label{app:fermionOverflow}
We note that the limits $F+k$ and $F-k$ on the number of fermions should not be exceeded at any point in the computation; if it is then the circuit behaviour is undefined. To ensure this, to simulate a system with $F$ many conserved fermions and a $k$-local Hamiltonian we can set capacity $F-k$ through $F+k$. Then, any $k$-local set of Majorana operators acts on a state that begins with $F$ many fermions and adds or removes at most $k$, which does not leave the capacity. If we perform the Majorana operators corresponding to a particle preserving rotation, then because these are guaranteed to commute, after all are performed the state will be left with $F$ many fermions. Further, for a $k$-local particle preserving rotation each product of $k$ Majorana operators acts on the same $k$ sites, so their product cannot change the full occupation by more than $k$. Hence simply a capacity $F-k$ through $F+k$ suffices.

If we consider the behaviour of the \textsc{SELECT} oracle described in Section \ref{sec:applications}, note that after many applications of the oracle the state may include contributions from garbage states beyond the $F$ fermion limit. However, so long as $F=F'+k$ (where again $k$ is the Hamiltonian locality and $F'$ is the conserved number of fermions) the \textsc{PREPARE}, \textsc{SELECT}, $\textsc{PREPARE}^\dagger$ protocol used in linear combination of unitary methods still block encodes the Hamiltonian. The undefined behaviour outside the block encoding does not matter for Hamiltonian simulation.

Another noteworthy point is the extra space that setting a larger cap incurs. Whenever $F > \log M$, the additional space used is $O(\log M)$. Furthermore, when $F=\Theta(M)$ the additional space is an additive $O(1)$. We will therefore ignore this point, though it is worth noting that it is an implicit assumption in our efficiency guarantees.

In a similar vein, encoding all Fock states of Hamming weight $\leq F$ (rather than exactly equal to $F$) again incurs some space loss. However, for $F=(0.5-\epsilon)M$ for any constant $\epsilon > 0$ this additional space usage is again an additive $O(1)$, and so we will ignore it.

\section{Circuit Subroutines}
\label{app:circSubroutines}

This section contains a list of circuit subroutines used, their description, and outlines how they could be implemented with the desired efficiency.

\begin{figure}[H]
    \centering
    \begin{quantikz}[wire types={q,q,q,q,n,q,q}, classical gap=0.7mm]
        \lstick{$\ket{0}$} & & & \gate[7][1.5cm]{I_p} \gateinput{$0$}\gateoutput{$f_{del}$} & \ctrl{2} & & \octrl{2} & & \ctrl{6} & \octrl{6} & &  \gate[7][1.5cm]{I_p^\dagger} \gateinput{$f_{del}$}\gateoutput{$1$} & \\
        \lstick{$\ket{0}$} & \qwbundle{\log F} & &{I_p}\gateinput{$0$}\gateoutput{$t$} & \control{} & & \control{} & \gate{U_{+p}} & \control{} & \control{} & \gate{U_{-p}} & \gateinput{$t$}\gateoutput{$0$}& \\
        \lstick{$\ket{l_1}$} & \qwbundle{\log M/F} & & & \gate[4]{L} & & \gate[4]{L^{\dagger}} & & & & &  & \\
        \lstick{$\ket{l_2}$} & \qwbundle{\log M/F} & & & & & & & & & & & \\
        \lstick{$\ldots$} & & & & & & & & & &\\
        \lstick{$\ket{l_F}$} & \qwbundle{\log M/F} & & & & \gate{\infty_l \leftrightarrow p_l}& & & & & & & \\
        \lstick{$\ket{m}$} & \qwbundle{O(F)} & & & & & & & \gate{L} & \gate{L^\dagger} & & & 
    \end{quantikz}
    \caption{The full few-qubit circuit for a $\bflip(p, \bb)$ operation, using initialization subroutine $I$ (Figure \ref{fig:FewQubitInitX}) and list cycling subroutine $L$ (Figure \ref{fig:MSBSwapladder}). Where subroutine $L$ is conditioned on the register $t$, $t$ is taken to be the input $i$ to the subroutine; when controlled on the bit $f_{del}$ this is a standard controlled operation. For a $\log M$ bit number $n$, let $n_m$ be the most significant bits and $n_l$ the least significant.}
    \label{fig:FewQubitX}
\end{figure}

\begin{figure}
    \centering
    \begin{quantikz}[wire types={q,q,q,n,q,q,q}, classical gap=0.7mm, column sep=0.35cm]
        \lstick{$\ket{0}$} & \qw & & & \slice[style=black]{(1): Compute start and end} & \targ{} & \gate{Z} & \targ{} & \targ{} & \gate{Z} & \targ{} \slice[style=black]{(2-3):Apply phases} & & & \\
        \lstick{$\ket{l_0}$} & \qwbundle{\log M/F} & & & & \gate{C_{\leq p}^{(0)'}}\wire[u][1]{q} & & \gate{C_{\leq p}^{(0)'}}\wire[u][1]{q} & & & & & & \\
        \lstick{$\ket{l_1}$} & \qwbundle{\log M/F} & & & & & & & \gate{C_{\leq p}^{(1)'}}\wire[u][2]{q} & & \gate{C_{\leq p}^{(1)'}}\wire[u][2]{q} & & & \\
        \lstick{$\ldots$} & & & & & & & & & \ddots & & & & & & \\
        \lstick{$\ket{m}$} & \qwbundle{O(F)} & & \gate{\#_{p_m}}\wire[d][1]{q} & \gate{\#_{p_m+1}}\wire[d][2]{q} & & & & & & & \gate{\#_{p_m+1}}\wire[d][2]{q} & \gate{\#_{p_m}}\wire[d][1]{q} & \\
        \lstick{$\ket{0}$} & \qwbundle{\log F} & & \targ{} & & \control{} & & \control{} & \control{} & & \control{} & & \targ{} & \\
        \lstick{$\ket{0}$} & \qwbundle{\log F} & & & \targ{} & \ctrl{-5} & & \ctrl{-5} & \ctrl{-5} & & \ctrl{-5} & \targ{} & & 
    \end{quantikz}
    \caption{The full few-qubit circuit for a $\srank(p, \bb)$ operation. Denote the most-significant bits of $p$ as $p_m$ (comparisons against $p_m$ interpret $p_m$ as a $G$ bit number), and the least-significant bits as $p_l$. Registers $l_0, l_1,...$ from the least significant bit array. Register $m$ contains the most-significant bit data.}
    \label{fig:FewQubitZZZ}
\end{figure}

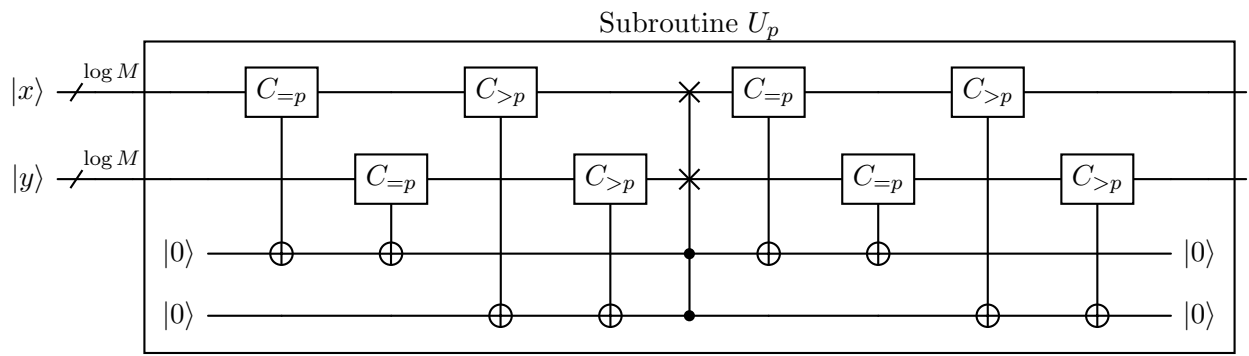
\begin{figure}
    \centering
    \begin{quantikz}
        \lstick{$\ket{x}$} & \qwbundle{\log M} &  & \gategroup[4,steps=13,style={inner sep=6pt}]{Subroutine $U_p$} & & \gate{C_{=p}} \wire[d][2]{q} & & \gate{C_{>p}} \wire[d][3]{q} & & \targX{}\wire[d][3]{q} & \gate{C_{=p}} \wire[d][2]{q} & & \gate{C_{>p}} \wire[d][3]{q} & & & & \\
        \lstick{$\ket{y}$} & \qwbundle{\log M} & & & & & \gate{C_{=p}} \wire[d][1]{q} & & \gate{C_{>p}} \wire[d][2]{q} & \targX{} & & \gate{C_{=p}} \wire[d][1]{q} & & \gate{C_{>p}} \wire[d][2]{q} & & & \\
        \setwiretype{n} & & & & \lstick{$\ket{0}$} &\targ{}\setwiretype{q} & \targ{} & & & \control{} & \targ{} & \targ{} & & & \rstick{$\ket{0}$} & \setwiretype{n} \\
        \setwiretype{n} & & & & \lstick{$\ket{0}$} & \setwiretype{q} & & \targ{} & \targ{} & \control{} & & & \targ{} & \targ{} & \rstick{$\ket{0}$} & \setwiretype{n}
    \end{quantikz}
    \caption{Subroutine $U_p$ swaps registers $x, y$ if one is equal to $p$ and the other is greater than $p$, using two ancillas. Note that $U_p = U_p^{\dag}$.}
    \label{fig:UpSubroutineForBubble}
\end{figure}

\begin{figure}
    \begin{center}
        \begin{quantikz}[wire types={q,n}, classical gap=0.7mm]
            \lstick{$\ket{x}$} & \qwbundle{\log M} & & & \gate{X^{a}}\gategroup[2,steps=9,style={inner sep=6pt}]{Subroutine $a \leftrightarrow b$} & \ocontrol{1}\wire[d][1]{q} & \gate{X^{a+b}} & \ocontrol{1}\wire[d][1]{q} & \gate{X^{a+b}} & \ocontrol{1}\wire[d][1]{q} & \gate{X^{a+b}} & \ocontrol{1}\wire[d][1]{q} & \gate{X^{a}} & \\
            & & & & \lstick{$\ket{0}$} & \targ{}\setwiretype{q} & & \targ{} & \control{}\wire[u][1]{q} & \targ{} & & \targ{} & \rstick{$\ket{0}$} & \setwiretype{n} &
        \end{quantikz}
    \end{center}
    \caption{A circuit to exchange two arbitrary classical basis states $\ket{a}$ and $\ket{b}$, where $X^s$ denotes an $X$ on any bit where the corresponding bit in $s$ is $1$. This circuit is denoted as $\boxed{a\leftrightarrow b}$.}
    \label{fig:exchangeClassicalsCircuit}
\end{figure}
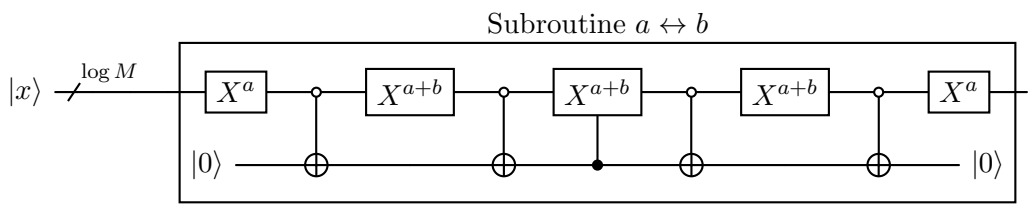

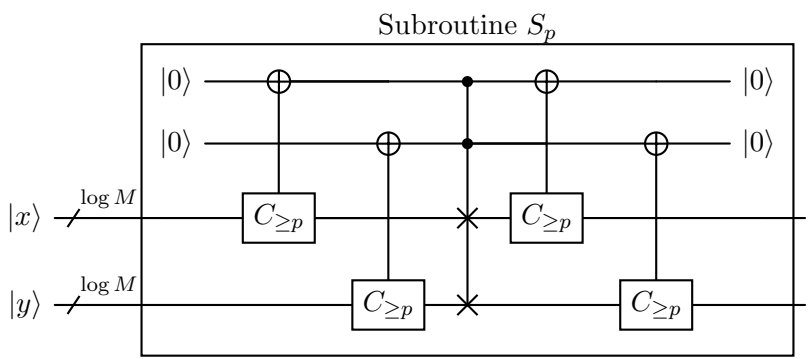
\begin{figure}
    \centering
    \begin{quantikz}[wire types={n,n,q,q}, classical gap=0.7mm]
        & & &\gategroup[4,steps=9,style={inner sep=6pt}]{Subroutine $S_p$} & \lstick{$\ket{0}$} & \targ{}\setwiretype{q} & \qw & \ctrl{2} & \targ{} & & \rstick{$\ket{0}$} & \setwiretype{n}& \\
        & & & & \lstick{$\ket{0}$} & \setwiretype{q} & \targ{} & \control{} &\qw  & \targ{} & \rstick{$\ket{0}$}& \setwiretype{n} & \\
        \lstick{$\ket{x}$} & \qwbundle{\log M} & & & & \gate{C_{\geq p}}\wire[u][2]{q} & \qw & \swap{1} & \gate{C_{\geq p}}\wire[u][2]{q} & & & & \\
        \lstick{$\ket{y}$} & \qwbundle{\log M} & & & & & \gate{C_{\geq p}}\wire[u][2]{q} & \targX{} & & \gate{C_{\geq p}}\wire[u][2]{q} & & &
    \end{quantikz}
    \caption{Subroutine $S_p$ for low depth. Swaps two registers $x, y$ if they are both greater than or equal to $p$.}
    \label{fig:swapIfGreaterP}
\end{figure}

\begin{figure}
    \centering
    \begin{quantikz}[wire types={n,n,q,q,q}, classical gap=0.7mm]
        & & & \gategroup[5,steps=9,style={inner sep=6pt}]{Subroutine $E_p$} & \lstick{$\ket{0}$} & \targ{}\setwiretype{q} & \qw & \ctrl{3}& \qw & \targ{} & \rstick{$\ket{0}$} & \setwiretype{n} \\
        & & & & \lstick{$\ket{0}$} & \setwiretype{q} & \targ{} & \control{} & \targ{} & & \rstick{$\ket{0}$} & \setwiretype{n} \\
        \lstick{$\ket{x}$} & \qwbundle{\log M} & & & & \gate{C_{<p}}\wire[u][2]{q}& \qw & \qw & \qw & \gate{C_{<p}}\wire[u][2]{q} & \qw & & \\
        \lstick{$\ket{y}$} & \qwbundle{\log M} & & & & \qw & \qw& \gate{p \leftrightarrow \infty} & \qw & \qw & \qw & & \\
        \lstick{$\ket{z}$} & \qwbundle{\log M} & & & & \qw & \gate{C_{>p}}\wire[u][3]{q} & \qw  & \gate{C_{>p}}\wire[u][3]{q} & & & &
    \end{quantikz}
    \caption{Subroutine $E_p$ for low depth. If the first register $x$ is $x<p$ and the third register $z$ is $z>p$, it exchanges the middle register $y: p \leftrightarrow \infty$.}
    \label{fig:delPLocal}
\end{figure}
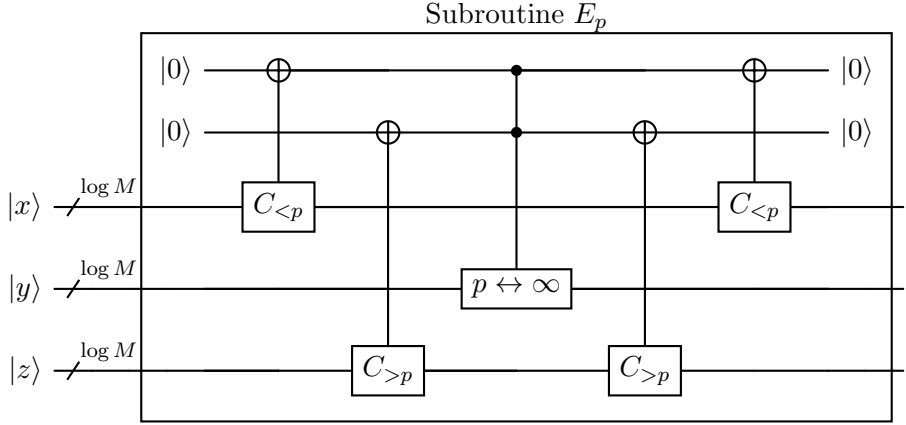

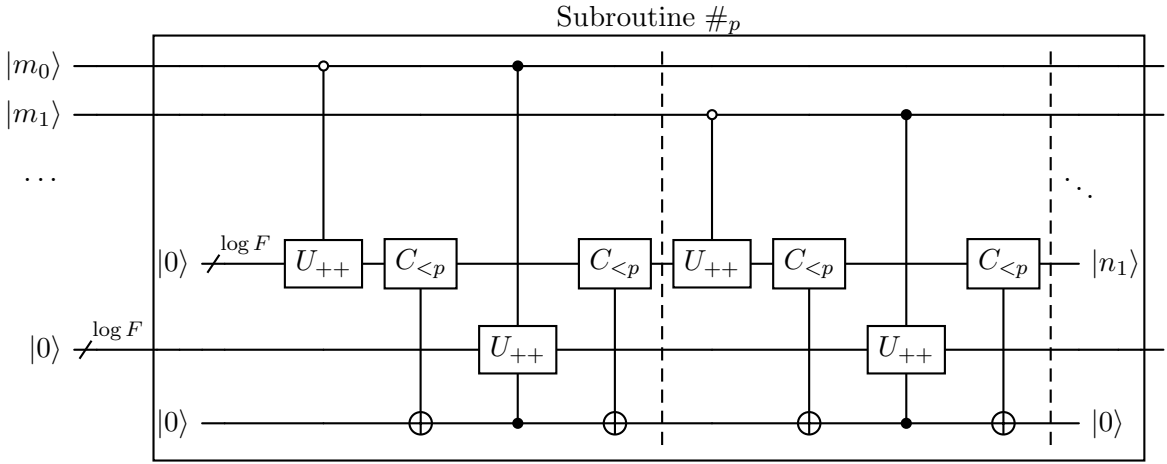
\begin{figure}
    \centering
    \begin{quantikz}[wire types={q,q,n,n,q,n}, classical gap=0.7mm, column sep=0.3cm]
        \lstick{$\ket{m_0}$} & &[0.5cm] &\gategroup[6,steps=13,style={inner sep=6pt}]{Subroutine $\#_p$} & & &[0.5cm] \octrl{3} & & \ctrl{4} & \slice[style=black]{}& & & & \slice[style=black]{} & & & & \\
        \lstick{$\ket{m_1}$} & & & & & & & & & & \octrl{2} & & \ctrl{3} & & & & & \\
        \lstick{\ldots} & & & & & & & & & & & & & & \ddots & \\
        & & && \lstick{$\ket{0}$} & \qwbundle{\log F}\setwiretype{q} & \gate{U_{++}} & \gate{C_{<p}}\wire[d][2]{q} & & \gate{C_{<p}}\wire[d][2]{q} & \gate{U_{++}} & \gate{C_{<p}}\wire[d][2]{q} & & \gate{C_{<p}}\wire[d][2]{q} & \rstick{$\ket{n_1}$} & \setwiretype{n} \\
        \lstick{$\ket{0}$} & \qwbundle{\log F} & & & & & & & \gate{U_{++}} & & & & \gate{U_{++}} & & & & & \\
        & & & & \lstick{$\ket{0}$} & \setwiretype{q}& & \targ{} & \ctrl{-1} & \targ{} & & \targ{} & \ctrl{-1} & \targ{} & \rstick{$\ket{0}$} & \setwiretype{n} 
    \end{quantikz}
    \caption{Subroutine $\#_p$ for counting the number of $1$'s before the $p$-th $0$, where $U_{++}$ increments the register by $1$ and dashed lines separate the steps for each bit. The ancilla register counts the number of $0$'s in the whole list, so it will always end in state $\ket{n_1}$ for $n_1=2^{\lceil \log F\rceil}-1$ (the number of bins minus one) and can therefore be uncomputed.}
    \label{fig:MSBScanningSubroutine}
\end{figure}

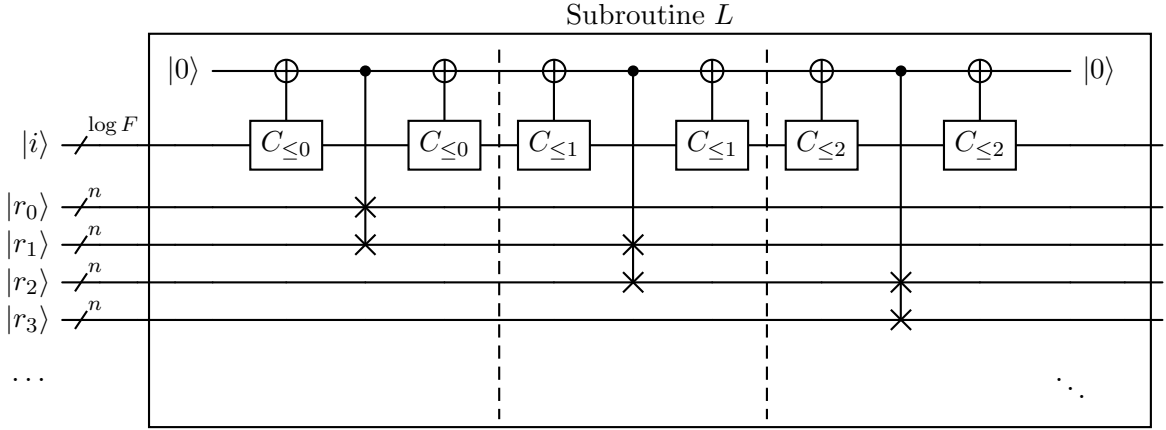
\begin{figure}
    \centering
    \begin{quantikz}[wire types={n,q,q,q,q,q,n}, classical gap=0.7mm]
        & & & \gategroup[7,steps=13,style={inner sep=6pt}]{Subroutine $L$} & \lstick{$\ket{0}$} & \targ{}\setwiretype{q} & \ctrl{2} & \targ{} \slice[style=black]{}& \targ{} & \ctrl{3} & \targ{}\slice[style=black]{} & \targ{} & \ctrl{4} & \targ{} & \rstick{$\ket{0}$} & \setwiretype{n} & \\
        \lstick{$\ket{i}$} & \qwbundle{\log F} & & & & \gate{C_{\leq 0}}\wire[u][1]{q} & & \gate{C_{\leq 0}}\wire[u][1]{q}& \gate{C_{\leq 1}}\wire[u][1]{q} & & \gate{C_{\leq 1}}\wire[u][1]{q} & \gate{C_{\leq 2}}\wire[u][1]{q} & & \gate{C_{\leq 2}}\wire[u][1]{q} & & & \\
        \lstick{$\ket{r_0}$} & \qwbundle{n} & & & & & \swap{1} & & & & & & & & & & \\
        \lstick{$\ket{r_{1}}$} & \qwbundle{n} & & & & & \targX{} & & & \swap{1} & & & & & & & \\
        \lstick{$\ket{r_{2}}$} & \qwbundle{n} & & & & & & & & \targX{} & & & \swap{1} & & & & \\
        \lstick{$\ket{r_{3}}$} & \qwbundle{n} & & & & & & & & & & & \targX{} & & & & \\
        \lstick{\ldots} & & & & & & & & & & & & & & \ddots & & 
    \end{quantikz}
    \caption{Subroutine $L$ for shuffling the $i$-th element ($i$ given in the top register) of a list of registers of arbitrary size to the end (or when inverted, shuffle the last element to the $p$-th position), otherwise maintaining order.}
    \label{fig:MSBSwapladder}
\end{figure}

\begin{figure}
    \centering
    \begin{quantikz}[wire types={q,q,q,n,n,n,n}, classical gap=0.7mm]
        \lstick{$\ket{r_i}$} & \qwbundle{\log M/F} &[0.4cm] &\gategroup[7,steps=12,style={inner sep=6pt}]{Subroutine $C_{\leq p}^{(i)'}$} & & \gate{C_{\leq p_l}}\wire[d][3]{q} & & & & & & & \gate{C_{\leq p_l}}\wire[d][3]{q} & & & \\
        \lstick{$\ket{p_{start}}$} & \qwbundle{\log F} & & & & & \gate{C_{\geq i}}\wire[d][3]{q} & & & & & \gate{C_{\geq i}}\wire[d][3]{q} & & & & \\
        \lstick{$\ket{p_{end}}$} & \qwbundle{\log F} & & & & & & \gate{C_{<i}}\wire[d][3]{q} & & & \gate{C_{<i}}\wire[d][3]{q} & & & & & \\
        & & & & \lstick{$\ket{0}$} & \targ{}\setwiretype{q} & & & \ctrl{3} & & & & \targ{} & \rstick{$\ket{0}$} & \setwiretype{n} & \\
        & & & & \lstick{$\ket{0}$} & \setwiretype{q} & \targ{} & & \control{}& \octrl{2} & & \targ{} & & \rstick{$\ket{0}$} & \setwiretype{n} & \\
        & & & & \lstick{$\ket{0}$} & \setwiretype{q} & & \targ{} & \control{}& & \targ{} & & & \rstick{$\ket{0}$} & \setwiretype{n} & \\
        \lstick{$\ket{0}$} & \setwiretype{q} & & & & & & & \targ{} & \targ{} & & & & & & 
    \end{quantikz}
    \caption{Subroutine $C_{\leq p}^{(i)'}$ for comparing the $i$-th (in $0$ based indexing) register to $p$, given the least-significant bits as well as a start/end range of registers matching the most-significant bits of $p$. $p_l$ refers to the least significant bits of $p$. Note that $i$ is a fixed register, so the circuit can depend on $i$ even inside circuits that are coherently controlled. The result is placed into the last wire. Other logical comparisons in this format can be done similarly.}
    \label{fig:LSBComparisons}
\end{figure}
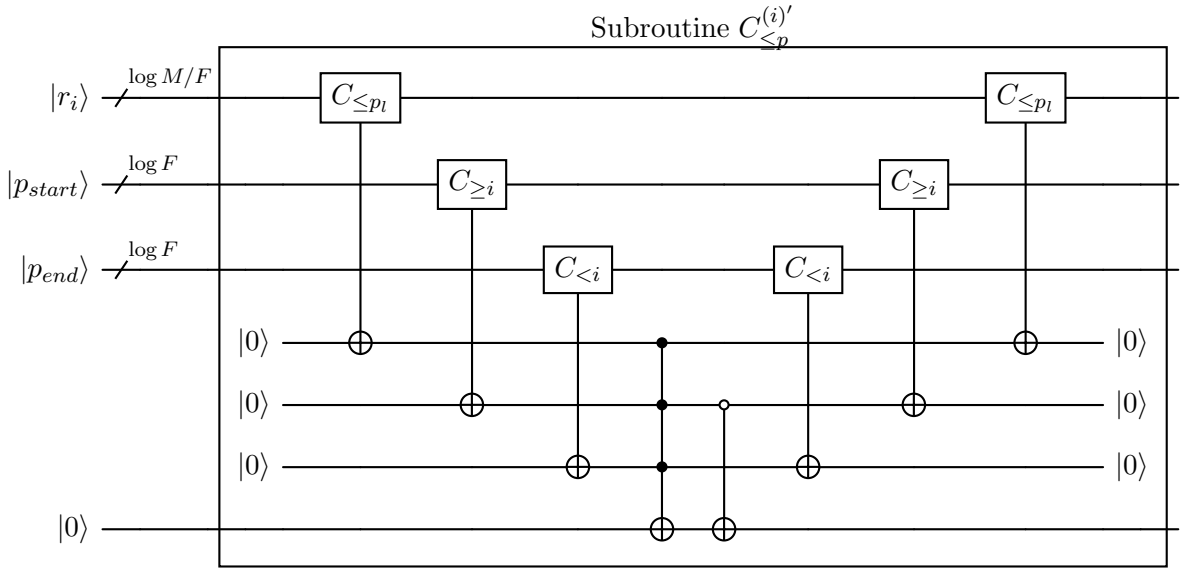

\begin{figure}
    \centering
    \begin{quantikz}[wire types={q,q,q,q,q,n,q,n,n}, classical gap=0.7mm, column sep=0.3cm]
        \lstick{$\ket{0}$} & &[1.3cm] \gategroup[9,steps=14,style={inner sep=6pt}]{Subroutine $I_p$} & & & & \slice[style=black]{} & \targ{} \slice[style=black]{}& & & & \slice[style=black]{}& & & & & & \\
        \lstick{$\ket{0}$} & \qwbundle{\log F} & & & & & & & \gate{X^{(0)}} & \gate{X^{(1)}} & \gate{X^{(2)}} & & & & & & & \\
        \lstick{$\ket{l_0}$} & \qwbundle{\log M/F} & & & & & & \gate{C_{=p}^{(0)'}}\wire[u][2]{q} & \gate{C_{\geq p}^{(0)'}}\wire[u][1]{q} & \gate{C_{<p}^{(0)'}}\wire[u][1]{q} & & & & & & & & \\
        \lstick{$\ket{l_1}$} & \qwbundle{\log M/F} & & & & & & \gate{C_{=p}^{(1)'}}\wire[u][1]{q} & &\gate{C_{\geq p}^{(1)'}}\wire[u][1]{q} & \gate{C_{< p}^{(1)'}}\wire[u][2]{q} & & & & & & & \\
        \lstick{$\ket{l_2}$} & \qwbundle{\log M/F} & & & & & & \gate{C_{=p}^{(2)'}}\wire[u][1]{q} & & & \gate{C_{\geq p}^{(2)'}}\wire[u][1]{q}& & & & & & & \\
        \lstick{$\ldots$} & & & & & & & & & & &\ddots & & & & &  & & &  \\
        \lstick{$\ket{m}$} & \qwbundle{O(F)} & & & &  \gate{\#_{p_m}}\wire[d][1]{q} & \gate{\#_{p_m+1}}\wire[d][2]{q} & & & & & & \gate{\#_{p_m+1}}\wire[d][2]{q} & \gate{\#_{p_m}}\wire[d][1]{q} & & & & \\
        & & & \lstick{$\ket{0}$} & \qwbundle{\log F} \setwiretype{q}& \targ{} & & \control{} &\control{} &\control{} &\control{} & & & \targ{} & \rstick{$\ket{0}$} &\setwiretype{n} & & \\
        & & & \lstick{$\ket{0}$} & \qwbundle{\log F}\setwiretype{q} & & \targ{} & \ctrl{-4} & \ctrl{-6} &\ctrl{-5} &\ctrl{-4} & & \targ{} & &\rstick{$\ket{0}$} & \setwiretype{n} & & & 
    \end{quantikz}
    \caption{Subroutine $I_p$ for initializing the $f_{del}$ flag (top wire) and target (second from top wire) used in the few qubit encoding of an $X_p$. Denote the most-significant bits of $p$ as $p_m$ (comparisons against $p_m$ interpret $p_m$ as a $G$ bit number), and the least-significant bits as $p_l$. }
    \label{fig:FewQubitInitX}
\end{figure}
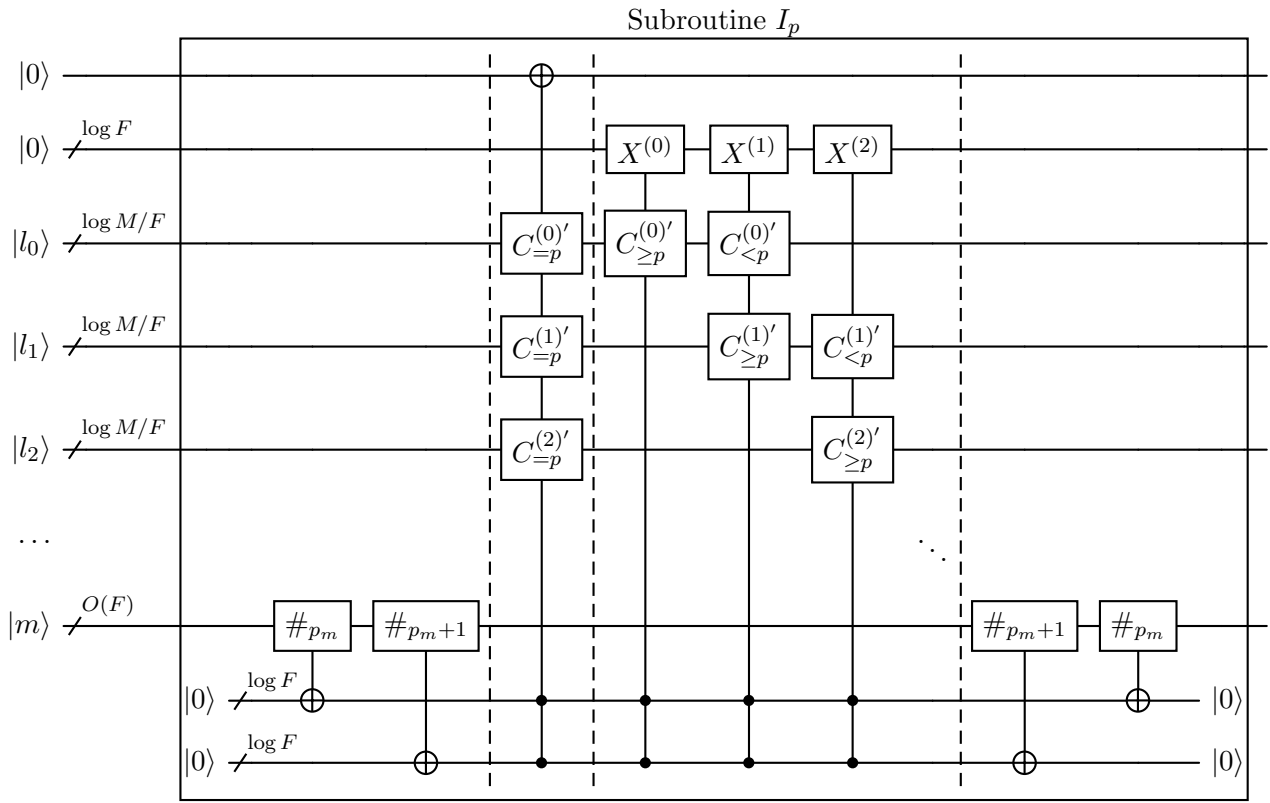

\begin{figure}
    \centering 
    \begin{quantikz}[classical gap=0.7mm, column sep=0.3cm, row sep=0.7cm]
        \lstick{$\ket{o_{in}}$} & \qwbundle{\log (b-a)} & & [1cm]\gategroup[7,steps=9,style={inner sep=6pt}]{Subroutine $T_p^{(d)}$} & & & & & \gate[2]{C_{S_{a,b} < o_{in}}}\wire[d][6]{q}& & \ctrl{4} & \gate[2]{o_{in}-S_{a,b}}\wire[d][6]{q} & & \rstick{$\ket{o_{in}}$} \\
        \lstick{$\ket{S_{a, b}}$} & \qwbundle{\log (b-a)} & & & & & & & & & & & & \rstick{$\ket{S_{a,b}}$} \\
        \lstick{$\ket{f_{in}}$} & & & & & & & & \control{} & \ctrl{4} & & & & \rstick{$\ket{f_{in}}$} \\
        \setwiretype{n} & & & & \lstick{$\ket{0}$} & \qwbundle{\log \left(\frac{b-a}{2}\right)}\setwiretype{q} & & & & & \targ{} & & & \rstick{$\ket{o_l}$} \\
        \setwiretype{n} & & & & \lstick{$\ket{0}$} & \setwiretype{q} & & & & \targ{} & \control{} & & & \rstick{$\ket{f_l}$} \\
        \setwiretype{n} & & & & \lstick{$\ket{0}$} & \qwbundle{\log \left(\frac{b-a}{2}\right)}\setwiretype{q} & & & & & & \targ{} & & \rstick{$\ket{o_r}$} \\
        \setwiretype{n} & & & & \lstick{$\ket{0}$} & \setwiretype{q} & & & \targ{} & \ocontrol{} & & \control{} & & \rstick{$\ket{f_r}$}
    \end{quantikz}
    \caption{Subroutine $T_p^{(d)}$ for walking down the succinct tree to find the $p$-th $1$ from the left. Marks the leaf whose sub-array has the $p$-th $1$, as well as the $p$-th $1$'s rank (i.e. number of preceding $1$'s) within said sub-array. Gates controlled on a bit and a register and targeting a register will copy the controlled register into the target if the control bit is on. When there is a register size mismatch, only the least significant bits are placed.}
    \label{fig:succinctTreeCountWalkdown}
\end{figure}
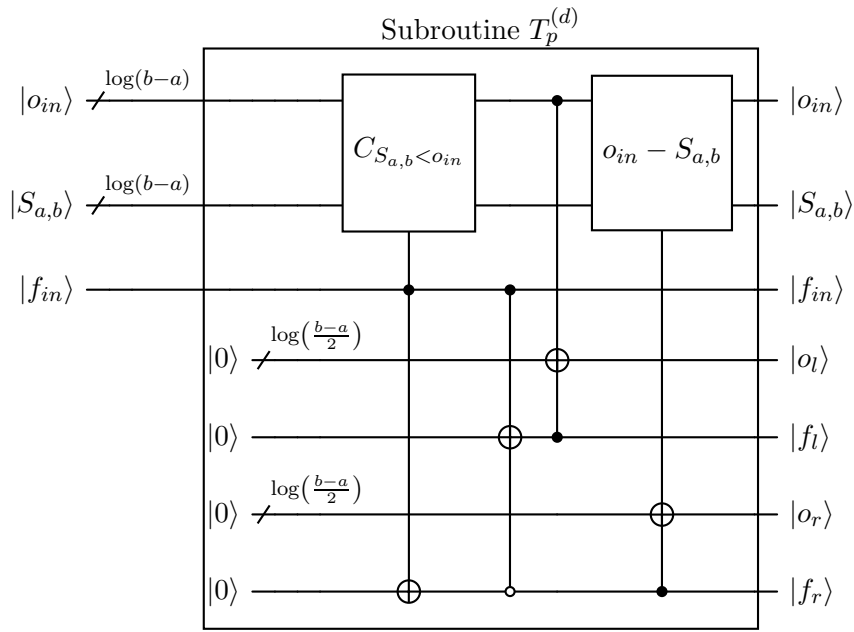

\begin{figure}
    \centering
     
    \begin{quantikz}[classical gap=0.7mm, column sep=0.3cm, row sep=0.7cm]
        \setwiretype{n} & & & \gategroup[10,steps=12,style={inner sep=6pt}]{Subroutine $T_p^{(u)}$} & \lstick{$\ket{0}$} & \qwbundle{\log \left(b-a\right)}\setwiretype{q} & & & &\swap{2} & \gate{U_{+(b-a)}} & \swap{1}& & & & & \rstick{$\ket{i}$} \\
        \lstick{$\ket{i_l}$} & \qwbundle{\log (b-a)} & \hphantomgate{vwide} & & & & & & & & & \targX{}& & \rstick{$\ket{0}$} & \setwiretype{n} &  \\
        \lstick{$\ket{i_r}$} & \qwbundle{\log (b-a)} & & & & & & & & \targX{} & & & & \rstick{$\ket{0}$} & \setwiretype{n} &  \\
        \lstick{$\ket{o_{in}}$} & \qwbundle{\log (b-a)} & & & & & & & & & & & \gate[7]{(T_p^{(d)})^\dagger} & & & & \rstick{$\ket{o}$} \\
        \lstick{$\ket{S_{a, b}}$} & \qwbundle{\log (b-a)} & & & & & & & & & & & & & & & \rstick{$\ket{S_{a,b}}$} \\
        \lstick{$\ket{f_{in}}$} & & & & & & & & & & & & & & & & \rstick{$\ket{f_{in}}$} \\
        \lstick{$\ket{o_l}$} & \qwbundle{\log \left(\frac{b-a}{2}\right)} & & & & & & & & & & & & \rstick{$\ket{0}$} & \setwiretype{n} &  \\
        \lstick{$\ket{f_l}$} & & & & & & & & & & & \ctrl{-6} & & \rstick{$\ket{0}$} \\
        \lstick{$\ket{o_r}$} & \qwbundle{\log \left(\frac{b-a}{2}\right)}\setwiretype{q} & & & & & & & & & & & & \rstick{$\ket{0}$} & \setwiretype{n} &  \\
        \lstick{$\ket{f_r}$} & \setwiretype{q} & & & & & & & & \ctrl{-7} & \ctrl{-9} & & & \rstick{$\ket{0}$} & \setwiretype{n} & 
    \end{quantikz}
    \caption{Subroutine $T_p^{(u)}$ for walking up the succinct tree recursively to find the $p$-th $1$ from the left. Combines the position $i$ from the left $(i_l)$ or right $(i_r)$ subtree with the fixed offset to determine position in the union of the two subtrees, and uncomputes intermediate information. Swaps between registers of different sizes swap only the first bits of the larger register.}
    \label{fig:succinctTreeCountWalkup}
\end{figure}
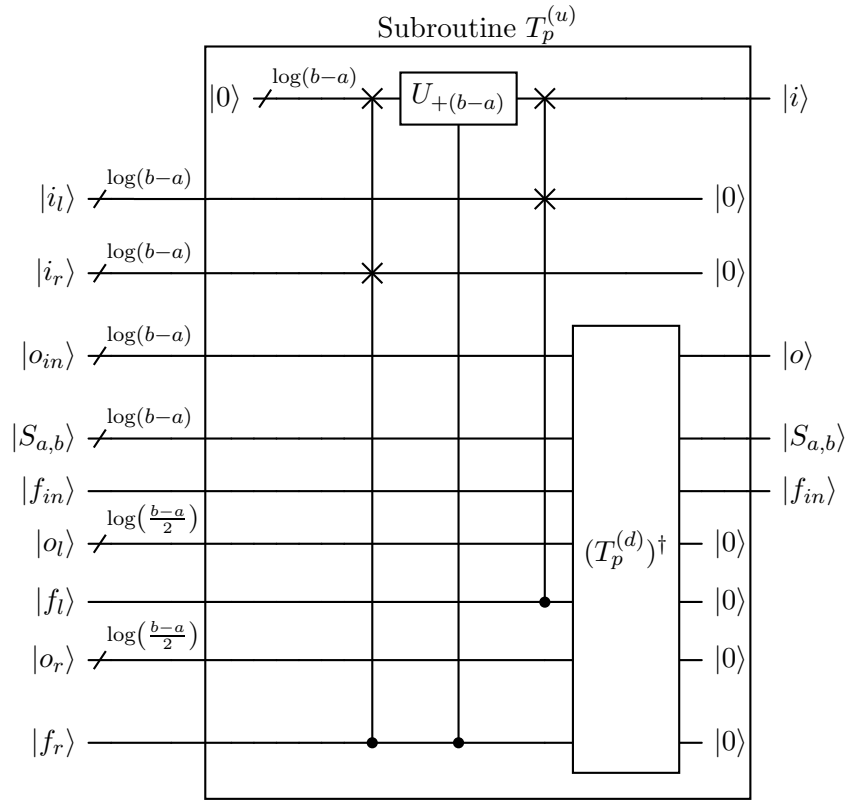

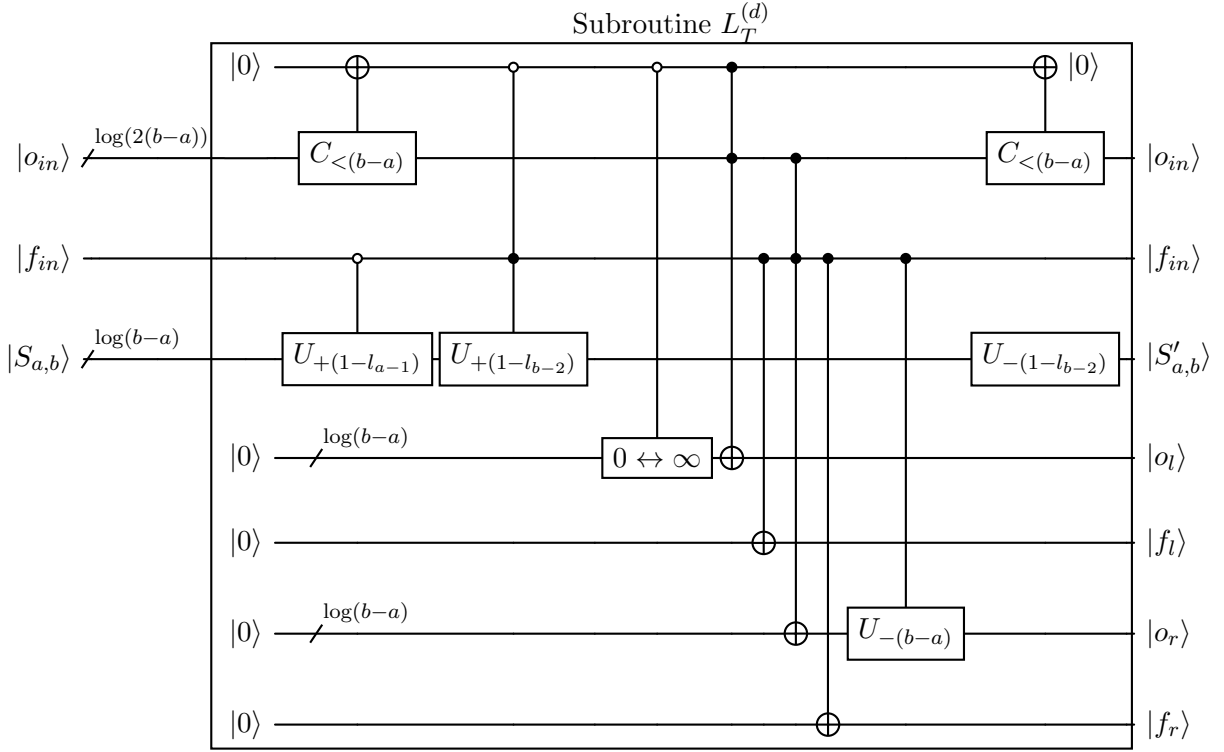
\begin{figure}
    \centering
     
    \begin{quantikz}[classical gap=0.7mm, column sep=0.1cm, row sep=0.7cm]
        \setwiretype{n} & & & \gategroup[8,steps=12,style={inner sep=1pt}]{Subroutine $L_T^{(d)}$} & \lstick{$\ket{0}$} & \targ{} \setwiretype{q} & \octrl{3} & & \octrl{4}& \ctrl{1} & & & & & \targ{} \rstick{$\ket{0}$}& \setwiretype{n} & \\
        \lstick{$\ket{o_{in}}$} & \qwbundle{\log (2(b-a))} & & \hphantomgate{v} & & \gate{C_{<(b-a)}}\wire[u][1]{q} & & & & \ctrl{3} & & \ctrl{5} & & & \gate{C_{<(b-a)}}\wire[u][1]{q} & & \rstick{$\ket{o_{in}}$} \\
        \lstick{$\ket{f_{in}}$} & &\hphantomgate{verywide} & & & \octrl{1} & \control{} & & & & \ctrl{3} & \control{} & \ctrl{5} & \ctrl{4} & & & \rstick{$\ket{f_{in}}$} \\
        \lstick{$\ket{S_{a, b}}$} & \qwbundle{\log (b-a)} & & & & \gate{U_{+(1-l_{a-1})}} & \gate{U_{+(1-l_{b-2})}} & & & & & & & & \gate{U_{-(1-l_{b-2})}} & & \rstick{$\ket{S_{a,b}'}$} \\
        \setwiretype{n} & & & & \lstick{$\ket{0}$} & \qwbundle{\log \left(b-a\right)}\setwiretype{q} & & & \gate{0 \leftrightarrow \infty} & \targ{} & & & & & & & \rstick{$\ket{o_l}$} \\
        \setwiretype{n} & & & & \lstick{$\ket{0}$} & \setwiretype{q} & & & & & \targ{} & & & & & & \rstick{$\ket{f_l}$} \\
        \setwiretype{n} & & & & \lstick{$\ket{0}$} & \qwbundle{\log \left(b-a\right)}\setwiretype{q} & & & & & & \targ{} & & \gate{U_{-(b-a)}} & & & \rstick{$\ket{o_r}$} \\
        \setwiretype{n} & & & & \lstick{$\ket{0}$} & \setwiretype{q} & & & & & & & \targ{} & & & & \rstick{$\ket{f_r}$}
    \end{quantikz}
    \caption{Subroutine $L_T^{(d)}$ for walking down the succinct tree when called recursively to shift a prefix of the array downwards. $f_{in}$ denotes whether the offset lies to the right of position  $a$, and if $f_{in}=\top$ then $o_{in}$ is how far to the right of $a$ the rotated position is (or $\infty$ if it is beyond $b$). Note that a gate controlled on a register XOR's each bit of the register into the target; if the target is smaller then only the least significant bits are XORed. This subroutine adjusts sublist sums depending on the elements which will be shifted, and propagates the relevant information to left (left flag $f_l$ and left offset $o_l$) and right (right flag $f_r$ and right offset $o_r$) subtrees. Further, define $L_T^{(u)}$ to be the inverse of $L_T^{(d)}$, except with every gate touching $S_{a,b}$ removed. In this way, $L_T^{(u)}$ uncomputes the flag and offset used in $L_T^{(d)}$, while leaving the $S_{a,b}$ registers unaffected.}
    \label{fig:succinctTreeShiftWalkdown}
\end{figure}

\end{document}